\documentclass[authorversion, sigconf,screen]{acmart}
 \pdfoutput=1
 
\AtBeginDocument{}

\copyrightyear{2024} 
\acmYear{2024} 
\setcopyright{rightsretained} 
\acmConference[LICS '24]{39th Annual ACM/IEEE Symposium on Logic in Computer Science}{July 8--11, 2024}{Tallinn, Estonia}
\acmBooktitle{39th Annual ACM/IEEE Symposium on Logic in Computer Science (LICS '24), July 8--11, 2024, Tallinn, Estonia}\acmDOI{10.1145/3661814.3662130}
\acmISBN{979-8-4007-0660-8/24/07}

\usepackage{graphicx}
\usepackage[disable]{todonotes}
\usepackage[outline]{contour}
\usepackage{quiver}

\usepackage{cmll}
\usepackage{enumitem}
\usepackage{cleveref}
\usepackage{mathtools} 
\usepackage{enumitem}
\usepackage{tikz}
\usepackage{stmaryrd}

\usepackage{subfig}
\usepackage{wrapfig}
\usepackage{caption}
\usepackage{thm-restate}
\usepackage{makecell}

\usepackage{booktabs}

\tikzcdset{scale cd/.style={every label/.append style={scale=#1},
    cells={nodes={scale=#1}}}}
\tikzcdset{scalenodes/.style=
  {cells={nodes={scale=#1}}}}
\usetikzlibrary{calc,matrix,decorations.markings,decorations.pathreplacing,arrows,cd,positioning,shapes.misc}
\usetikzlibrary{decorations.pathmorphing}
\usetikzlibrary{backgrounds}
\tikzstyle{game-causality}=[dotted, thick]
\tikzstyle{strat-causality}=[->, thick, -{Triangle[open]}]
\tikzset{curve/.style={settings={#1},to path={(\tikztostart)
    .. controls ($(\tikztostart)!\pv{pos}!(\tikztotarget)!\pv{height}!270:(\tikztotarget)$)
    and ($(\tikztostart)!1-\pv{pos}!(\tikztotarget)!\pv{height}!270:(\tikztotarget)$)
    .. (\tikztotarget)\tikztonodes}},
    settings/.code={\tikzset{quiver/.cd,#1}
        \def\pv##1{\pgfkeysvalueof{/tikz/quiver/##1}}},
    quiver/.cd,pos/.initial=0.35,height/.initial=0}

\tikzset{
  prof/.style = {decoration = {markings, mark = at position 0.5 with { \node[transform shape, yscale=.4] {$|$}; } }, postaction = {decorate} },
}
\tikzstyle{neutralnode}=[fill=gray!25, draw, thick,  inner sep=1pt, minimum size=9pt]
\tikzstyle{posnode}= [fill=blue!25, draw, circle, thick,  inner sep=1pt, minimum size=9pt]
\tikzstyle{negnode}= [fill=red!25, draw, thick, circle, inner sep=1pt, minimum size=9pt]

\newcommand{\modification}[1]{\boldsymbol{#1}}
\newcommand{\x}{\modification{x}}
\newcommand{\y}{\modification{y}}
\newcommand{\z}{\modification{z}}
\newcommand{\w}{\modification{w}}
\newcommand{\n}{\modification{n}}
\newcommand{\m}{\modification{m}}
\newcommand{\p}{\modification{p}}
\renewcommand{\b}{\modification{b}}

\newcommand{\commute}{ \modification{c} }
\newcommand{\bistrong}{ \b}

\newcommand{\mongamma}{\gamma}
\newcommand{\mondelta}{\delta}
\newcommand{\monomega}{\varpi}

\newcommand{\proddec}[1]{ {#1}_2 }
\newcommand{\unitdec}[1]{  {#1}_0 }

\newcommand{\etaprod}{ \proddec{\eta} }
\newcommand{\etaunit}{ \unitdec{\eta} }
\newcommand{\muprod}{ \proddec{\mu} }
\newcommand{\muunit}{ \unitdec{\mu} }

\newtheorem{notation}{Notation}

\usepackage{listings}
\definecolor{darkgreen}{RGB}{0,102,0}
\definecolor{lightblue}{RGB}{111, 151, 242}
\definecolor{darkred}{RGB}{178,0,0}
\definecolor{lightgrey}{rgb}{0.5,0.5,0.5}
\definecolor{mymauvnne}{rgb}{0.58,0,0.82}
\definecolor{brightyellow}{RGB}{237, 217, 83}

\definecolor{orange}{RGB}{255,123,0}
\definecolor{lightpink}{RGB}{255, 103, 129}
\definecolor{brightpink}{RGB}{232, 88, 232}
\definecolor{grassgreen}{RGB}{0,154,23}
\definecolor{seablue}{RGB}{52,111,111}
\definecolor{skyblue}{RGB}{135, 206, 235}
\newcommand{\opacity}{0.3}

\newcommand{\alias}[1]{ |[alias=#1]| }

\newcommand{\pseudomonadcolour}						{seablue}
\newcommand{\monoidalcolour}							 {lightpink}
\newcommand{\strongpseudofuncolour}					{grassgreen}
\newcommand{\strongpseudomoncolour}				{brightyellow}
\newcommand{\altstrongpseudomoncolour}			{brightpink}
\newcommand{\bistrongcolour}								{lightblue}

\newcommand{\commutecolour}								{orange}

\newcommand{\monadmcolour}{\pseudomonadcolour}
\newcommand{\monadncolour}{\pseudomonadcolour}
\newcommand{\monadpcolour}{\pseudomonadcolour}

\newcommand{\monoidalpcolour}{\monoidalcolour}
\newcommand{\monoidalmcolour}{\monoidalcolour}
\newcommand{\monoidallcolour}{\monoidalcolour}
\newcommand{\monoidalrcolour}{\monoidalcolour}

\newcommand{\strongxcolour}{\strongpseudofuncolour}
\newcommand{\strongycolour}{\strongpseudofuncolour}
\newcommand{\strongzcolour}{\strongpseudomoncolour}
\newcommand{\strongwcolour}{\strongpseudomoncolour}

\newcommand{\strongzprimecolour}{\altstrongpseudomoncolour}
\newcommand{\strongwprimecolour}{\altstrongpseudomoncolour}

\newcommand{\actiondecorate}[2]{ {#1}^{#2} }

\newcommand{\acttri}{\act}
\newcommand{\alphatri}{ \actiondecorate{\alpha}{\acttri} } 
\newcommand{\lambdatri}{ \actiondecorate{\lambda}{\acttri} } 
\newcommand{\penttri}{ \actiondecorate{\pentagonator}{\acttri} } 
\newcommand{\ltri}{ \actiondecorate{\montrianglel}{\acttri} }
\newcommand{\mtri}{ \actiondecorate{\montrianglem}{\acttri} }

\newcommand{\actstar}{\star}
\newcommand{\alphastar}{ \actiondecorate{\alpha}{\actstar} } 
\newcommand{\lambdastar}{ \actiondecorate{\lambda}{\actstar} } 
\newcommand{\pentstar}{ \actiondecorate{\pentagonator}{\actstar} } 
\newcommand{\lstar}{ \actiondecorate{\montrianglel}{\actstar} }
\newcommand{\mstar}{ \actiondecorate{\montrianglem}{\actstar} }

\newcommand{\actmaptrans}{\chi}
\newcommand{\actmapassoccell}{\omega}
\newcommand{\actmapunitcell}{\gamma}

\newcommand{\acttranstrans}{\sigma}
\newcommand{\acttransmodif}{\Gamma}

\newcommand{\actthreecell}{ q }

\newcommand{\vertequals}{\rotatebox{90}{$\,=$}}
\newcommand{\para}[1]{\widetilde{#1}}
\DeclareMathOperator{\paratens}{\para{\tens}}

\newcommand{\homBicat}[2]{ \mathrm{Hom}({#1}, {#2}) }

\newcommand{\cf}{\emph{c.f.}}

\newcommand{\eg}{\emph{e.g.}}
\newcommand{\ie}{\emph{i.e.}}

\newcommand{\st}{\mid}

\newcommand{\To}{\ensuremath{\Rightarrow}} \newcommand{\xra}[1]{\ensuremath{\xrightarrow{#1}}}  \newcommand{\XRA}[1]{\ensuremath{\xRightarrow{#1}}}
\newcommand{\id}{\ensuremath{\mathrm{id}}} \newcommand{\Id}{\mathsf{Id}}

  \renewcommand{\exp}[2]{{#1} \To {#2}}

\renewcommand{\epsilon}{\varepsilon}

\newcommand{\psinv}[1]{{#1}^{{\bullet}}}  \newcommand{\iso}{\cong}
 
 \newcommand{\inr}{\ensuremath{\mathrm{inr}}} \newcommand{\inl}{\ensuremath{\mathrm{inl}}}    
    \newcommand{\seq}[1]{\langle #1 \rangle}

\newcommand{\Nat}{\mathbb{N}}

\newcommand{\act}{\triangleright}
\newcommand{\tens}{\otimes}
\newcommand{\tensu}{I}
\newcommand{\cellOf}[1]{{#1}}

\newcommand{\monoidal}[1]{ (#1, \tens, \tensu) }

\newcommand{\pentagonator}{\mathfrak{p}}
\newcommand{\montrianglel}{\mathfrak{l}}

\newcommand{\montrianglem}{\mathfrak{m}}

\newcommand{\actalpha}{{\widetilde\alpha}}
\newcommand{\actlambda}{{\widetilde\lambda}}
\newcommand{\actpentagonator}{{\widetilde\pentagonator}}
\newcommand{\acttrianglel}{{\widetilde\montrianglel}}
\newcommand{\acttrianglem}{{\widetilde\montrianglem}}

\newcommand{\Rel}{\mathbf{Rel}}

\newcommand{\catA}{\ensuremath{\mathbb{A}}} \newcommand{\catB}{\ensuremath{\mathbb{B}}} \newcommand{\catC}{\ensuremath{\mathbb{C}}} \newcommand{\catD}{\ensuremath{\mathbb{D}}}

\newcommand{\baseCat}{\B}
    \newcommand{\Set}{\mathbf{Set}}
\newcommand{\Cat}{\mathbf{Cat}}

\newcommand{\Bicat}{\mathbf{Bicat}}

\newcommand{\Para}{\mathbf{Para}}

\renewcommand{\a}{\mathsf{a}}
\renewcommand{\r}{\mathsf{r}}
\renewcommand{\l}{\mathsf{l}}

\newcommand{\braid}{\beta}

\newcommand{\ltie}{ \rtimes }
\newcommand{\rtie}{ \ltimes }

\DeclareMathOperator{\binder}{ . }
\newcommand{\bind}{\binder}

\DeclareMathOperator{\bulletop}{\bullet}
\renewcommand{\vert}{\bulletop}
\newcommand{\vertsub}[1]{\vert}

  \newcommand{\Span}{\mathbf{Span}}

\newcommand{\twocell}[1]{\overset{\small{#1}}{\Leftarrow}}

\newenvironment{bprooftree}
  {\leavevmode\hbox\bgroup}
  {\DisplayProof\egroup}

\makeatletter
\renewcommand{\@marginparreset}{\reset@font\footnotesize
  \footnote
  \raggedright
  \@setminipage
}
\makeatother

\newcommand{\oncell}[1]{  #1 }

\newcommand{\B}{\mathscr{B}}
\newcommand{\C}{\mathscr{C}}
\newcommand{\Ccat}{\mathbb{C}}

\newcommand{\V}{\mathscr{V}}
\newcommand{\Vcat}{\mathbb{V}}

\newcommand{\mto}{\to} 

\newcommand{\LeftAct}{\mathbf{LeftExt}}

\newcommand{\catname}[1]{ \mathbf{#1} }
\newcommand{\MonCat}{ \catname{MonCat} }
\newcommand{\MonBicat}{ \catname{MonBicat} }
\newcommand{\Vact}[1]{ {#1}\text{-}\catname{Act} }

\newcommand{\seqop}{;}
\DeclareMathOperator{\linseq}{\seqop}

\newcommand{\concseqOp}{ \mid \mid } 
\DeclareMathOperator{\concseq}{\concseqOp}

\usepackage{framed}

\begin{document}

\title[Strong, commutative, and concurrent pseudomonads]{
	 Effectful semantics in bicategories: strong, commutative, and concurrent pseudomonads
}

\author{Hugo Paquet}
\email{paquet@lipn.fr}
\affiliation{\institution{LIPN, Universit{\'e} Sorbonne Paris Nord}
  \city{Villetaneuse}
  \country{France}
}

\author{Philip Saville}
\email{philip.saville@cs.ox.ac.uk}
\affiliation{\institution{University of Oxford}
  \city{Oxford}
  \country{UK}}

\begin{abstract}
We develop the theory of strong and commutative monads in the 2-dimensional setting of bicategories. This provides a framework for the analysis of effects in many recent models which form bicategories and not categories, such as those based on profunctors, spans, or strategies over games.

We then show how the 2-dimensional setting provides new insights into the semantics of concurrent functional programs. We introduce concurrent pseudomonads, which capture the fundamental weak interchange law connecting parallel composition and sequential composition. This notion brings to light an intermediate level, strictly between strength and commutativity, which is invisible in traditional categorical models. We illustrate the concept with the continuation pseudomonad in concurrent game semantics.

In developing this theory, we take care to understand the coherence laws governing the structural 2-cells. We give many examples and prove a number of practical and foundational results.
      \end{abstract}


\maketitle

\section{Introduction}
\label{sec:introduction}
Moggi~\cite{Moggi1989,Moggi1991} famously observed that the structure of effectful computation is captured by the category-theoretic notion of \emph{strong monad}. 
This gives a framework for constructing new models and relating existing ones, abstracting away from any particular effect. 
This paper lays the foundations for modelling effects using monads in 
	\emph{2-dimensional}
category theory, where one has not just morphisms between objects, but also morphisms between morphisms (\Cref{sec:semant-2-dimens,sec:examples-of-bicats}). 
We have two motivations:
\begin{enumerate}
\item 
	Many recent semantic models are not categories but
  	\emph{bicategories} (\eg~\cite{FioreSpecies,cg1,
          template-games, fscd-poly}). However, we lack a unifying
        framework for these models. The time is right to set up the proper theoretical foundations for these models. 
\item  
Some well-known effects are already 2-categorical (see \Cref{sec:semant-2-dimens,sec:pseudo-mono-lax}). Making this structure explicit lets us see them as instances of a larger pattern, highlighting new connections, theoretical insights, and examples.
\end{enumerate}

In this paper we lift Moggi's foundational framework to the \mbox{2-dimensional} setting 
	(\Cref{sec:strong-pseudomonads,sec:symmetry}),
and show this is a suitable setting for modelling effectful programs
	(\Cref{sec:premonoidal-Kleisli-bicats}).
In doing so, we discover new notions that are invisible in 1-dimensional approaches (\Cref{sec:concurrency}).
Throughout we give plenty of examples (\eg~\Cref{sec:examples,sec:games}) and take care to mathematically justify our choice of definitions (\Cref{sec:justify}).

\subsection{Semantics in 2-dimensional categories}
\label{sec:semant-2-dimens}

A 2-dimensional category comes with objects $(A, B, \ldots)$, morphisms 
($f, g, \ldots : A \to B$), often called \emph{1-cells}, and \emph{2-cells} 
($\sigma, \tau, \ldots : {f \To g} $)
between the 1-cells. 
There are various kinds of 2-dimensional categories. In this paper we
work with \emph{bicategories}, a general notion in which the
associativity and identity laws for the composition of morphisms only
hold up to isomorphism. 

Bicategories typically arise when the composition of morphisms uses a
universal property (\eg~a categorical limit or colimit), because it is
then determined only up to isomorphism. 
There are many examples from
semantics:  game semantics
	\cite{cg1, template-games}, 
recent models of linear logic based on profunctors 
	\cite{fscd-poly, FioreSpecies,galal-profunctors},
and models describing the \mbox{$\beta\eta$-rewrites} of the simply-typed $\uplambda$-calculus~\cite{Seely1987,Hilken1996,LICS2019}.
These models come with more structure, and typically provide
  finer-grained or more intensional information than
  categorical~ones.
(See also \Cref{sec:examples-of-bicats} for detailed examples.) 

In addition to these recent models, many 
traditional categories from semantics are already 2-dimensional:
\begin{description}
\item \emph{Domain theory:} The basic idea of domain theory is to
  model recursion using a partial order on sets of continuous
  functions. This is a simple form of 2-dimensional structure on
  categories of domains, but there is a rich theory
  (\eg~\cite{Sterling2023,hyland-somereasons,taylor-algebraic-approach}).
\item \emph{Non-determinism:} Perhaps the simplest model for
  non-determinism is the category of sets and relations, where programs
  correspond to functions $A \to \mathcal{P}(B)$. The inclusion order on
  relations gives 2-dimensional structure with a natural semantic
  interpretation in terms of possible returned~values. 
\item \emph{Concurrency:} Maps of processes play a central role in
  models of concurrency based on event structures or presheaves
  \cite{winskel1986event,cattani1996presheaf}, and the abstract
  framework of \emph{concurrent Kleene Algebra} is similarly based on
  a partial order over processes \cite{Hoare2011}.
  \end{description}
2-dimensional aspects are also relevant on the syntactic
  side (see \cite{LICS2019,Olimpieri2021, Kerinec2023}). Other
  2-dimensional notions are also important, such as lax 2-dimensional functors
  for comparing models
\cite{DBLP:conf/csl/BaillotDER97,DBLP:conf/lics/ClairambaultOP23}.

\subsection{The monadic theory of effects}

We recall the traditional framework (\eg~\cite{Moggi1989,Moggi1991}). A strong monad on a monoidal category $\monoidal\catC$ is a monad $(T, \mu, \eta)$ equipped with  natural transformations
\[A \tens T(B) \xra{t_{A, B}} T(A \tens B) \qquad T(A) \tens B
  \xra{s_{A, B}} T(A \tens B)\]
called the \emph{left strength} and the \emph{right strength},
compatible with both the monoidal structure of $\catC$ and the monad
structure of $T$ (see~\eg~\cite{Kock1972,McDermott2022}).
An effectful program
	$(\Gamma \vdash M : A)$ 
is then modelled by a Kleisli arrow 
	$\Gamma \to TA$
        in $\catC$.
        
The strength makes substitution possible
even in the presence of free variables. For example, we can
substitute $M$ for a variable $x : A$ in another program
$(\Delta, x: A \vdash N : B)$ using the strength and the
Kleisli extension operation:
\[
\Delta \otimes \Gamma 
		\xra{\Delta \otimes M} 
	\Delta \otimes TA
		\xra{t_{\Delta, A}} 
	T(\Delta \otimes A) 
		\xra{\mathtt{>\!\!>\!=} N} 
	TB.
\]

This paper is about a notion of pseudostrength for
2-dimensional pseudomonads, where \emph{pseudo} indicates that the
equations in the definition of a strong monad have been replaced by
2-dimensional isomorphisms. These isomorphisms must in turn satisfy a number of
equations, which we justify in various
ways; see \Cref{sec:justify}.

\subsection{Pseudo monoidality and lax monoidality: commutativity and concurrency}
\label{sec:pseudo-mono-lax}

The theory of strong monads provides a basis for reasoning about
sequential composition. A natural question is whether the order of
execution matters for the two components of a pair: if
$(\Gamma \vdash M : A)$ and $(\Delta \vdash N : B)$ are effectful
programs then typically the program
\[
\Gamma, \Delta \vdash (M, N) : A \otimes B
\]
behaves differently depending on which component is evaluated
first. (We model contexts linearly to remain
as general as possible, since categories with products are instances
of monoidal categories. But this is orthogonal to the topic of
this paper.)

\subsubsection{Commutativity}
An effect is called commutative if the choice of
evaluation order for pairs has no impact on program behaviour. 
For example, random choice and divergence are commutative effects;
printing and state are not.
Correspondingly, a strong monad
is called {commutative} when the equation 
\begin{equation} \label{eq:commutativity-defn}
\begin{tikzcd}[row sep=.8em]
	& {TA \otimes TB} \\
	{T(A \otimes TB)} && {T(TA \otimes B)} \\[1em]
	{TT(A \otimes B)} && {TT(A \otimes B)} \\
	& {T(A \otimes B)}
	\arrow["s"', from=1-2, to=2-1]
	\arrow["t", from=1-2, to=2-3]
	\arrow[""{name=0, anchor=center, inner sep=0}, "Tt"', from=2-1, to=3-1]
	\arrow["\mu"', from=3-1, to=4-2]
	\arrow[""{name=1, anchor=center, inner sep=0}, "Ts", from=2-3, to=3-3]
	\arrow["\mu", from=3-3, to=4-2]
\end{tikzcd}
\end{equation}
holds. This is a semantic counterpart to the property that the
evaluation order for pairs does not affect program behaviour: commutative monads model commutative effects.
In
\Cref{sec:symmetry} we will define commutative pseudomonads by
replacing (\ref{eq:commutativity-defn}) with an invertible 2-cell, subject to coherence axioms.
  
\subsubsection{Monoidality}
Kock~\cite{Kock1970,Kock1972} showed that, for a
commutative monad $T$, the family of maps
\begin{equation}
  \label{eq:3}
 \chi_{A, B} : TA \otimes TB \longrightarrow T(A \otimes B)
\end{equation}
defined by either of the routes around (\ref{eq:commutativity-defn}) gives $T$ the structure of a \emph{monoidal}
monad; and that, conversely, given maps as in \eqref{eq:3} satisfying suitable equations we can
recover a commutative strength for $T$. In this paper we prove a general
2-categorical version of Kock's theorem (\Cref{res:monoidal-iff-commutative}):
pseudomonoidality of a pseudomonad corresponds to pseudocommutativity.

\subsubsection{Concurrency}
By moving to a 2-dimensional setting we can give a presentation of
concurrency. The starting observation is that a monoidal structure for
$T$ could be used to evaluate program fragments in parallel:
\[
P \parallel Q \ \ :=\ \  \Gamma \otimes \Delta \xra{P \otimes Q} TA
\otimes TB \overset{\chi_{A,B}}{\longrightarrow}{T(A \otimes B)}
\]
By Kock's theorem, this parallel evaluation is semantically indistinguishable from either of the two sequential executions: modelling concurrency in this way forces the effect to be commutative.

In a 2-dimensional category, however, we can weaken the notion of
monoidality to obtain a setting in which programs with \emph{non-commutative} effects can be evaluated in parallel, according to a 2-dimensional constraint:
\[\begin{tikzcd}[row sep=.8em]
	& {TA \otimes TB} \\
	{T(A \otimes TB)} && {T(TA \otimes B)} \\[1em]
	{TT(A \otimes B)} && {TT(A \otimes B)} \\
	& {T(A \otimes B)}
	\arrow["s"', from=1-2, to=2-1]
	\arrow["t", from=1-2, to=2-3]
	\arrow[""{name=0, anchor=center, inner sep=0}, "Tt"', from=2-1, to=3-1]
	\arrow["\mu"', from=3-1, to=4-2]
	\arrow[""{name=1, anchor=center, inner sep=0}, "Ts", from=2-3, to=3-3]
	\arrow["\mu", from=3-3, to=4-2]
	\arrow[""{name=2, anchor=center, inner sep=0}, "\chi", from=1-2, to=4-2]
	\arrow[shorten <=16pt, shorten >=16pt, Rightarrow, from=0, to=2]
	\arrow[shorten <=16pt, shorten >=16pt, Rightarrow, from=1, to=2]
\end{tikzcd}\]
The 2-cells above are not invertible in general, and do not make the pseudomonad commutative.
Replacing the equation~(\ref{eq:commutativity-defn}) by a pair of
non-invertible 2-cells, as above, corresponds to replacing the equation
$
	(P \concseq Q) \linseq\ (P' \concseq Q')  = (P \linseq P') \concseq (Q \linseq Q')
$	
relating sequential and parallel composition of processes by the \emph{weak interchange law} for parallel and sequental composition
\begin{equation} \label{eq:weak-interchange}
	(P \concseq Q) \linseq\ (P' \concseq Q')  \implies (P \linseq P') \concseq (Q \linseq Q')
\end{equation}
attributed to Hoare, M\"oller, Struth, and Wehrman \cite{Hoare2011}. 
This law is a basic feature of maps in
models of concurrency. Intuitively, the program on the left has more
dependencies---and so fewer possible traces---than the right one: see~\Cref{fig:introweak} for an illustration
with event structures (made formal in \Cref{sec:games}). 

The 2-categorical nature of the weak interchange law is already appreciated (see \cite{mellies2020concurrent}); in this paper we reframe it in the general context of 2-dimensional monad theory and computational effects.
We show that the appropriate monadic abstraction for modelling
the parallel execution of effectful programs is a particular class of lax monoidal pseudomonads, in
which certain structural 2-cells are not required to be invertible. 
These are a fully 2-dimensional generalisation of the concurrent monads of Rivas and Jaskelioff~\cite{Rivas2019}.
Accordingly, we call these \emph{concurrent pseudomonads}
(\Cref{def:concurrent-pseudomonad}).

Concurrent pseudomonads are always strong
(\Cref{res:concurrent-implies-bistrong}) and, as we explain, in the
Kleisli bicategory for a concurrent pseudomonad, the premonoidal
structure determines a lax functor $\otimes$ of two arguments (\Cref{res:concurrent-to-lax-monoidal}). 
This corresponds precisely to requiring a 2-cell as in~(\ref{eq:weak-interchange}).

\begin{figure}[t]
\[
\begin{minipage}{3cm}
	\centering
	$(a \concseq c) \linseq\: (b \concseq d)$ \\[2mm]
  \begin{tikzpicture}
\node[posnode] (4) at (4.8, 2.2) {};
	\node at (5.2, 2.2) {$c$};    
    \node[posnode] (5) at (4.8, 1) {};  
    \node at (5.2, 1) {$d$};    
    \node[posnode] (6) at (3.3, 2.2) {};
    \node at (2.9, 2.2) {$a$};          
    \node[posnode] (7) at (3.3, 1) {};
    \node at (2.9, 1) {$b$};

\draw [bend right=0, strat-causality] (4) to (5);
\draw [bend left=0, strat-causality] (6) to (7);
    \draw [bend right=5, strat-causality] (4) to (7);
    \draw [bend left=5, strat-causality] (6) to (5);
\end{tikzpicture}
 \end{minipage}
\begin{minipage}{1cm}
 \centering
 	 \vspace{1.4cm}
         \begin{tikzcd}[column sep=3.5em]
           {} \arrow[Rightarrow, thick]{r}{} & {}
           \end{tikzcd}
  \vspace{7mm}
  \end{minipage}
  \hspace{7mm}
\begin{minipage}{3cm}
	\centering
	$(a \linseq b) \concseq (c \linseq d)$ \\[2mm]
  \begin{tikzpicture}
\node[posnode] (4) at (4.8, 2.2) {};
	\node at (5.2, 2.2) {$c$};    
    \node[posnode] (5) at (4.8, 1) {};  
    \node at (5.2, 1) {$d$};    
    \node[posnode] (6) at (3.3, 2.2) {};
    \node at (2.9, 2.2) {$a$};          
    \node[posnode] (7) at (3.3, 1) {};
    \node at (2.9, 1) {$b$};              
\draw [bend right=0, strat-causality] (4) to (5);
\draw [bend left=0, strat-causality] (6) to (7);
\end{tikzpicture}  
 \end{minipage}
\]
\caption{The weak interchange law of sequential and parallel
  composition, as a map of event structures
  (see~\Cref{sec:games}).}
\label{fig:introweak}
\vspace{-3mm}
\end{figure}

\subsection{Outline}
\label{subsec:thispaper}

We begin with an introduction to bicategories and their basic theory 
	(\Cref{sec:examples-of-bicats,sec:bicat-theory}). 
We then introduce a new definition of strong pseudomonads 
	(\Cref{sec:strong-pseudomonads}),
and illustrate this 
with plenty of examples 
	(\Cref{sec:examples}).

We then turn to commutative and monoidal structure (\Cref{sec:symmetry}). We define monoidal pseudomonads and generalise Hyland \& Power's definition for commutative pseudomonads~\cite{Hyland2002}, then prove a version of Kock's theorem that the two are interchangeable (\Cref{res:monoidal-iff-commutative}).
We also explore the structure of the Kleisli bicategory for strong and commutative pseudomonads (\Cref{sec:premonoidal-Kleisli-bicats}). 

In \Cref{sec:concurrency} we introduce concurrent pseudomonads and show they are strong; we also observe their Kleisli bicategory does indeed model the weak interchange law~(\ref{eq:weak-interchange}). \Cref{sec:games} illustrates the key ideas with an extended example in concurrent game semantics.

Finally, in \Cref{sec:justify} we put the definitions in their proper mathematical context---namely, as internal pseudomonads---and establish a form of coherence result. Together, these give us confidence in the correctness of our definitions, especially the often-subtle question of how to choose coherence axioms on the 2-cells.

\section{Two examples of bicategories}
\label{sec:examples-of-bicats}

As an introduction to bicategories, we consider two illustrative examples.
First we look at a model based on spans. Spans occur widely in models of programming languages and computational processes (\eg~\cite{template-games,  Fiadeiro2007, Genovese2021, Baez2016}). 

\subsubsection*{\textsc{Example.} Spans of sets}
Consider a model in which objects are sets and a morphism from $A$ to
$B$ consists of a set $S$ and a span of functions 
	$A \leftarrow S \to B$. 
We can compose a pair of morphisms 
	$A \leftarrow S \to B$ 
and 
	$B \leftarrow R \to C$ 
using a pullback of functions:
\begin{equation*}
\label{eq:composition-of-spans}
\begin{tikzcd}[column sep=.3em, row sep=0.3em]
	&[1em] & {R \circ S} \\
	& S && R \\
	A && B &&[1em] C
	\arrow[from=2-2, to=3-1]
	\arrow[from=2-2, to=3-3]
	\arrow[from=2-4, to=3-3]
	\arrow[from=2-4, to=3-5]
	\arrow[from=1-3, to=2-2]
	\arrow[from=1-3, to=2-4]
	\arrow["\lrcorner"{anchor=center, pos=0.125, rotate=-45}, draw=none, from=1-3, to=3-3]
\end{tikzcd}
\end{equation*}
This correctly captures a notion of `plugging together' spans but is only 
associative in a weak sense: the two ways of taking pullbacks
\begin{equation}
	\label{eq:assoc-for-spans}
	\begin{tikzcd}[column sep=0.7em, row sep=0.7em, scalenodes=0.6]
	&&& \bullet \\
	&& \bullet \\
	& \bullet && \bullet && \bullet \\
	\bullet && \bullet && \bullet && \bullet
	\arrow[from=3-2, to=4-3]
	\arrow[from=3-4, to=4-3]
	\arrow[from=3-4, to=4-5]
	\arrow[from=2-3, to=3-2]
	\arrow[from=2-3, to=3-4]
	\arrow["\lrcorner"{anchor=center, pos=0.125, rotate=-45}, draw=none, from=2-3, to=4-3]
	\arrow[from=3-6, to=4-5]
	\arrow[from=1-4, to=2-3]
	\arrow[from=1-4, to=3-6]
	\arrow[from=3-6, to=4-7]
	\arrow["\lrcorner"{anchor=center, pos=0.125, rotate=-45}, draw=none, from=1-4, to=3-4]
	\arrow[from=3-2, to=4-1]
      \end{tikzcd}
\hspace{8mm}
\begin{tikzcd}[column sep=0.7em, row sep=0.7em, scalenodes=0.6]
	&&& \bullet \\
	&&&& \bullet \\
	& \bullet && \bullet && \bullet \\
	\bullet && \bullet && \bullet && \bullet
	\arrow[from=3-2, to=4-3]
	\arrow[from=3-4, to=4-3]
	\arrow[from=3-4, to=4-5]
	\arrow[from=3-6, to=4-5]
	\arrow[from=3-6, to=4-7]
	\arrow[from=3-2, to=4-1]
	\arrow[from=1-4, to=3-2]
	\arrow[from=2-5, to=3-4]
	\arrow[from=2-5, to=3-6]
	\arrow[from=1-4, to=2-5]
	\arrow["\lrcorner"{anchor=center, pos=0.125, rotate=-45}, draw=none, from=1-4, to=3-4]
	\arrow["\lrcorner"{anchor=center, pos=0.125, rotate=-45}, draw=none, from=2-5, to=4-5]
	\end{tikzcd}
    \end{equation}
are not generally equal, but they can be shown to be isomorphic by the universal property that
defines pullbacks. 
Similarly, the span 
    	$A \xleftarrow{\id} A \xrightarrow{\id} A$ 
   	is only a weak identity for composition, because pulling back along
    $\id$ only gives an isomorphic set. 
    
   To describe the laws of composition in this model, therefore, we require a notion of morphism between spans. 
If $S$
    and $S'$ are spans from $A$ to
    $B$, then a map between them is a function $\sigma : S \to S'$
    that commutes with the span legs on
    each side:
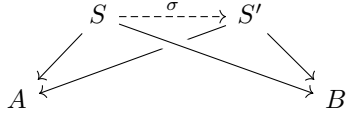
\[
     \begin{tikzcd}[column sep = 1.5em]
         \: & S & \: & {S'} \\
         A & \: & \: & \: &  B  
         \arrow[from=1-2, to=2-1]
         \arrow[from=1-4, to=2-1]
         \arrow[crossing over, from=1-2, to=2-5]
         \arrow[from=1-4, to=2-5]
         \arrow["\sigma", dashed, from=1-2, to=1-4]
       \end{tikzcd}
   \]
The two iterated composites in~(\ref{eq:assoc-for-spans}) are
isomorphic as spans, so composition of spans is associative up to
isomorphism. Similarly, the identity span is unital up to
isomorphism. Because these isomorphisms arise from a universal
property, they behave well together. 
Bicategories axiomatise this situation.

\begin{definition}[{\cite{Benabou1967}}]
A bicategory $\B$ consists of: 
\begin{itemize}
\item A collection of objects $A, B, \ldots$
\item For all objects $A$ and $B$, a collection of morphisms from
  $A$ to $B$, themselves related by morphisms: thus we have a \emph{hom-category} $\B(A, B)$ whose
  objects (typically denoted $f, g : A \to B$) are called \emph{1-cells}, 
  and whose morphisms (typically denoted $\sigma, \tau : f \To g$) are called
  \emph{2-cells}. The category structure means we can compose 2-cells between parallel 1-cells.
\item For all objects $A, B,$ and $C$, a composition functor 
     $
    	 \circ_{A, B, C} : \B(B, C) \times \B(A, B) \longrightarrow \B(A, C)
       $
      and, for all $A$, an identity 1-cell $\Id_A \in \B(A, A)$.
    \item Coherent structural 2-cells: since the composition of \mbox{1-cells} is weak,
we have a natural family of invertible 2-cells
    $\a_{f, g, h} : 
    	(f \circ g) \circ h \Longrightarrow f \circ (g \circ h)$ 
    instead of the usual associativity equation. 
    Similarly, we have natural families of invertible 2-cells
$\l_f : \Id_B \circ f \Longrightarrow f$ 
    and 
    $\r_f : f \circ \Id_A \Longrightarrow f$
instead of the left and right identity laws. 
These 2-cells must satisfy coherence axioms similar to those for a monoidal category. 
\end{itemize}
\end{definition}

To illustrate further we consider the $\Para$ construction, which is
a general way to build models of parametrized processes~\cite{Hermida2012,Fong2019}
	(see also~\cite{Capucci2022,Cruttwell2022,hugos-popl-paper}). 
In this bicategory, the 2-cells are reparametrizations, and the
weakness arises because we are tracking extra information.
We will use this bicategory several times, so we spell out the definition in detail.
	
\paragraph*{{\sc Example}: the $\Para$ construction.}
Starting from a monoidal category $(\catC, \otimes, I)$, the 
 bicategory $\Para(\catC)$ is defined as follows:
\begin{itemize}
\item The objects are those of $\catC$.
 \item A 1-cell from $A$ to $B$ is a parametrized
   $\catC$-morphism, defined as an object $P \in \catC$ together with a morphism 
   	$f : P \tens A \to B$ in $\catC$. The object
   $P$ is thought of as a space of parameters. 
  \item A 2-cell from $f : P \tens A \to B$ to $g : P' \tens A \to B$
     is a reparametrization
 map, \emph{i.e.}~a map $\sigma : P \to P'$ such that 
$g \circ (\sigma \tens A) = f$.
\end{itemize}
Composition of 1-cells is defined using the tensor product of
parameters: if $f : P \tens A \to B$ and $g : Q \otimes B \to C$,
then $g \circ f$ is the object $Q \otimes P$ equipped with the map
\[
  (Q \tens P) \tens A
  \xra{\cong}
   Q \tens (P \tens A)
   \xra{Q\tens f}
   Q \tens B
   \xra{g}
   C
   \]
where the first map is the associativity of the tensor product. 

If we also have $h : R \otimes C \to D$, then the 
two composites $(h \circ g) \circ f$ and $h \circ (g \circ f)$ 
have parameter spaces $(R \tens Q) \tens P$ and 
	$R \tens (Q \tens P)$, respectively. 
Because the tensor product in a monoidal category
is generally associative only up to isomorphism, 
these 1-cells are only isomorphic in $\Para(\catC)$. 
A similar argument applies to the identity laws, so
$\Para(\catC)$ is a bicategory with associativity and unit isomorphisms given by 
$\catC$'s monoidal structure.

\section{Pseudofunctors, pseudomonads, and monoidal bicategories}
\label{sec:bicat-theory}

Many concepts in category theory have corresponding versions for
bicategories. 
We first summarise the basic definitions of
pseudofunctors, pseudonatural transformations, and modifications
	(\Cref{sec:basic-notions}), then discuss the
bicategorical notions of monad
	(\Cref{sec:pseudomonads-and-kleisli-bicats})
and monoidal structure 
	(\Cref{sec:monoidal-bicategories})
needed for this paper. 
For reasons of space we only give a brief outline and omit the coherence axioms. For a full overview of the basic bicategorical definitions,
see~\cite{Leinster1998}; for the definition of (symmetric) monoidal bicategories, including many beautiful diagrams, see~\cite{Stay2016}.
Gentle introductions to the wider subject of bicategories include~\cite{Benabou1967, Johnson2021}; 
a more theoretical-computer science perspective is available 
	in~\cite{Power1995, Power1998}.

\subsection{Basic notions}
\label{sec:basic-notions}

Morphisms of bicategories are called pseudofunctors. Just as bicategories are categories `up to isomorphism', so pseudofunctors are functors `up to isomorphism'.

\begin{definition}
A \emph{pseudofunctor} $F : \B \to \C$ consists of:
\begin{itemize}
    \item A mapping $F : ob(\B) \to ob(\C)$ on objects;
    \item A functor $F_{A, B} : \B(A, B) \to \C(FA, FB)$ for each $A, B \in \B$;
    \item 
    	A \emph{unitor}
    		$\psi_A : \Id_{FA} \XRA\iso F(\Id_A)$
    	for each $A \in \B$;
    \item 
    	A \emph{compositor}
    		$\phi_{f,g} : F(f) \circ F(g) \XRA\iso F(f \circ g)$
    	for every composable pair of 1-cells $f$ and $g$,
    	natural in $f$ and $g$.
\end{itemize}
This data is subject to three axioms similar to those for strong
monoidal functors~(see~\eg~\cite{Leinster1998}). 
\end{definition}

We generally abuse notation by referring to a pseudofunctor $(F, \phi, \psi)$ simply as $F$; where there is no risk of confusion, we shall employ similar abuses for structure throughout. A pseudofunctor is called \emph{strict} if $\phi$ and $\psi$ are both the identity.

\begin{example}
	\label{ex:pseudofunctor-on-para}	
	Every endofunctor $F$ 
	on a monoidal category $\monoidal{\catC}$ 
	with a strength
		$t_{A,B} : A \tens F(B) \to F(A \tens B)$
	 (see~\eg~\cite{Kock1972})
	 determines a strict endo-pseudofunctor $\para{F}$ on $\Para(\catC)$. 
The action on objects is the same, and on 1-cells 
		$\para{F}(P \tens A \xra{f} B)$
	is the object $P$ together with the composite
		$\big( 
			P \tens FA \xra{t} F(P \tens A) \xra{Ff} FB
		\big)$.
\end{example}

Transformations between pseudofunctors are like natural transformations, except one must say in what sense naturality holds for each 1-cell.

\begin{definition}
  \label{def:psnat}
For pseudofunctors $F, G : \B \to \C$, a \emph{pseudonatural transformation} $\eta : F \To G$ consists of:
\begin{itemize}
\item 
  A 1-cell $\eta_A : FA \to GA$ for every $A \in \B$; 
\item 
  For every $f : A \to B$ in $\B$ an invertible 2-cell 
  \begin{equation}
    \label{eq:pseudonat}
  \begin{tikzcd}
    FA 
    \arrow{r}{Ff}
    \arrow{d}[swap]{\eta_A}
    &
    FB
    \arrow[phantom]{dl}[description]{\twocell{\cellOf{\eta}_f}}
    \arrow{d}{\eta_{B}}
    \\
    GA 
    \arrow{r}[swap]{Gf}
    & 
    GB
  \end{tikzcd}
\end{equation}
  natural in $f$ and satisfying identity and composition laws.
\end{itemize}
\end{definition}

\begin{example}
	\label{ex:pseudonat-on-para}
	Every natural transformation
		$\sigma : F \To F'$
	between strong endofunctors
		$(F, s)$
	and
		$(G, t)$
	which is compatible with the strengths
			(`strong natural transformation': see~\eg~\cite{McDermott2022}) 
determines a pseudonatural transformation 
		$\para\sigma : \para{F} \To \para{G}$
	on $\Para(\catC)$. 
Each component $(\para\sigma)_A$ is just $\para{\sigma_A}$, and for a 1-cell
		$f : P \tens A \to B$
	the 2-cell
		$\cellOf{\para{\sigma}}_f$
	witnessing naturality 
	is the canonical isomorphism
		${\tensu \tens P \xra{\iso} P \tens \tensu}$. 
\end{example}

Because bicategories have a second layer of structure,
there is also a notion of map between 
pseudonatural transformations.

\begin{definition}
	A \emph{modification} 
		$\m : \eta \mto \theta$
	between pseudonatural transformations 
		$F \To G : \B \to \C$ 
	consists of a 2-cell 
		$\m_B : \eta_B \To \theta_B$
	for every $B \in \B$, subject to an axiom expressing compatibility between 
	$\m$ and each $\cellOf{\eta}_f$ and $\cellOf{\theta}_f$.
\end{definition}

For any bicategories $\B$ and $\C$ there exists a bicategory
	$\homBicat{\B}{\C}$
with objects pseudofunctors, 1-cells pseudonatural transformations, and 2-cells modifications.

\subsection{Pseudomonads and Kleisli bicategories}
\label{sec:pseudomonads-and-kleisli-bicats}

The bicategorical correlate of a monad is a \emph{pseudomonad}. 

\begin{definition}[{\cite{Marmolejo1997}}]
A \emph{pseudomonad} on a bicategory $\baseCat$ consists of a 
	pseudofunctor $T : \baseCat \to \baseCat$
equipped with:
	\begin{itemize}
\item 
		\emph{Unit} and \emph{multiplication}
		pseudonatural transformations 
$\eta : \id \To T$ 
			and 
				$\mu : T^2 \Longrightarrow T$,
            where $T^2 = T \circ T$;
\item Invertible modifications $\m, \n, \p$ with
                  components 
\[\begin{tikzcd}[
		execute at end picture={
			\foreach \nom in  {A,...,D}
  				{\coordinate (\nom) at (\nom.center);}
			\fill[\monadmcolour,opacity=\opacity] 
  				(A) -- (B) -- (D) -- (C);
		}
	]
	|[alias=A]|
	{T^3A} 
	& 
	|[alias=B]|
	{T^2A} \\
	|[alias=C]|
	{T^2A}  & 
	|[alias=D]|
	TA
	\arrow["{\mu_{TA}}", from=1-1, to=1-2]
	\arrow[""{name=0, anchor=center, inner sep=0}, "{T\mu_A}"', from=1-1, to=2-1]
	\arrow["{\mu_A}"', from=2-1, to=2-2]
	\arrow[""{name=1, anchor=center, inner sep=0}, "{\mu_A}", from=1-2, to=2-2]
	\arrow["\oncell{\m_A}"{}, Rightarrow, from=0, to=1,shorten=15, yshift=-1mm]
\end{tikzcd}
\qquad
\begin{tikzcd}[
		execute at end picture={
			\foreach \nom in  {1u,1d, 2l, 2r}
  				{\coordinate (\nom) at (\nom.center);}
			\fill[\monadncolour,opacity=\opacity] 
  				(1u) to[bend right = 32] (2l) -- (1d);
  			\fill[\monadpcolour,opacity=\opacity] 
  			  	(1u) to[bend left = 32] (2r) -- (1d);
		}]
	& |[alias=1u]| TA \\
	|[alias=2l]| {T^2 A} & |[alias=1d]| TA & |[alias=2r]| {T^2 A}
	\arrow["{\eta_{TA}}"', from=1-2, to=2-1, bend right = 25]
	\arrow["{\mu_A}"', from=2-1, to=2-2]
	\arrow[""{name=0, anchor=center, inner sep=0}, Rightarrow, no head, from=1-2, to=2-2]
	\arrow["{T\eta_{A}}", from=1-2, to=2-3, bend left = 25]
	\arrow["{\mu_A}", from=2-3, to=2-2]
	\arrow["\oncell{\n_A}", Rightarrow, from=2-1, to=0, shorten = 8]
	\arrow["\oncell{\p_A}"{swap}, Rightarrow, from=2-3, to=0, shorten = 8]
\end{tikzcd}\]

replacing the usual monad laws, and satisfying two
                further coherence axioms.
	\end{itemize}
\end{definition}

A simple example is given by the Writer pseudomonad on $\Cat$, the bicategory with objects small categories, 1-cells functors, and 2-cells natural transformations. The structural isomorphisms $\a, \l$ and $\r$ are all the identity (giving a \emph{2-category}).

\begin{example}\label{ex:writer-on-cat}
Let $(\catC, \otimes, I)$ be a 
monoidal category. The pseudofunctor $(-) \times \catC : \Cat \to \Cat$
has a pseudomonad structure with 1-cell components 
	\begin{equation*}
\begin{aligned}
         & \eta_\catD\  =\  \catD \xra{\iso} \catD \times 1 \xra{\catD \times I} \catD \times \catC \\
	 & \mu_\catD \ = \ 	(\catD \times \catC) \times \catC  \xra{\iso} \catD
                                                     \times (\catC
                                                     \times \catC)
                                                     \xra{\catD \times
                                                     \otimes} \catD
                                                     \times \catC
	\end{aligned}
	\end{equation*}
      and 2-cell components $\m, \n$ and $\p$ given by the associator and unitors for the
      monoidal structure in $\catC$.
\end{example}

\begin{example}
	\label{ex:pseudomonad-on-para}
	Every strong monad $(T, \mu, \eta, t)$ on a monoidal category $\monoidal\catC$ determines a pseudomonad on $\Para(\catC)$: the underlying pseudofunctor is $\para{T}$ and the pseudonatural transformations are 
		$\para\mu$ and $\para\eta$ (recall \Cref{ex:pseudonat-on-para}).
	This remains true if the monoidal structure is replaced by an \emph{action} (as in~\eg~\cite{Pareigis1977}).
\end{example}

\subsection{Monoidal bicategories}
\label{sec:monoidal-bicategories}
\label{sec:para-construction}

A monoidal bicategory is a bicategory equipped with a unit object and a tensor product
which is only weakly associative and unital. 
To motivate the construction, we explain 
how a symmetric monoidal category $\monoidal{\catC}$ induces a
monoidal structure on $\Para(\catC)$, with the same action on
objects.

The idea is that
we can combine the parameters using $\otimes$.  For 1-cells
	$f : P \tens A \to B$ 
and
	$g: P' \tens A' \to B'$, 
we set $f \paratens g$ to be the object
	$P \otimes P'$ 
equipped with
\[
(P \otimes P') \tens (A \otimes A')  
    	\xra{\cong} 
    (P \tens A) \otimes (P' \otimes A')  
    	 \xra{f \otimes g} 
    B \otimes B'
\]
where the first map is defined using the symmetry of $\otimes$. 
On \mbox{2-cells}, we use the tensor product of maps in $\catC$.
This construction does not strictly preserve identities and
composition, but it does preserve them
up to isomorphism. 
Thus, we get a pseudofunctor
$
  \para\tens : \Para(\catC) \times \Para(\catC) \longrightarrow
  \Para(\catC).
 $

We examine the sense in which this tensor is associative and unital,
by lifting the structural isomorphisms from $\catC$.
Every map $f : A \to B$ in $\catC$ determines a 1-cell 
 	$\para{f}$
 in $\Para(\catC)$ given by the object $\tensu$ and the composite 
 	$(\tensu \tens A \xra{\iso} A \xra{f} B)$, 
 where $\iso$ is the unit isomorphism.
If $f$ has an inverse $f^{-1}$,  the composite $\para{f} \circ \para{f^{-1}}$
 has parameter $\tensu \tens \tensu$ and thus cannot be the
 identity. But it is isomorphic to the identity: the pair $(\para{f},
 \para{f^{-1}})$ is known as an \emph{equivalence} 
 (an `isomorphism up to isomorphism').  
Thus, although the tensor  $\tens$ on $\catC$ is associative and
unital up to isomorphism, the tensor $\paratens$ on $\Para{(\catC)}$
is only associative and unital up to equivalence. 
The structural 1-cells are all pseudonatural in a canonical way (\Cref{ex:pseudonat-on-para}).

Following the general pattern of ``bicategorification'', the triangle and pentagon axioms of a monoidal category now only hold up to isomorphism:  one route round the pentagon has three sides and the other has two, so one composite has parameter $\tensu^{\tens 3}$ and the other has parameter $\tensu^{\tens 2}$.
These are canonically isomorphic, so we get families of invertible \mbox{2-cells} witnessing the categorical axioms. All the structure we have defined so far has used the canonical isomorphisms of $\catC$, so these families are actually modifications on $\Para(\catC)$.
Moreover, by the axioms of a monoidal category, these structural modifications satisfy axioms of their own.

In summary, a monoidal bicategory is a bicategory equipped with an object $\tensu$, a pseudofunctor $\paratens$, pseudonatural families of equivalences witnessing the weak associativity and unitality of $\paratens$, and invertible modifications witnessing the axioms of a monoidal category. 
We now make this precise, starting with the definition of equivalences. 
These generalize equivalences of categories.

\begin{definition}
	\label{def:equivalence}
	An \emph{equivalence} between objects $A$ and $B$ in a
        bicategory $\B$ is a pair of 1-cells $f : A \to B$ and
        $\psinv{f} : B \to A$ together with invertible 1-cells
        $f \circ \psinv{f} \To \Id_B$ and $\Id_A \To \psinv{f} \circ f$.
\end{definition}
  
A \emph{pseudonatural equivalence} is a pseudonatural transformation in which each component has the structure of an equivalence.

The definition is now as advertised.  To state it, we introduce some notation for the 2-cell diagrams---known as \emph{pasting diagrams}---that we will use in the rest of the paper. 
  
  \begin{notation}
  \label{diagram-notation}
    To save space and improve readability, 
    \begin{itemize}
      \item 
      	We use juxtaposition for the tensor product, e.g.~$(AB)C$ means
        $(A \tens B) \tens C$;
      \item 
      	We omit the subscripts on the components of pseudonatural
        transformations and modifications, \eg~$\m$ instead of $\m_A$;
      \item 
      	We use a subscript notation for the action of a
        pseudofunctor $T$, \eg~$T_{AT_B}$ means $T(A \tens T(B))$.
      \item 
      	We write $\cong$ for any pseudonaturality 2-cell as in \eqref{eq:pseudonat}, and in equations we omit the arrows showing the directions of 2-cells. These labels can be inferred from the type. 
      \end{itemize}
  \end{notation}

\begin{figure*}[t]
\centering
	\[
\hspace{7mm}
\begin{tikzcd}
	[
execute at end picture={
					\foreach \nom in  {A,B,C, D,E}
		  				{\coordinate (\nom) at (\nom.center);}
					\fill[\monoidalpcolour,opacity=\opacity] 
		  				(A) -- (B) -- (C) -- (D) -- (E);
		}
		]
\alias{A}
	{\left((A  B)  C\right)}  D 
	\arrow{r}{\alpha}
	\arrow[swap]{d}{\alpha D}
	\arrow[	"{\mathfrak{p}\:}"{left},
					from=2-2, 
					to=1-2,  
					Rightarrow, 
					shorten = 2pt
					]
	&
	\alias{B}
	(AB) (CD)
	\arrow{r}{\alpha}
	&
	\alias{C}
	A{\left(B(C D) \right)}
	\\
\alias{E}
	{\left( A (BC) \right)}D
	\arrow[swap]{rr}{\alpha}
	&
	\:
	&
	\alias{D}
	A{\left((BC) D\right)}
	\arrow[swap]{u}{A\alpha}
\end{tikzcd}
\hspace{4mm}
\begin{tikzcd}[
	column sep = 1em,
	execute at end picture={
								\foreach \nom in  {A,B,C}
					  				{\coordinate (\nom) at (\nom.center);}
								\fill[\monoidalmcolour,opacity=\opacity] 
					  				(A) -- (B) -- (C);
	  }
	]
\: &
\alias{C}
{A  B}
&
\: 
\\
\alias{A}
{(A  I) B} 
\arrow{ur}{\rho B}
\arrow[swap]{rr}{\alpha}
&
\:
& 
\alias{B}
{A (I B)} 
\arrow[name=0]{ul}[swap]{A\lambda}
\arrow["{\mathfrak{m}\:}"{}, 
 				shift right=1, 
 				shorten <=6pt, 
 				shorten >=8pt, 
 				Rightarrow, 
 				from=2-2, 
 				to=1-2]
\end{tikzcd}
\hspace{3mm}
\begin{tikzcd}[
	column sep=2em,
	execute at end picture={
			\foreach \nom in  {A,B,C}
  				{\coordinate (\nom) at (\nom.center);}
			\fill[\monoidallcolour,opacity=\opacity] 
  				(A) -- (B) to[bend left = 36] (C);
		}
	]
\alias{A}
{(I A)  B} 
& 
\alias{B}
{I  (A  B)} \\
\alias{C}
{A B}
\arrow["\alpha", from=1-1, to=1-2]
\arrow[""{name=0, anchor=center, inner sep=0}, bend left = 25, "\lambda", from=1-2, to=2-1]
\arrow[""{name=1, anchor=center, inner sep=0}, "{\lambda  B}"', from=1-1, to=2-1]
\arrow["{\mathfrak{l}}"{}, 
	    					shift left=1,
	    					yshift = 1mm,
	    					shorten <=6pt, 
	    					shorten >=6pt, 
	    					Rightarrow, 
	    					from=1, 
	    					to=0]
\end{tikzcd}		
\hspace{2mm}
\begin{tikzcd}[
	column sep=2em,
	execute at end picture={
											\foreach \nom in  {A,B,C}
								  				{\coordinate (\nom) at (\nom.center);}
											\fill[\monoidalrcolour,opacity=\opacity] 
								  				(A) to[bend right=38] (B) -- (C);
				  }
]
\alias{A}
(AB)I 
\arrow{r}{\alpha}
\arrow[bend right =30, name=0]{dr}[swap]{\rho}
&
\alias{C}
A(BI)
\arrow{d}{A\rho}
\\
\:
\arrow["{\mathfrak{r}}",
			Rightarrow, 
			xshift=6mm,
			yshift=1mm,
			shorten <=14pt, 
			shorten >=14pt, 
			from=2-1, 
			to = 1-2]
&
\alias{B}
AB
\end{tikzcd}
\] 	\vspace{-3mm}
	\caption{The structural modifications of a monoidal bicategory}
	\label{fig:monoidal-bicat-modifications}
	\vspace{-2mm}
\end{figure*}

\begin{definition}[{\eg~\cite{Stay2016}}]
\label{def:monoidal-bicategory}
A \emph{monoidal bicategory} is a bicategory $\B$ equipped with a 
pseudofunctor $\otimes : \B \times \B \to \B$ and an object $I \in \B$, 
together with the following data: 
\begin{itemize}
    \item 
    	Pseudonatural equivalences $\alpha, \lambda$ and $\rho$ with 
    	components 
    		$\alpha_{A, B, C} : (A \otimes B) \otimes C \to A \otimes 
    		(B \otimes C)$
    		(the \emph{associator}),
    		$\lambda_A : I \otimes A \to A$,
    	and
    		$\rho_A : A \otimes I \to A$
    		(the \emph{unitors});
    \item 
    	Invertible modifications
    		$\mathfrak{p}, \mathfrak{l}, \mathfrak{m}$ and $\mathfrak{r}$
    	with components shown in \Cref{fig:monoidal-bicat-modifications},
    	subject to coherence axioms.
    \end{itemize}
    A \emph{symmetric} monoidal bicategory is a monoidal bicategory equipped with
    	\newcommand{\braidingR}{R}
    a pseudonatural equivalence $\braid$ with components
      $\braid_{A, B} : A \otimes B \to B \otimes A$, 
    called the \emph{braiding}, and invertible modifications governing the possible shufflings of three objects and expressing the
    symmetry of the braiding, subject to coherence axioms. 
  \end{definition}
  
  For example (see~\eg~\cite{Stay2016} for full details), the cartesian product on the category $\Set$ induces a
  monoidal structure on the bicategory $\Span(\Set)$
introduced in
	\Cref{sec:examples-of-bicats}. 
The pseudofunctor $\otimes$ is defined on objects as 
	$A \otimes A' = A \times A'$, and for spans 
		$A \leftarrow S \rightarrow B$ 
	and 
	$A' \leftarrow S' \rightarrow B'$ 
	we take the component-wise product to obtain 
  	$A \times A' \leftarrow S \times S' \rightarrow B \times B'$. 
  \if 
  The action of $\otimes$ on 2-cells is
  similar, and composition is preserved up to isomorphism because pullbacks commute with products.
To define the structural 1-cells,
  we use that any function $f : A \to B$ gives a span
  $A \xleftarrow{\id} A \xrightarrow{f} B$ from $A$ to $B$, and that if $f$
  is an isomorphism then this span is an equivalence. 
We then  obtain $\alpha, \lambda$ and  $\rho$ by lifting the associativity and unit
  isomorphisms for $\times$ in $\Set$.
  This lifting has a universal property, and the
structural 2-cells are induced in a canonical way. 
  \fi

We also record the outcome of our discussion above; this establishes a conjecture made in~\cite{Capucci2022}. 

\begin{example}
   	If $\monoidal{\catC}$ is a symmetric monoidal category, this lifts to a symmetric monoidal structure on 
   		$\Para(\catC)$.
\end{example}

 \paragraph*{General point.} 
 The  coherence axioms of a monoidal
 bicategory can be difficult to verify directly. 
However,
 in many cases of interest the monoidal structure is induced from a
 more fundamental construction, as in $\Span(\Set)$ above.
This gives a systematic method for constructing (symmetric) monoidal
 bicategories: see \cite{wester2019constructing}.

\subsection{Coherence theorems} 
\label{sec:coherence}

As we have seen, bicategorical structures involve considerable 
data and many equations.
Much of the difficulty, however, is tamed by various \emph{coherence theorems}. 
These generally show that any two parallel 2-cells built out of the
structural data are equal. 
Appropriate coherence theorems apply to bicategories~\cite{MacLane1985} 
pseudofunctors~\cite{Gurski2013}, 
	(symmetric) monoidal bicategories~\cite{Gordon1995,Gurski2013symmetricbicats}
	and pseudomonads~\cite{Lack2000}.

We rely heavily on the coherence of bicategories and pseudofunctors
when writing pasting diagrams of 2-cells: in particular we omit all compositors
and unitors for pseudofunctors, and 
ignore the weakness of 1-cell composition. 
Thus, strictly speaking our diagrams do not type-check, but coherence guarantees the resulting 2-cell is the same no matter how one fills in the structural details. This is standard practice; for precise justification see \eg ~\cite[\S2.2]{SchommerPries2009}.

 \section{Strong pseudomonads}
 \label{sec:strong-pseudomonads}
 We follow the categorical setting by first saying what it means for a pseudo\emph{functor} to be strong, then giving the additional data and axioms to make a pseudo\emph{monad} strong.

\subsection{Strong pseudofunctors}
\label{sec:strong-pseudofunctors}

For the moment we only consider strengths on the left. In all diagrams below we follow our \Cref{diagram-notation}.

\begin{definition}
Let $(\B, \otimes, I)$ be a monoidal bicategory. A \emph{left strength} for
a pseudofunctor $T : \B \to \B$ is a pseudonatural
transformation $t_{A, B} : A \otimes TB \to T(A \otimes B)$, equipped with
invertible modifications $\x$ and $\y$ expressing the
compatibility of $t$ with the left unitor and the associator:
\[
  \begin{tikzcd}[
  	execute at end picture={
  	      	  						\foreach \nom in  {A,B,C}
  	      	  			  				{\coordinate (\nom) at (\nom.center);}
  	      	  						\fill[\strongxcolour,opacity=\opacity] 
  	      	  			  				(A) to[bend right = 27] (C) -- (B);
  	      	  	}
  ]
  	\alias{A}
	{T_{IA}} & \alias{B} {IT_A} \\
	& \alias{C} {T_A}
	\arrow["t"', from=1-2, to=1-1]
	\arrow[""{name=0, anchor=center, inner sep=0}, "{T_\lambda}"', bend right = 20, 
						from=1-1, to=2-2]
	\arrow["\lambda", from=1-2, to=2-2]
	\arrow["\oncell{\x}", Rightarrow, pos=0.6, from=0, to=1-2, shorten = 6]
      \end{tikzcd}
\qquad
\begin{tikzcd}[
      		column sep = 2em,
      	  	execute at end picture={
      	  						\foreach \nom in  {A,B,C, D,E}
      	  			  				{\coordinate (\nom) at (\nom.center);}
      	  						\fill[\strongycolour,opacity=\opacity] 
      	  			  				(A) -- (B) -- (C) -- (D) -- (E);
      	  	}
      ]
     \alias{A}
	{(AB)T_C} && \alias{B} {T_{(AB)C}} \\
	\alias{E} {A(BT_C)} & \alias{D} {AT_{BC}} & \alias{C} {T_{A(BC)}}
	\arrow["t", from=1-1, to=1-3]
	\arrow[""{name=0, anchor=center, inner sep=0}, "\alpha"', from=1-1, to=2-1]
	\arrow["At"', from=2-1, to=2-2]
	\arrow["t"', from=2-2, to=2-3]
	\arrow[""{name=1, anchor=center, inner sep=0}, "{T_\alpha}", from=1-3, to=2-3]
	\arrow[Rightarrow, from=0, to=1, "\y", shorten = 30]
      \end{tikzcd}\]
These modifications must themselves be compatible with the
monoidal structure, as per the two axioms of \Cref{fig:coherence-axioms-for-strong-pseudofunctor}.
\end{definition}

A left strength for a pseudofunctor $T$
can be used to define a parametrised version of the functorial
action: for any map $\Gamma \otimes X \to Y$ we can now define a map
$\Gamma \otimes TX \to TY$. 
This suggests the following
	(recall \Cref{ex:pseudofunctor-on-para} and \Cref{ex:pseudonat-on-para}).

\begin{example}
	\label{ex:pseudofunctor-on-para-strong}
	If $(F, t)$ is a strong functor on a symmetric monoidal category 
		$\monoidal{\catC}$ (see~\eg~\cite{Kock1972,McDermott2022}), 
	then the induced pseudofunctor $\para{F}$ on $\Para(\catC)$ is also strong.
The pseudonatural transformation has components 
		$\para{t}_{A,B} := \para{t_{A,B}}$;
	this has parameter $\tensu$, so 
$\x$ and $\y$  are both of the form
		$\tensu^{\tens i} \xra\iso \tensu^{\tens j}$
	for $i, j \in \Nat$. 
\end{example}

\subsection{Strong pseudomonads}

If a strong pseudofunctor $T : \B \to \B$ is also a pseudomonad, then we must
ask for additional data to relate the strength and the monad
structure, and this data must be compatible with the modifications 
	$\x, \y$ we already have. 
	
\begin{figure*}[!ht]
\centering
\[
\hspace{-6mm}
\begin{tikzcd}[
  		scale cd=1, 
  		column sep=.5em,
  		execute at end picture={
  		      	  						\foreach \nom in  {A,B,C, D}
  		      	  			  				{\coordinate (\nom) at (\nom.center);}
  		      	  						\fill[\monoidalmcolour,opacity=\opacity] 
  		      	  			  				(A) -- (B) -- (C);
  		      	  			  			\fill[\strongxcolour,opacity=\opacity] 
  		      	  			  			  	(A) -- (C) -- (D);
  		 }
  	]
	\alias{A} {A(IT_B)} && \alias{B} {(AI)T_B} \\
	\alias{D} {AT_{IB}} & \alias{C} {AT_B} \\
	{T_{A(IB)}} && {T_{(AI)B}} \\
	& {T_{AB}}
	\arrow["\alpha"', from=1-3, to=1-1]
	\arrow["At"', from=1-1, to=2-1]
	\arrow["t"', from=2-1, to=3-1]
	\arrow["{\rho T_B}"{below, xshift=2mm}, from=1-3, to=2-2]
	\arrow["{AT_\lambda}"', from=2-1, to=2-2]
	\arrow[""{name=0, anchor=center, inner sep=0}, "t", from=2-2, to=4-2]
	\arrow[""{name=1, anchor=center, inner sep=0}, "t", from=1-3, to=3-3]
	\arrow["{T_{\rho B}}", from=3-3, to=4-2]
	\arrow["{T_{A\lambda}}"', from=3-1, to=4-2]
	\arrow[""{name=2, anchor=center, inner sep=0}, 
					"{\small A\lambda}"{above}, from=1-1, to=2-2]
	\arrow["\cong"{pos=0.2}, draw=none, from=3-1, to=0]
	\arrow["\cong"{description}, draw=none, from=0, to=1]
	\arrow["{ \oncell{\montrianglem} }"{description}, draw=none, from=2, to=1-3]
	\arrow["{\oncell{A\x}}"{yshift=-2mm, xshift = 2mm}, draw=none, from=2-1, to=2]
      \end{tikzcd}
\hspace{-2mm}
	=
	\hspace{-2mm}
\begin{tikzcd}[
	     		scale cd= 1, 
	     		column sep=0em,
	     		execute at end picture={
	  						\foreach \nom in  {A,B,C, D, E}
	  			  				{\coordinate (\nom) at (\nom.center);}
	  						\fill[\strongycolour,opacity=\opacity] 
	  			  				(A) -- (B) -- (C) -- (D);
	  			  			\fill[\monoidalmcolour,opacity=\opacity] 
	  			  			  	(C) -- (E) -- (D);
		   }
	]
	\alias{A} {A(IT_B)} && \alias{B} {(AI)T_B} \\
	{AT_{IB}} \\
	\alias{D} {T_{A(IB)}} && \alias{C} {T_{(AI)B}} \\
	& \alias{E} {T_{AB}}
	\arrow["\alpha"', from=1-3, to=1-1]
	\arrow["At"', from=1-1, to=2-1]
	\arrow["t"', from=2-1, to=3-1]
	\arrow[""{name=0, anchor=center, inner sep=0}, "t", from=1-3, to=3-3]
	\arrow[""{name=1, anchor=center, inner sep=0}, "{T_{\rho B}}", from=3-3, to=4-2]
	\arrow[""{name=2, anchor=center, inner sep=0}, "{T_{A\lambda}}"', from=3-1, to=4-2]
	\arrow["{T_\alpha}"{description}, from=3-3, to=3-1]
	\arrow["\oncell{\y}"{description, pos=0.4}, draw=none, from=2-1, to=0]
	\arrow["{ \oncell{T_\montrianglem} }"{description, yshift=2mm}, draw=none, from=2, to=1]
	\end{tikzcd}
\hspace{2mm}
\begin{tikzcd}[
	      		column sep=-1.5em, 
	      		row sep=1.1em, 
	      		scale cd=1,
	      		execute at end picture={
					\foreach \nom in  {A,B,C, D, E, X, Y, Z, P, Q}
		  				{\coordinate (\nom) at (\nom.center);}
					\fill[\strongycolour,opacity=\opacity]  (E) -- (D) -- (X) -- (Y) -- (E);
					\fill[\strongycolour,opacity=\opacity] (Y) -- (Z) -- (P) -- (Q);
		  			\fill[\monoidalpcolour,opacity=\opacity] (A) -- (B) -- (C) -- (D) -- (E) -- (A);
	      		}
	      	]
	  &[-1em] \alias{B} {A((BC)T_D)} &[-0.7em] &[-0.7em] \alias{C} {(A (BC))T_D} \\
		\alias{A} {A(B(CT_D))} &&&&[-1em] \alias{D} {((AB)C)T_D} \\
		\alias{P} {A(BT_{CD})} && \alias{E} {(AB)(CT_D)} \\
		{AT_{B(CD)}} && \alias{Q} {(AB)T_{CD}} \\
		\alias{Z} {T_{A(B(CD))}} &&&& \alias{X} {T_{((AB)C)D}} \\
		&& \alias{Y} {T_{(AB)(CD)}}
		\arrow["{\alpha T_D}"', from=2-5, to=1-4]
		\arrow["\alpha"', from=1-4, to=1-2]
		\arrow[""{name=0, anchor=center, inner sep=0}, "A\alpha"', from=1-2, to=2-1]
		\arrow[""{name=1, anchor=center, inner sep=0}, "\alpha"', from=3-3, to=2-1]
		\arrow[""{name=2, anchor=center, inner sep=0}, "\alpha"', from=2-5, to=3-3]
		\arrow["{A(Bt)}"', from=2-1, to=3-1]
		\arrow["At"', from=3-1, to=4-1]
		\arrow["{T_\alpha}", from=6-3, to=5-1]
		\arrow["{(AB)t}", from=3-3, to=4-3]
		\arrow[""{name=3, anchor=center, inner sep=0}, "t"', from=4-3, to=6-3]
		\arrow["{T_\alpha}", from=5-5, to=6-3]
		\arrow[""{name=4, anchor=center, inner sep=0}, "t", from=2-5, to=5-5]
		\arrow[""{name=5, anchor=center, inner sep=0}, "\alpha"'{swap}, from=4-3, to=3-1]
		\arrow["t"', from=4-1, to=5-1]
		\arrow["\oncell{ \y }"{description, pos=0.4}, draw=none, from=4-1, to=3]
		\arrow["\cong"{description}, draw=none, from=5, to=1]
		\arrow["\oncell{ \y }"{description}, draw=none, from=4-3, to=4]
		\arrow["\oncell{ \pentagonator }", draw=none, from=0, to=2]
	\end{tikzcd}
\hspace{-2mm}
	=
	\hspace{-2mm}
\begin{tikzcd}[
	 			scale cd=1,
	      column sep=-.8em, 
	      row sep=1.1em, 
	      every matrix/.append style={},
	         	execute at end picture={
	         					\foreach \nom in  {A,B,C, D, E, W, X, Y, Z, P}
	         		  				{\coordinate (\nom) at (\nom.center);}
	         					\fill[\strongycolour,opacity=\opacity]  (W) -- (X) -- (Y) -- (Z);
	         					\fill[\strongycolour,opacity=\opacity] (X) -- (B) -- (C) -- (P);
	         		  			\fill[\monoidalpcolour,opacity=\opacity] (A) -- (B) -- (C) -- (D) -- (E) -- (A);
	       	 }
	]
	& \alias{X} {A((BC)T_D)} &[-1.2em]&[-1.2em] \alias{P} {(A (BC))T_D} \\
		\alias{W} { A(B(CT_D))} &&&& {((AB)C)T_D} \\
		{A(BT_{CD})} & \alias{Y} {AT_{(BC)D}} \\
		\alias{Z} {AT_{B(CD)}} & \alias{B} {T_{A((BC)D)}} && \alias{C} {T_{(A(BC))D}} \\
		\alias{A} {T_{A(B(CD))}} &&&& \alias{D} {T_{((AB)C)D}} \\
		&& \alias{E} {T_{(AB)(CD)}}
		\arrow["{\alpha T_D}"', from=2-5, to=1-4]
		\arrow["\alpha"'{}, from=1-4, to=1-2]
		\arrow["A\alpha"', from=1-2, to=2-1]
		\arrow["{A(Bt)}"', from=2-1, to=3-1]
		\arrow["At"', from=3-1, to=4-1]
		\arrow["{T_\alpha}", from=6-3, to=5-1]
		\arrow["{T_\alpha}", from=5-5, to=6-3]
		\arrow[""{name=0, anchor=center, inner sep=0}, "t", from=2-5, to=5-5]
		\arrow["t"', from=4-1, to=5-1]
		\arrow["{T_{A\alpha}}", from=4-2, to=5-1]
		\arrow["{T_\alpha}"', from=4-4, to=4-2, shorten=-0.4em]
		\arrow["{T_{\alpha D}}", from=5-5, to=4-4]
		\arrow[""{name=1, anchor=center, inner sep=0}, "At", from=1-2, to=3-2]
		\arrow["t", from=3-2, to=4-2]
		\arrow["{AT_\alpha}"'{above}, from=3-2, to=4-1]
		\arrow["\iso"{yshift=2mm}, draw=none, shift left=2, from=4-1, to=4-2]
		\arrow[""{name=2, anchor=center, inner sep=0}, "t", from=1-4, to=4-4]
		\arrow["{ \oncell{T_\pentagonator} }", draw=none, 
						shorten <=36pt, from=5-1, to=5-5]
		\arrow["\oncell{ {A\y} }"{description, pos=0.3}, draw=none, from=3-1, to=1]
		\arrow["\cong"{description}, draw=none, from=2, to=0]
		\arrow["\oncell{ \y }"{description}, shift right=4, draw=none, from=1, to=2]
	\end{tikzcd}
\]

 \vspace{-3mm}
\caption{Coherence axioms for a strong pseudofunctor.}
\label{fig:coherence-axioms-for-strong-pseudofunctor}
\end{figure*}

\begin{definition}
\label{def:strong-pseudomonad}
Let $(\B, \otimes,
I)$ be a monoidal bicategory. A \emph{left strength} for a pseudomonad 
	$(T, \eta, \mu)$ 
consists of a left strength
	$(t, \x, \y)$ 
for the underlying pseudofunctor, together with 
invertible modifications
\[
	\begin{tikzcd}[
		execute at end picture={
				\foreach \nom in  {A,B,C, D}
	  				{\coordinate (\nom) at (\nom.center);}
				\fill[\strongwcolour,opacity=\opacity] 
		  				(A) -- (B) -- (C) -- (D);
		}
	]
    \alias{A} {A T^2_B}
    \arrow{rr}{A\mu}
    \arrow[swap]{d}{t} 
    & 
    \: 
    & 
    \alias{B} {A T_B} \\
\alias{D} {T_{AT_B}} 
   \arrow{r}[swap]{T_t}
   & 
   {T^2_{AB}} 
   \arrow{r}[swap]{\mu}  
   \arrow[Rightarrow, shorten = 2]{u}{\w\:} 
   & 
   \alias{C} {T_{AB}}
  	    \arrow["{t}", from=1-3, to=2-3]
 \end{tikzcd}
\hspace{6mm}
\begin{tikzcd}[
  	execute at end picture={
  					\foreach \nom in  {A,B,C, D, X, Y, Z}
  		  				{\coordinate (\nom) at (\nom.center);}
  					\fill[\strongzcolour,opacity=\opacity] 
  			  				(X) to[bend right=35] (Z) -- (Y);		  		
  			}
  	]
  	\alias{X} {A B} &&  \alias{Y} {A T_B} \\
  	    && \alias{Z} {T_{A  B}} 
  	    \arrow["{t}", from=1-3, to=2-3]
  	    \arrow[""{name=0, anchor=center, inner sep=0}, "{\eta}"', bend right, from=1-1, to=2-3]
  	    \arrow["{A \eta}", from=1-1, to=1-3]
  	    \arrow["{\z}", shorten <=8pt, shorten >=8pt, Rightarrow, from=1-3, to=0]
  	\end{tikzcd}
  \]
expressing the compatibility of $t$ with the pseudomonad
structure. This is subject to two axioms expressing compatibility with the monad structure and two axioms expressing compatibility for $\x$ with $\w$ and $\z$, respectively (\Cref{fig:other-coherence-axioms-for-strong-pseudomonad}), and two axioms expressing compatibility for $\y$ with $\w$ and $\z$, respectively (\Cref{fig:y-coherence-axioms-for-strong-pseudomonad}).
\end{definition}

\begin{figure}
		\vspace{-5mm}
\[
	\hspace{-4.5cm}
    \begin{minipage}{0.15\textwidth}
	\[ 	  
\begin{tikzcd}[
	 		 	scalenodes = .9,
		    	column sep = 2em, 
		    	row sep = 1em, 
		    	execute at end picture={
					\foreach \nom in  {A,B,C, D, X, Y, Z}
		  				{\coordinate (\nom) at (\nom.center);}
					\fill[\strongwcolour,opacity=\opacity] 
			  				(A) -- (B) -- (C) -- (D);
					\fill[\monadncolour,opacity=\opacity] 
			  				(A) -- (X) -- (B);	  				
		    	}
		    ]
			\alias{A} {AT^2_B} &  \alias{X} {AT_B} \\
			{T_{AT_B}} &  \alias{B} AT_B \\
			\alias{D} {T^2_{AB}} & \alias{C} {T_{AB}}
			\arrow["t"', from=1-1, to=2-1]
			\arrow["t",from=2-2, to=3-2]
			\arrow[""{name=0, anchor=center, inner sep=0}, "Tt"', from=2-1, to=3-1]
			\arrow["\mu"', from=3-1, to=3-2]
			\arrow["{A \mu}"{below}, from=1-1, to=2-2]
			\arrow["{A \eta}"', from=1-2, to=1-1]
			\arrow[""{name=1, anchor=center, inner sep=0}, Rightarrow, no head, from=1-2, to=2-2]
			\arrow["\oncell{\w}"{description, xshift=2mm}, draw=none, from=0, to=2-2]
			\arrow["\oncell{ {A\n} }"{description, xshift=1mm}, draw=none, from=1-1, to=1]
	\end{tikzcd}
	\]
	\vspace{-2mm}
	\[ \vertequals \]
	\vspace{-4mm}
	\[ 		   
	 \begin{tikzcd}[
	 				scalenodes = .9,
			    	column sep = 1.8em, 
			    	row sep = 1em,       
			      		 execute at end picture={
			      		 			\foreach \nom in  {A,B,C, D, X, Y, Z}
			      		   				{\coordinate (\nom) at (\nom.center);}
			      		 			\fill[\strongzcolour,opacity=\opacity] 
			      		 	  				(A) -- (B) -- (C);
			      		 			\fill[\monadncolour,opacity=\opacity] 
			      		 	  				(X) -- (Y) -- (Z);	  				
			      		   }
			      ]
				\alias{A} {AT^2_B} &  \alias{B} {AT_B} \\
				\alias{C} {T_{AT_B}} & \alias{X} {T_{AB}} \\
				\alias{Y} {T^2_{AB}} & \alias{Z} {T_{AB}}
				\arrow["t"', from=1-1, to=2-1]
				\arrow[""{name=0, anchor=center, inner sep=0}, Rightarrow, no head, from=2-2, to=3-2]
				\arrow["T_t"', from=2-1, to=3-1]
				\arrow["\mu"', from=3-1, to=3-2]
				\arrow["t", from=1-2, to=2-2]
				\arrow["\eta"'{yshift=-1mm}, from=2-2, to=3-1]
				\arrow[""{name=1, anchor=center, inner sep=0}, "\eta"{yshift=1mm}, from=1-2, to=2-1]
				\arrow["{A \eta}"', from=1-2, to=1-1]
				\arrow["{\iso}"{description}, draw=none, from=2-1, to=2-2]
				\arrow["\oncell{ \z}"{description, xshift=-1mm}, draw=none, from=1-1, to=1]
				\arrow["\oncell{ \n }"{description}, draw=none, from=3-1, to=0]
	\end{tikzcd}
	\]
    \end{minipage}
\hspace{0mm}
\begin{minipage}{0.1\textwidth}
\[
      \begin{tikzcd}[
      		scalenodes = .9,
      		row sep = 1em, 
      		column sep=1em,
      		execute at end picture={
      		 			\foreach \nom in  {A,B,C, D, X, Y, Z, P, Q}
      		   				{\coordinate (\nom) at (\nom.center);}
      		 			\fill[\strongwcolour,opacity=\opacity] (A) -- (B) -- (D) -- (C); 
      		 			\fill[\strongwcolour,opacity=\opacity] (B) -- (X) -- (Y) -- (Z);	  				
      		 			\fill[\monadmcolour,opacity=\opacity] (X) -- (Y) -- (P) -- (Q);	  				      		 	  		
      		  }
      	]
	\alias{A} {A T^3_B} & \alias{B} {A T^2_B} \\
	{T_{A T^2_B}} && \alias{Z} {A T_B} \\
	\alias{C} {T^2_{AT_B}} & \alias{D} {T_{A T_B}} \\
	\alias{Q} {T^3_{AB}} & \alias{X} {T^2_{A B}} \\
	& \alias{P} {T^2_{AB}} &  \alias{Y} {T_{AB}}
	\arrow["t"', from=1-1, to=2-1]
	\arrow["Tt"', from=2-1, to=3-1]
	\arrow[""{name=0, anchor=center, inner sep=0}, "{T^2t}"', from=3-1, to=4-1]
	\arrow[""{name=1, anchor=center, inner sep=0}, "T\mu"', from=4-1, to=5-2]
	\arrow["\mu"', from=5-2, to=5-3]
	\arrow[""{name=2, anchor=center, inner sep=0}, "\mu", from=4-2, to=5-3]
	\arrow["\mu"', from=4-1, to=4-2]
	\arrow["{A \mu}", from=1-1, to=1-2]
	\arrow["{A \mu}", from=1-2, to=2-3]
	\arrow[""{name=3, anchor=center, inner sep=0}, "t", from=2-3, to=5-3]
	\arrow[""{name=4, anchor=center, inner sep=0}, "t", from=1-2, to=3-2]
	\arrow["\mu", from=3-1, to=3-2]
	\arrow[""{name=5, anchor=center, inner sep=0}, "Tt", from=3-2, to=4-2]
	\arrow["\oncell{ \w }"{description, xshift=-2mm}, draw=none, from=2-1, to=4]
	\arrow["\cong"{description}, shorten <=13pt, draw=none, shorten >=13pt, Rightarrow, from=0, to=5]
	\arrow["\oncell{ \m }"{description, yshift=-1mm}, draw=none, shorten <=13pt, shorten >=13pt, Rightarrow, from=1, to=2]
	\arrow["\oncell{ \w }"{description, phantom}, draw=none, from=4, to=3]
      \end{tikzcd}
     \hspace{-1mm} = \hspace{-1mm} 
      \begin{tikzcd}[
      	     scalenodes = .95,
      		row sep = 1em, 
      		column sep=1.2em,
      	execute at end picture={
      	      		 			\foreach \nom in  {A,B,C, D, X, Y, Z, P, Q, R, S, M}
      	      		   				{\coordinate (\nom) at (\nom.center);}
      	      		 			\fill[\strongwcolour,opacity=\opacity] (R) -- (S) -- (A) -- (M); 
      	      		 			\fill[\strongwcolour,opacity=\opacity] (Q) -- (P) -- (B) -- (A);	  				
      	      		 			\fill[\monadmcolour,opacity=\opacity] (X) -- (Y) -- (P) -- (Q);	  				      		 	  		
      	      		  }
      	 ]
	\alias{X} {A T^3_B} 
	& \alias{Y}  {A T^2_B} \\
	\alias{R} {T_{A T^2_B}} & \alias{Q} {AT^2_B} & \alias{P} {A T_B} \\
	{T^2_{AT_B}} & \alias{S} {T_{AT_B}} \\
	\alias{M} {T^3_{AB}} \\
	& \alias{A} {T^2_{AB}} & \alias{B} {T_{AB}}
	\arrow["t"', from=1-1, to=2-1]
	\arrow["Tt"', from=2-1, to=3-1]
	\arrow["{T^2t}"', from=3-1, to=4-1]
	\arrow["T\mu"', from=4-1, to=5-2]
	\arrow["\mu"', from=5-2, to=5-3]
	\arrow["{A \mu}", from=1-1, to=1-2]
	\arrow[""{name=0, anchor=center, inner sep=0}, "{A \mu}", from=1-2, to=2-3]
	\arrow[""{name=1, anchor=center, inner sep=0}, "t", from=2-3, to=5-3]
	\arrow[""{name=2, anchor=center, inner sep=0}, 
		"{\scriptscriptstyle AT_\mu}"{below, xshift=-2mm, yshift=1mm}, from=1-1, to=2-2]
	\arrow["A\mu"', from=2-2, to=2-3]
	\arrow["t", from=2-2, to=3-2]
	\arrow[""{name=3, anchor=center, inner sep=0}, "{T_t}", from=3-2, to=5-2]
	\arrow["{\scriptscriptstyle T_{A\mu}}"{below, xshift=-2mm, yshift=1mm}, from=2-1, to=3-2]
	\arrow["\cong"{description}, draw=none, from=2-1, to=2-2]
	\arrow["\oncell{ {A\m}}"{description}, draw=none, from=2, to=0]
	\arrow["{ \oncell{T_{\w}}}"{description, pos=0.4}, draw=none, from=3-1, to=3]
	\arrow["\oncell{ \w }"{description}, draw=none, from=3, to=1]
      \end{tikzcd}
      \]
     \end{minipage}
\]
\[
		\hspace{-4mm}
\begin{tikzcd}[
	            		scalenodes = .9,
	            			column sep = 2em, 
	            			row sep = 1em, 
	            		execute at end picture={
	             			\foreach \nom in  {A,B,C, D, X, Y, Z}
	               				{\coordinate (\nom) at (\nom.center);}
	             			\fill[\strongxcolour,opacity=\opacity]
	             	  				(A) -- (B) -- (C); 
	             			\fill[\strongxcolour,opacity=\opacity]
	             	  				(A) -- (X) -- (C);      	  		
	             	  	}
	            	]
	      &[-1em] \alias{A} {T_{IT_A}} &[-1em] \alias{X} {IT^2_A} \\
	        	\alias{B} {T^2_{IA}} && {IT_A} \\
	        	& \alias{C} {T^2_A} & {T_A}
	        	\arrow["{T_t}"', from=1-2, to=2-1]
	        	\arrow["{T^2_\lambda}"', from=2-1, to=3-2]
	        	\arrow["\mu"', from=3-2, to=3-3]
	        	\arrow["t"', from=1-3, to=1-2]
	        	\arrow["I\mu", from=1-3, to=2-3]
	        	\arrow[""{name=0, anchor=center, inner sep=0}, "\lambda", from=2-3, to=3-3]
	        	\arrow[""{name=1, anchor=center, inner sep=0}, "{T_\lambda}"{description}, from=1-2, to=3-2]
	        	\arrow[""{name=2, anchor=center, inner sep=0}, "\lambda"{description}, from=1-3, to=3-2]
	        	\arrow["\oncell{ {T_{\x}} }"{xshift=-1mm, description}, draw=none, from=2-1, to=1]
	        	\arrow["\oncell{ \x }"{description}, draw=none, from=1-2, to=2]
	        	\arrow["\cong"{description, yshift=-3mm, xshift=-1mm}, draw=none, from=2, to=0]
	        \end{tikzcd}
\hspace{-1mm}= \hspace{-1mm}
\begin{tikzcd}[
	        					scalenodes = .9,
	              				column sep = 2em, 
	              				row sep = 1em, 
	                  		execute at end picture={
	                   			\foreach \nom in  {A,B,C, D, X, Y, Z}
	                     				{\coordinate (\nom) at (\nom.center);}
	                   			\fill[\strongwcolour,opacity=\opacity]
	                   	  				(X) -- (Y) -- (B) -- (Z); 
	                   			\fill[\strongxcolour,opacity=\opacity]
	                   	  				(A) -- (B) -- (C);      	  		
	                   	  	}
	         ]
	        	&[-1em] \alias{X} T_{IT_A} &[-0.5em] \alias{Y} {IT^2_A} \\
	        	\alias{Z} {T^2_{IA}} & \alias{A} {T_{AI}} & \alias{B} {IT_A} \\
	        	& {T^2_A} & \alias{C} {T_A}
	        	\arrow[""{name=0, anchor=center, inner sep=0}, "{T_t}"', from=1-2, to=2-1]
	        	\arrow["{T^2_\lambda}"', from=2-1, to=3-2]
	        	\arrow["\mu"', from=3-2, to=3-3]
	        	\arrow["t"'{yshift=-.5mm}, from=1-3, to=1-2]
	        	\arrow[""{name=1, anchor=center, inner sep=0}, "I\mu", from=1-3, to=2-3]
	        	\arrow["\lambda", from=2-3, to=3-3]
	        	\arrow["\mu"{yshift=-.5mm}, from=2-1, to=2-2]
	        	\arrow["t"{above}, from=2-3, to=2-2]
	        	\arrow[""{name=2, anchor=center, inner sep=0}, 
	        					"{T_\lambda}"'{near end, below, xshift=-1mm, yshift=1mm}, 
	        					from=2-2, to=3-3]
	        	\arrow["\cong"{description}, draw=none, from=3-2, to=2-2]
	        	\arrow["\oncell{ \w }"{yshift=-2mm}, draw=none, from=0, to=1]
	        	\arrow["\oncell{ \x }"{description, xshift=1mm, yshift=.5mm}, 
	        					draw=none, from=2, to=2-3]
	        \end{tikzcd}
\hspace{2mm}
\begin{tikzcd}[
	    		scalenodes = .9,
	    		column sep = 1em, 
	    		row sep = 1.8em, 
	    		execute at end picture={
	     			\foreach \nom in  {A,B,C, D, X, Y, Z}
	       				{\coordinate (\nom) at (\nom.center);}
	     			\fill[\strongzcolour,opacity=\opacity]
	     	  				(A) -- (B) -- (C); 
	     	  	}
	    	]
	    	\alias{A} {IT_A} & \alias{B} IA \\
	    	&  A \\
	    	\alias{C} {T_{IA}} & TA
	    	\arrow["I\eta"', from=1-2, to=1-1]
	    	\arrow["t"', from=1-1, to=3-1]
	    	\arrow["T\lambda"', from=3-1, to=3-2]
	    	\arrow["\lambda", from=1-2, to=2-2]
	    	\arrow[""{name=0, anchor=center, inner sep=0}, "\eta", from=2-2, to=3-2]
	    	\arrow[""{name=1, anchor=center, inner sep=0}, "\eta"{description}, from=1-2, to=3-1]
	    	\arrow["\oncell{ \z }"{description}, draw=none, from=1-1, to=1]
	    	\arrow["\cong"{description, yshift=-3mm}, draw=none, from=1, to=0]
	       \end{tikzcd}
\hspace{-1mm}= \hspace{-1mm}
\begin{tikzcd}[
	     			scalenodes = .9,
	     			column sep = 1em, 
	     			row sep = 1.8em, 
	     		execute at end picture={
	      			\foreach \nom in  {A,B,C, D, X, Y, Z}
	        				{\coordinate (\nom) at (\nom.center);}
	      			\fill[\strongxcolour,opacity=\opacity]
	      	  				(A) -- (B) -- (C); 
	    	  }
	     	]
	     	\alias{A} {IT_A} & IA \\
	    	& A \\
	    	\alias{C} {T_{IA}} & \alias{B} TA
	    	\arrow["I\eta"', from=1-2, to=1-1]
	    	\arrow["t"', from=1-1, to=3-1]
	    	\arrow["T\lambda"', from=3-1, to=3-2]
	    	\arrow["\lambda", from=1-2, to=2-2]
	    	\arrow["\eta", from=2-2, to=3-2]
	    	\arrow[""{name=0, anchor=center, inner sep=0}, "\lambda"{description}, from=1-1, to=3-2]
	    	\arrow["\oncell{ \x }"{description}, draw=none, from=3-1, to=0]
	    	\arrow["\cong"{description}, draw=none, from=0, to=1-2]
	       \end{tikzcd}	     
\]	
\vspace{-3mm}
\caption{Coherence axioms for a strong pseudomonad: compatibility with the pseudomonad structure, and relating $\x$ with $\z$ and~$\w$.}
\label{fig:other-coherence-axioms-for-strong-pseudomonad}
\vspace{-2mm}
\end{figure}

\begin{figure*}
\centering
\vspace{-3.5mm}

\[
\begin{tikzcd}[
			scalenodes = .9,
			row sep=1em,
			column sep = 1em,
			execute at end picture={
	    			\foreach \nom in  {A,B,C, D, X, Y, Z, P, Q, R}
	      				{\coordinate (\nom) at (\nom.center);}
	    			\fill[\strongwcolour,opacity=\opacity] (A) -- (B) -- (C) -- (D); 
	    			\fill[\strongycolour,opacity=\opacity] (A) -- (D) -- (X) -- (Y);      	  		
	    			\fill[\strongycolour,opacity=\opacity] (D) -- (P) -- (Q) -- (R);      	  		
			 }
		]
		\alias{A} {(A B) T^2_C} &&& \alias{B} {(AB)T_C} \\
		{A(BT^2_C)} & \alias{D} {T_{(AB)T_C}} & \alias{P} {T^2_{(AB) C}} & \alias{C} {T_{(AB)C}} \\
		\alias{Y} {AT_{BT_C}} & \alias{X} {T_{A(B T_C)}} \\
		{AT^2_{B C}} &  \alias{R} {T_{AT_{B C}}} & \alias{Q} {T^2_{A (B C)}} & {T_{A (B C)}}
		\arrow["{(A B) \mu}", from=1-1, to=1-4]
		\arrow["t", from=1-4, to=2-4]
		\arrow[""{name=0, anchor=center, inner sep=0}, "{T_\alpha}", from=2-4, to=4-4]
		\arrow["\alpha"', from=1-1, to=2-1]
		\arrow[""{name=1, anchor=center, inner sep=0}, "{A t}"', from=2-1, to=3-1]
		\arrow[""{name=2, anchor=center, inner sep=0}, "{A T_t}"', from=3-1, to=4-1]
		\arrow["t"', from=4-1, to=4-2]
		\arrow["{T_t}"', from=4-2, to=4-3]
		\arrow["\mu"', from=4-3, to=4-4]
		\arrow["t"{}, from=3-1, to=3-2]
		\arrow[""{name=3, anchor=center, inner sep=0}, "{T_{A t}}", from=3-2, to=4-2]
		\arrow["t", from=1-1, to=2-2]
		\arrow[""{name=4, anchor=center, inner sep=0}, "T_\alpha", from=2-2, to=3-2]
		\arrow["Tt"', from=2-2, to=2-3]
		\arrow[""{name=5, anchor=center, inner sep=0}, "\mu"{below}, from=2-3, to=2-4]
		\arrow[""{name=6, anchor=center, inner sep=0}, "{T^2_\alpha}"', from=2-3, to=4-3]
		\arrow["\oncell{ \y }"{description}, draw=none, from=1, to=4]
		\arrow["\cong"{description}, draw=none, from=2, to=3]
		\arrow["\cong"{description}, draw=none, from=6, to=0]
		\arrow["\oncell{ \w }"{description}, draw=none, from=1-1, to=5]
		\arrow["\oncell{ {T_{\y}} }"{description, yshift=-1mm}, draw=none, from=4, to=6]
	      \end{tikzcd}
\hspace{-1mm}
	=
	\hspace{-1mm}
\begin{tikzcd}[
	      		scalenodes = .9,
	      		row sep=1em,
	      		column sep = 1em, 
				execute at end picture={
		    			\foreach \nom in  {A,B,C, D, X, Y, Z, P, Q, R}
		      				{\coordinate (\nom) at (\nom.center);}
		    			\fill[\strongwcolour,opacity=\opacity] (A) -- (B) -- (C) -- (D); 
		    			\fill[\strongwcolour,opacity=\opacity] (C) -- (D) -- (X);      	  		
		    			\fill[\strongycolour,opacity=\opacity] (B) -- (C) -- (X) -- (Y);      	  		
			}
	     ]
	     {(A B) T^2_C} &&& \alias{Y} {(AB)T_C} \\
		\alias{A} {A(BT^2_C)} && \alias{B} {A(BT_C)} & {T_{(AB)C}} \\
		{AT_{BT_C}} && \alias{C} {AT_{BC}} \\
		\alias{D} {AT^2_{B C}} & {T_{AT_{B C}}} & {T^2_{A (B C)}} & \alias{X} {T_{A (B C)}}
		\arrow[""{name=0, anchor=center, inner sep=0}, "{(A B) \mu}", from=1-1, to=1-4]
		\arrow["t", from=1-4, to=2-4]
		\arrow["{T_\alpha}", from=2-4, to=4-4]
		\arrow["\alpha"', from=1-1, to=2-1]
		\arrow["{A t}"', from=2-1, to=3-1]
		\arrow["{A T_t}"', from=3-1, to=4-1]
		\arrow["t"', from=4-1, to=4-2]
		\arrow["{T_t}"', from=4-2, to=4-3]
		\arrow["\mu"', from=4-3, to=4-4]
		\arrow[""{name=1, anchor=center, inner sep=0}, "{A(B\mu)}", from=2-1, to=2-3]
		\arrow["\alpha"', from=1-4, to=2-3]
		\arrow[""{name=2, anchor=center, inner sep=0}, "At"', from=2-3, to=3-3]
		\arrow[""{name=3, anchor=center, inner sep=0}, "t", from=3-3, to=4-4]
		\arrow[""{name=4, anchor=center, inner sep=0}, "A\mu"{yshift=-.5mm}, from=4-1, to=3-3]
		\arrow["\oncell{ \y }"{description}, draw=none, from=3-3, to=2-4]
		\arrow["\cong"{description, xshift=2mm}, draw=none, from=0, to=1]
		\arrow["\oncell{ {A\w} }"{description, yshift=1mm}, draw=none, from=3-1, to=2]
		\arrow["\oncell{ \w }"{description}, draw=none, from=4, to=3]
\end{tikzcd}
\hspace{2mm}
\begin{tikzcd}[
				      	column sep =1em,
				      	row sep = 1em,
				     	execute at end picture={
							\foreach \nom in  {A,B,C, D, X, Y, Z}
					 				{\coordinate (\nom) at (\nom.center);} 
							\fill[\strongzcolour,opacity=\opacity] (A) -- (B) -- (C);
							\fill[\strongzcolour,opacity=\opacity] (B) -- (C) -- (X);      	  		
				  	}
	 ]
    {(AB)T_C} && {(AB)C} \\
					\alias{A} {A (BT_C)} && \alias{B} {A (BC)} \\
					\\
					\alias{C} {A T_{BC}} && \alias{X} {T_{A(BC)}}
					\arrow["{(AB)\eta}"', from=1-3, to=1-1]
					\arrow[""{name=0, anchor=center, inner sep=0}, "\alpha"', from=1-1, to=2-1]
					\arrow[""{name=1, anchor=center, inner sep=0}, "{A t}"', from=2-1, to=4-1]
					\arrow["t"', from=4-1, to=4-3]
					\arrow[""{name=2, anchor=center, inner sep=0}, "\alpha", from=1-3, to=2-3]
					\arrow["{A(B\eta)}"', from=2-3, to=2-1]
					\arrow[""{name=3, anchor=center, inner sep=0}, "\eta", from=2-3, to=4-3]
					\arrow[""{name=4, anchor=center, inner sep=0}, "A\eta"{description}, from=2-3, to=4-1]
					\arrow["\oncell{ {A \z} }"{description, yshift=-.5mm}, shift left=3, draw=none, from=1, to=4]
					\arrow["\oncell{\z}"{description,yshift=1mm}, shift right=4, draw=none, from=4, to=3]
					\arrow["\cong"{description, yshift=1mm}, draw=none, from=0, to=2]
	  \end{tikzcd}
\hspace{-1mm}
			=
		\hspace{-1mm}
\begin{tikzcd}[
		      	column sep =.3em,
		      	row sep = 1em,
		      	execute at end picture={
		      	           			\foreach \nom in  {A,B,C, D, X, Y, Z}
		      	             				{\coordinate (\nom) at (\nom.center);}
		      	           			\fill[\strongycolour,opacity=\opacity] (A) -- (B) -- (X) -- (Y); 
		      	           			\fill[\strongzcolour,opacity=\opacity] (A) -- (B) -- (C);      	  		
		      	           	  	}
		 ]
			\alias{A} {(AB)T_C} && \alias{C} {(AB)C} \\
 			{A (BT_C)} & \alias{B} {T_{(AB)C}} & {A (BC)} \\
			\\
			\alias{Y} {A T_{BC}} && \alias{X} {T_{A(BC)}}
			\arrow["{(AB)\eta}"', from=1-3, to=1-1]
			\arrow["\alpha"', from=1-1, to=2-1]
			\arrow[""{name=0, anchor=center, inner sep=0}, "{A t}"', from=2-1, to=4-1]
			\arrow["t"', from=4-1, to=4-3]
			\arrow["\alpha", from=1-3, to=2-3]
			\arrow["\eta", from=2-3, to=4-3]
			\arrow["t"{description}, from=1-1, to=2-2]
			\arrow[""{name=1, anchor=center, inner sep=0}, "\eta"{description}, from=1-3, to=2-2]
			\arrow[""{name=2, anchor=center, inner sep=0}, "T\alpha"{description}, from=2-2, to=4-3]
			\arrow["\oncell{ \z }"{description, yshift=-1mm}, draw=none, from=1-1, to=1]
			\arrow["\oncell{ \y }"{yshift=-2mm}, draw=none, from=0, to=2]
			\arrow["\cong"{description, yshift=-1mm}, draw=none, from=2, to=1]
		\end{tikzcd}
\]
 \vspace{-3.5mm}
\caption{Coherence axioms for a strong pseudomonad: relating $\y$ with $\z$ and $\w$.}
\label{fig:y-coherence-axioms-for-strong-pseudomonad}
\end{figure*}

Extending \Cref{ex:pseudomonad-on-para} and
\Cref{ex:pseudofunctor-on-para-strong}, we obtain the following.
The definitions of $\w$ and $\z$ are similar to those for $\x$ and $\y$.

\begin{example}
	\label{ex:strong-pseudomonad-on-para}
	A strong monad on a symmetric monoidal category $\monoidal{\catC}$ determines a strong pseudomonad on $\Para(\catC)$.
\end{example}

\subsubsection{Note on related work.}
\label{sec:redundancy}
Strengths for pseudomonads were first defined by
Tanaka~\cite{Tanaka2004thesis,Tanaka2006journal} for applications in
categorical universal algebra. We improve on this definition in
several ways. We make conceptual progress by cleanly separating strong
pseudofunctors from strong pseudomonads. We also show two natural axioms are redundant, and hence that only eight axioms suffice for a coherent definition
(\Cref{lem:redundancy} below). Finally, in \Cref{sec:justify} we
bring a new perspective on pseudostrengths in terms of higher monoidal
actions (\cf~\cite{Gambino2021formaltheory}).

In more recent related work, Slattery~\cite{Slattery2023TAC} defines strong
(relative) 2-monads via 2-multicategories. An investigation in this
direction is important but seems orthogonal to the work presented here.

\begin{restatable}{lemma}{redundancy}
  \label{lem:redundancy}
\begin{enumerate}
  \item Given the axioms of
    \Cref{def:strong-pseudomonad}, 
    the modifications $\x$ and $\y$ are suitably compatible with the
    monoidal modification $\mathfrak{l}$.
  \item Given the axioms of
    \Cref{def:strong-pseudomonad},
    the modifications $\z$ and $\w$ are suitably compatible with the
    monad modification $\p$.
  \end{enumerate}
\end{restatable}

Precisely, the two redundant compatibility axioms are as follows: 
\[
\begin{tikzcd}[
    	scale cd= 1, 
    	column sep=1em,
    	execute at end picture={
    		  						\foreach \nom in  {A,B,C, X, Y, Z}
    		  			  				{\coordinate (\nom) at (\nom.center);}
    		  						\fill[\strongxcolour,opacity=\opacity] 
    		  			  				(X) to[bend left = 12] (Y) -- (Z);
    		  			  			\fill[\monoidallcolour,opacity=\opacity] 
    		  			  			  	(A) -- (B) -- (C);
    	}
    ]
	\alias{A} {I(AT_B)} && \alias{B} {(IA)T_B} \\
	\alias{X} {IT_{AB}} & \alias{C} {AT_B} \\
	\alias{Z} {T_{I(AB)}} && {T_{(IA)B}} \\
	& \alias{Y} {T_{AB}}
	\arrow["\alpha"', from=1-3, to=1-1]
	\arrow["It"', from=1-1, to=2-1]
	\arrow["t"', from=2-1, to=3-1]
	\arrow[""{name=0, anchor=center, inner sep=0}, "t", from=1-3, to=3-3]
	\arrow["{T_{\lambda B}}", from=3-3, to=4-2]
	\arrow["{T_{\lambda}}"', from=3-1, to=4-2]
	\arrow[""{name=1, anchor=center, inner sep=0}, "\lambda"{below}', from=1-1, to=2-2]
	\arrow[""{name=2, anchor=center, inner sep=0}, "t"'{swap}, from=2-2, to=4-2]
	\arrow["{\lambda T_B}"{below, xshift=2mm}, from=1-3, to=2-2]
	\arrow[""{name=3, anchor=center, inner sep=0}, 
						"\lambda"{near start, above}, 					
						curve={height=-6pt}, from=2-1, to=4-2]
	\arrow["{ \oncell{\x} }"{description, pos=-0.2}, draw=none, from=3-1, to=3]
	\arrow["\cong"{description}, draw=none, from=2, to=0]
	\arrow["\cong"{description, yshift =1mm}, draw=none, from=3, to=1]
	\arrow["{ \oncell{\montrianglel}  }"{description}, draw=none, from=1, to=1-3]
      \end{tikzcd}
      =
      \begin{tikzcd}[
      		scale cd = 1, 
      		column sep=1em,
	      	execute at end picture={
	 			\foreach \nom in  {A,B,C, D}
	   				{\coordinate (\nom) at (\nom.center);}
	 			\fill[\strongycolour,opacity=\opacity] 
	 	  				(A) -- (B) -- (C) -- (D);			
 	   }
		]
	\alias{A} {I(AT_B)} && \alias{B} {(IA)T_B} \\
	{IT_{AB}} \\
	\alias{D} {T_{I(AB)}} && \alias{C} {T_{(IA)B}} \\
	& \alias{E} {T_{AB}}
	\arrow["\alpha"', from=1-3, to=1-1]
	\arrow["It"', from=1-1, to=2-1]
	\arrow["t"', from=2-1, to=3-1]
	\arrow[""{name=0, anchor=center, inner sep=0}, "t", from=1-3, to=3-3]
	\arrow[""{name=1, anchor=center, inner sep=0}, "{T_{\lambda B}}", from=3-3, to=4-2]
	\arrow[""{name=2, anchor=center, inner sep=0}, "{T_{\lambda}}"', from=3-1, to=4-2]
	\arrow["{T_\alpha}"{description}, from=3-3, to=3-1]
	\arrow["{ {T_\montrianglel}  }"{description}, draw=none, from=2, to=1]
	\arrow["{ \oncell{\y} }"{description, pos=0.4}, draw=none, from=2-1, to=0]
      \end{tikzcd}
\]
\vspace{0mm}
\[
      \begin{tikzcd}[
    	column sep = 2.5em, 
execute at end picture={
 			\foreach \nom in  {A,B,C, D, X, Y, Z, P, Q}
   				{\coordinate (\nom) at (\nom.center);}
 			\fill[\strongwcolour,opacity=\opacity] 
 	  				(X) -- (Z) -- (P) -- (Q);
 			\fill[\monadpcolour,opacity=\opacity] 
 	  				(X) -- (Y) -- (Z);	  				
      	   }
      ]
	\alias{X} {AT^2_B} &  \alias{Y} {AT_B} \\
	{T_{AT_B}} &  \alias{Z} {AT_B} \\
	\alias{Q} {T^2_{AB}} & \alias{P} {T_{AB}}
	\arrow["t"', from=1-1, to=2-1]
	\arrow["{T_t}"', from=2-1, to=3-1]
	\arrow["\mu"', from=3-1, to=3-2]
	\arrow[Rightarrow, no head, from=1-2, to=2-2]
	\arrow["{A T_\eta}"', from=1-2, to=1-1]
	\arrow[""{name=0, anchor=center, inner sep=0}, "t", from=2-2, to=3-2]
	\arrow[""{name=1, anchor=center, inner sep=0}, 
						"{\scriptstyle A\mu}"'{below}, from=1-1, to=2-2]
	\arrow["\oncell{ {A \p}  }"'{yshift=-2mm, xshift=1mm}, shift left=5, draw=none, from=1, to=1-2]
	\arrow["\oncell{ \w }"{description, pos=0.3, yshift=-1mm}, draw=none, from=2-1, to=0]
      \end{tikzcd}
\hspace{1mm}
      =
      \hspace{1mm}
\begin{tikzcd}[
    	column sep = 2.5em, 
execute at end picture={
 			\foreach \nom in  {A,B,C, D, X, Y, Z, P, Q}
   				{\coordinate (\nom) at (\nom.center);}
 			\fill[\strongzcolour,opacity=\opacity] 
 	  				(A) -- (B) -- (C);
 			\fill[\monadpcolour,opacity=\opacity] 
 	  				(B) -- (C) -- (Z);	  				
      	   }      
      ]
	{AT^2_B} & {AT_B} \\
	\alias{A} {T_{AT_B}} &  \alias{B} {T_{AB}} \\
	\alias{C} {T^2_{AB}} &  \alias{Z} {T_{AB}}
	\arrow["{AT_\eta}"', from=1-2, to=1-1]
	\arrow["T_t"', from=2-1, to=3-1]
	\arrow["\mu"', from=3-1, to=3-2]
	\arrow[""{name=0, anchor=center, inner sep=0}, "t"', from=1-1, to=2-1]
	\arrow[""{name=1, anchor=center, inner sep=0}, "t", from=1-2, to=2-2]
	\arrow["{T_{A \eta}}"'{yshift=-.3mm}, from=2-2, to=2-1]
	\arrow[""{name=2, anchor=center, inner sep=0}, 
						"{\scriptscriptstyle T_\eta}"{description}, from=2-2, to=3-1]
	\arrow[""{name=3, anchor=center, inner sep=0}, Rightarrow, no head, from=2-2, to=3-2]
	\arrow["{\iso}"{description}, draw=none, from=0, to=1, yshift=1mm]
	\arrow["\oncell{ {T_{\z}} }"{yshift=0mm, xshift=-3mm, description}, 
					draw=none, from=2-1, to=2]
	\arrow["\oncell{ \p }"{description, yshift=-2mm}, draw=none, from=2, to=3]
      \end{tikzcd}
\]

\subsection{Basic examples of strong pseudomonads}
\label{sec:examples}

In this section we show that several important classes of pseudomonad
are strong in the way one would expect from the categorical setting. Many of the proofs essentially come down to the relevant coherence theorem.

Recall that if $(M, m, e)$  is a monoid in a monoidal category $\monoidal\catC$ then 
	$(-) \tens M$
becomes a monad with unit and multiplication given via $e$ and $m$ (\cf~\Cref{ex:writer-on-cat}).	
This monad is canonically strong, with strength given by the structural isomorphism
	$A \tens (B \tens M) \xra{\iso} (A \tens B) \tens M $. 
Also note that every monad $T$ is strong with respect to the cocartesian structure $(0, +)$, with strength $[T\inl \circ \eta_A, T\inr] : A + TB \to T(A + B)$.
These facts bicategorify. The bicategorical version of a monoid is called a \emph{pseudomonoid}~\cite{Joyal1993, Day1997}, and every pseudomonoid defines a pseudomonad similarly to \Cref{ex:writer-on-cat}.

\begin{restatable}{lemma}{easyexamplesofstrengths}
  \label{res:easy-examples-of-strengths}
  \quad
\begin{enumerate}
\item \label{c:pseudomonad-from-pseudomonoid}
	For any pseudomonoid $(M, m, e, a, l, r)$ on a monoidal bicategory $\monoidal{\baseCat}$
	the pseudomonad $(-) \tens M$ has a strength given by the pseudo-inverse $\psinv\alpha$ of the associator for $\tens$.
\item 
	\label{c:pseudomonads-strong-wrt-plus}
	Every pseudomonad is canonically strong with respect to the cocartesian monoidal structure 
			$(+, 0)$.
\end{enumerate}
\end{restatable}

A pseudomonoid in $(\Cat, \times, 1)$ is 
exactly
a monoidal category, so 			
	\Cref{res:easy-examples-of-strengths}(\ref{c:pseudomonad-from-pseudomonoid}) 
applies in particular to the Writer pseudomonad
(\Cref{ex:writer-on-cat}).
We can also use this lemma to derive a result about pseudomonads on spans.
For any category $\catC$ with pullbacks there exists a bicategory of spans $\Span(\catC)$ similar to that defined in \Cref{sec:examples-of-bicats} for $\Set$.
For $\catC := \Set$, or more generally any \emph{lextensive} category~\cite{Carboni1993}, 
the bicategory $\Span(\catC)$ has finite biproducts---bicategorical products and coproducts which coincide---by \cite[Theorem~6.2]{Lack2010spans}. 
Moreover, by~\cite[Corollary~A.4]{Hermida2000}, every 
		\emph{cartesian monad}
(monad for which the underlying functor preserves pullbacks, and such that every naturality square for $\mu$ and $\eta$ is a pullback square) lifts to a pseudomonad on $\Span(\catC)$.
So we have the following.

\begin{corollary}
	\label{res:pseudomonad-on-Span}
	Any cartesian monad on a lextensive category $\catC$
		(such as~$\Set$) lifts to a strong pseudomonad on $\Span(\catC)$ 
\end{corollary}

The next example covers two cases of importance in the semantics of programming languages. The proof follows essentially immediately from the corresponding categorical facts and the particularly strong form of coherence enjoyed by cartesian closed bicategories
	(see~\cite[Principle 1.3]{Fiore2021}). 

\begin{lemma}
	For any cartesian closed bicategory
		$(\baseCat, \times, 1, {\To})$
		(see \eg~\cite{LICS2019})  
	and objects $S, R \in \baseCat$, there exist strong pseudomonads
		$\exp{S}{(S \times -)}$ (the \emph{state} pseudomonad)
	and 
		$\exp{  (\exp{-}{R})  }{R}$ (the \emph{continuation} pseudomonad).
\end{lemma}

For our final class of examples, recall that every functor $F$ on $\Set$ is canonically strong with respect to the cartesian structure, with 
	$t_{A, B} : A \times FB \to F(A \times B)$
defined by 
	$t_{A, B}(a, w) := F(\lambda b \bind \seq{a, b})(w)$,
and moreover that the same construction makes every monad on $\Set$ strong~\cite[Proposition~3.4]{Moggi1991}. 
A similar fact holds for bicategories; the statement for pseudomonads was first proved by Tanaka~\cite{Tanaka2004thesis}.

\begin{restatable}{proposition}{everymonadstrong}
\label{res:every-pseudomonad-strong}
Every pseudofunctor (resp. pseudomonad) on the 2-category $(\Cat, \times, 1)$ has a canonical choice of strength. 
\end{restatable}

\section{Bistrong, commutative, and monoidal pseudomonads}
\label{sec:symmetry}

Categorically, it is often the case that a monad $T$ supports a strength on both sides, and the two strengths are compatible: $T$ is then called \emph{bistrong} (see~\eg~\cite{McDermott2022}). 
This is the case, for instance, if $T$ has left strength $t$ and the underlying
category is symmetric monoidal, because we can construct a right strength using the symmetry $\beta$:
\begin{equation}
\label{eq:left-strength-to-right-strength}
	T(A) \tens B \xra{\beta} 
	B \tens T(A) \xra{t} 
	T(B \tens A) \xra{T\beta}
	T(A \tens B).
\end{equation}
For a bistrong monad $(T, t, s)$ it makes sense to ask whether the two morphisms below coincide:
\begin{align}
TA \otimes TB \xra{t} T(TA \otimes B) \xra{Ts} T^2(A \otimes B)
  \xra{\mu} T(A \otimes B) \label{eq:1} \\
TA \otimes TB \xra{s} T(A \otimes TB) \xra{Tt} T^2(A \otimes B)   \xra{\mu} T(A \otimes B)\label{eq:2} 
\end{align}
When they do, $T$ is said to be \emph{commutative}
\cite{Kock1970, Kock1972}. Kock showed that, in this case, the map 
	$TA \otimes TB \to T(A \otimes B)$ (defined in either way above) gives $T$ the structure of a
monoidal monad, and conversely that any monoidal monad is in particular
bistrong and commutative. 

We now bicategorify these results. 
We introduce the notion of bistrong pseudomonad in
\Cref{sec:bistr-pseud}. In \Cref{sec:towards-commutative} we discuss
the equivalence of commutative and monoidal pseudomonads, which we connect to existing notions due to Hyland \& Power \cite{Hyland2002}.
Finally, in \Cref{sec:premonoidal-Kleisli-bicats} we show the Kleisli bicategory for a bistrong pseudomonad forms a bicategorical version of a well-known model for effectful call-by-value programs.

\subsection{Bistrong pseudomonads}
\label{sec:bistr-pseud}

A \emph{right strength} for
a pseudomonad consists of a pseudonatural transformation $s_{A, B} :
T(A) \otimes B \to T(A \otimes B)$ equipped with four invertible modifications analogous to $x, y, z, w$ and satisfying corresponding axioms.

Informally, a left strength 
	$t_{A, B} : A \otimes TB \to T(A \otimes B)$ 
and a right strength 
	$s_{A, B} : T(A) \otimes B \to T(A\otimes B)$ 
        are \emph{compatible} if parameters on each side can be
        passed through $T$ in any order. 
 Categorically, one makes this precise by asking that the two obvious maps
         ${(A \otimes TB) \otimes C} \to T(A \otimes (B \otimes C))$
        are equal. 
 For the bicategorical definition, we replace this equation by a coherent isomorphism.
        
\begin{definition}
	\label{def:bistrong-pseudomonad}
  A \emph{bistrong} pseudomonad on a monoidal bicategory 
  	$(\B,\otimes, I)$ is a pseudomonad $T$ equipped with a left strength $t$
  and a right strength $s$, and an invertible modification
\[\begin{tikzcd}[
		column sep = 1.5em,
		row sep = .8em,
		execute at end picture={
			\foreach \nom in  {A,B,C, D, X, Y}
     				{\coordinate (\nom) at (\nom.center);}
   						\fill[\bistrongcolour,opacity=\opacity] (A) -- (B) -- (Y) -- (C) -- (D) -- (X); 
		}
	]
	\: & \alias{A} {(AT_B)C} & \alias{B} {A(T_BC)}  \: &\\
	\alias{X} {T_{AB}C} & \: &  \: & \alias{Y} {AT_{BC}} \\
	\: & \alias{D} {T_{(AB)C}} & \alias{C} {T_{A(BC)}} & \:
	\arrow["\alpha", from=A, to=B]
	\arrow["As", from=B, to=Y]
	\arrow["tC"', from=A, to=X]
	\arrow["s"', from=X, to=D]
	\arrow["{T_\alpha}"', from=D, to=C]
	\arrow[from=Y, to=C, "t"]
	\arrow[Rightarrow, "\b", from=X, to=Y , shorten = 42] 
\end{tikzcd}\]
satisfying the two axioms in \Cref{fig:bistrongaxioms}. 
\end{definition}

\begin{example}[{Extending \Cref{res:easy-examples-of-strengths}}]
	\label{ex:braided-pseudomonoid-pseudomonad}
	If $(\B, \tens, \tensu)$ is a braided monoidal bicategory and $M \in \B$ has the structure of a \emph{braided pseudomonoid} (see~\cite{Day1997}), the pseudomonad
		$(-) \tens M$ 
	is canonically bistrong,
	with $s$ defined using the braiding $\braid$ and $\b$ defined using the pentagonator $\pentagonator$ for $\B$.
The axioms follow by coherence~\cite{Gurski2011,Verdon2017}.
\end{example}

\begin{figure*}[ht]
\vspace{-2mm}
\centering
\[
\begin{minipage}{\textwidth}
\centering
\[
	\hspace{-10mm}
	\begin{tikzcd}[
		row sep=.8em, 
		column sep=1.5em, 
		scalenodes = .95,
		execute at end picture={
		    			\foreach \nom in  {A,B,C, D, X, Y}
		      				{\coordinate (\nom) at (\nom.center);}
		    			\fill[\bistrongcolour,opacity=\opacity] (A) -- (B) -- (C) -- (D); 
		    			\fill[\strongzprimecolour,opacity=\opacity] (X) -- (B) -- (Y);      	  		
		    			\fill[\strongzcolour,opacity=\opacity] (X) -- (Y) -- (C);      	  		
		}
	]
	&& {(AB)C} &&& \alias{X} {A(BC)} \\
	\alias{A} {(AT_B)C} &&& \alias{B} {A(T_BC)} \\
	& {T_{AB}C} &&& \alias{Y} {AT_{BC}} \\
	&& \alias{D} {T_{(AB)C}} &&& \alias{C} {T_{A(BC)}}
	\arrow["\alpha", from=1-3, to=1-6]
	\arrow[""{name=0, anchor=center, inner sep=0}, "{(A\eta)C}"', from=1-3, to=2-1]
	\arrow[""{name=1, anchor=center, inner sep=0}, "tC"', from=2-1, to=3-2]
	\arrow["s"', from=3-2, to=4-3]
	\arrow["T\alpha"', from=4-3, to=4-6]
	\arrow[""{name=2, anchor=center, inner sep=0}, "\eta", from=1-6, to=4-6]
	\arrow["t"{swap}, from=3-5, to=4-6]
	\arrow[""{name=3, anchor=center, inner sep=0}, "A\eta"{description}, from=1-6, to=3-5]
	\arrow[""{name=4, anchor=center, inner sep=0}, "{A(\eta C)}"{description}, from=1-6, to=2-4]
	\arrow["\alpha", from=2-1, to=2-4]
	\arrow["As"', from=2-4, to=3-5]
	\arrow["\oncell{{\z}}"{description}, shorten >=5pt, Rightarrow, draw=none, from=3-5, to=2]
	\arrow["\oncell{{A\z'}}"{yshift=-1mm, description}, draw=none, from=2-4, to=3]
	\arrow["\cong"{description}, draw=none, from=0, to=4]
	\arrow["\oncell{{\b}}"{description, yshift=-1mm}, draw=none, from=1, to=3-5]
	\end{tikzcd}
=
\begin{tikzcd}[
    		row sep=1.9em, 
    		column sep=1.5em, 
    		scalenodes=.95,
		execute at end picture={
		    			\foreach \nom in  {A,B,C, D, X, Y}
		      				{\coordinate (\nom) at (\nom.center);}
		    			\fill[\strongzcolour,opacity=\opacity] (A) -- (B) -- (C);      	  		
		    			\fill[\strongzprimecolour,opacity=\opacity] (A) -- (C) -- (X);      	  		  		
		}		
    ]
	& \alias{A} {(AB)C} &&&&& {A(BC)} \\
	\alias{B} {(AT_B)C} & {} \\
	& \alias{C} {T_{AB}C} && \alias{X} {T_{(AB)C}} &&& {T_{A(BC)}}
	\arrow["\alpha", from=1-2, to=1-7]
	\arrow["{(A\eta)C}"', from=1-2, to=2-1]
	\arrow["tC"', from=2-1, to=3-2]
	\arrow["s"', from=3-2, to=3-4]
	\arrow["T\alpha"', from=3-4, to=3-7]
	\arrow[""{name=0, anchor=center, inner sep=0}, "\eta", from=1-7, to=3-7]
	\arrow[""{name=1, anchor=center, inner sep=0}, "\eta"{}, from=1-2, to=3-4]
	\arrow[""{name=2, anchor=center, inner sep=0}, "{\eta C}"{description}, from=1-2, to=3-2]
	\arrow["\cong"{description}, draw=none, from=1, to=0]
	\arrow["\oncell{{\z'}}"{xshift=-1mm, description}, draw=none, from=3-2, to=1]
	\arrow["\oncell{{\z}C}"{description}, draw=none, from=2-1, to=2]
      \end{tikzcd}
    \]
   \end{minipage}
   \]
\[
    \begin{minipage}{\textwidth} 
    \centering
    \[ 
    \hspace{-7mm}
    \begin{tikzcd}[
    	scalenodes=.95, 
    	column sep=2em,
		row sep = 1em, 
    	execute at end picture={
    			    			\foreach \nom in  {A,B,C, D, E, X, Y, Z, P, Q, R, M}
    			      				{\coordinate (\nom) at (\nom.center);}
    			    			\fill[\bistrongcolour,opacity=\opacity] (A) -- (B) -- (C) -- (D) -- (E);      	  		
    			    			\fill[\bistrongcolour,opacity=\opacity] (E) -- (D) -- (X) -- (Y) -- (Z);      	  	
    			    			\fill[\strongwprimecolour,opacity=\opacity] (B) -- (P) -- (Q) -- (R);      	  	    			 
    			    			\fill[\strongwcolour,opacity=\opacity] (P) -- (Q) -- (M) -- (Y);      	  	    			    	  				   	  					  		
    	}		    	  				
    ]
	\alias{A} {(AT^2_B)C} &[1em] \: & \alias{B} {A(T^2_BC)} && \: \\
	{T_{AT_B}C} &[1em] \: & \alias{C} {AT_{T_BC}} && \alias{R} {A(T_BC)} \\
	\alias{E} {T_{(AT_B)C}} & \alias{D} {T_{A(T_BC)}} & \alias{P} {AT^2_{BC}} & \: \\
	{T_{T_{AB}C}} && \alias{X} {T_{AT_{BC}}} && \alias{Q} {AT_{BC}} \\
	\alias{Z} {T^2_{(AB)C}} && \alias{Y} {T^2_{A(BC)}} \\
	& {T_{(AB)C}} & && \alias{M} {T_{A(BC)}}
	\arrow["\alpha", from=1-1, to=1-3]
	\arrow["tC"', from=1-1, to=2-1]
	\arrow["s"', from=2-1, to=3-1]
	\arrow["{T_\alpha}"', from=3-1, to=3-2]
	\arrow["As", from=1-3, to=2-3]
	\arrow["t"', from=2-3, to=3-2]
	\arrow["{T_{tC}}"', from=3-1, to=4-1]
	\arrow["{T_{s}}"', from=4-1, to=5-1]
	\arrow["{T^2\alpha}", from=5-1, to=5-3]
	\arrow["{T_t}", from=4-3, to=5-3]
	\arrow[""{name=0, anchor=center, inner sep=0}, "\mu"', from=5-1, to=6-2]
	\arrow["T\alpha"{below}, from=6-2, to=6-5]
	\arrow[""{name=1, anchor=center, inner sep=0}, "\mu", from=5-3, to=6-5]
	\arrow["{AT_s}", from=2-3, to=3-3]
	\arrow["t", from=3-3, to=4-3]
	\arrow["A\mu"{below}, from=3-3, to=4-5]
	\arrow["{A(\mu C)}", from=1-3, to=2-5]
	\arrow[""{name=2, anchor=center, inner sep=0}, "As", from=2-5, to=4-5]
	\arrow[""{name=3, anchor=center, inner sep=0}, "t", from=4-5, to=6-5]
	\arrow["{T_{As}}"', from=3-2, to=4-3]
	\arrow["{\small \cong}"{description, yshift=1mm}, draw=none, from=3-2, to=3-3]
	\arrow["\oncell{{T_{\b}}}"{xshift=-2mm, yshift=-1mm, description}, draw=none, from=4-1, to=4-3]
	\arrow["\oncell{{\b}}"{xshift=-1mm, description}, draw=none, from=2-1, to=2-3]
	\arrow["\oncell{{A\w'}}"{yshift=-2mm, xshift=1mm, description}, draw=none, from=2-3, to=2]
	\arrow["\oncell{{\w}}"{description}, draw=none, from=4-3, to=3]
	\arrow["\cong"{description}, draw=none, from=0, to=1]
      \end{tikzcd}
\hspace{1mm}=\hspace{1mm}
\begin{tikzcd}[
		scalenodes=.95, 
		row sep = 1.1em, 
		column sep=1.5em,
    	execute at end picture={
   			\foreach \nom in  {A,B,C, D, E, X, Y, Z, P, Q, R, M, N}
     				{\coordinate (\nom) at (\nom.center);}
   			\fill[\strongwcolour,opacity=\opacity] (A) -- (B) -- (C) -- (D);      	  		
   			\fill[\strongwprimecolour,opacity=\opacity] (X) -- (C) -- (Y) -- (Z) -- (P);      	  	
   			\fill[\bistrongcolour,opacity=\opacity] (B) -- (N) -- (M) -- (Y);      	  	    			  				   	  					  		
    	}		    	  	
	]
	\alias{A} {(AT^2_B)C} && {A(T^2_BC)} \\
	\alias{D} {T_{AT_B}C} && \alias{B} {(AT_B)C} &&[1em] \alias{N} {A(T_BC)} \\
	 {T_{(AT_B)C}} & \alias{X} {T^2_{AB}C} \\
	\alias{P} {T_{T_{AB}C}} && \alias{C} {T_{AB}C} && {AT_{BC}} \\
	\alias{Z} {T^2_{(AB)C}} \\
	&&  \alias{Y} {T_{(AB)C}} && \alias{M} {T_{A(BC)}}
	\arrow["\alpha", from=1-1, to=1-3]
	\arrow["s"', from=2-1, to=3-1]
	\arrow["{T_{tC}}"', from=3-1, to=4-1]
	\arrow[""{name=0, anchor=center, inner sep=0}, "{T_{s}}"', from=4-1, to=5-1]
	\arrow["\mu"', from=5-1, to=6-3]
	\arrow["T\alpha"{swap}, from=6-3, to=6-5]
	\arrow[""{name=1, anchor=center, inner sep=0}, "{A(\mu C)}", from=1-3, to=2-5]
	\arrow["As", from=2-5, to=4-5]
	\arrow["t", from=4-5, to=6-5]
	\arrow[""{name=2, anchor=center, inner sep=0}, "tC"', from=1-1, to=2-1]
	\arrow[""{name=3, anchor=center, inner sep=0}, "{(A\mu)C}"{description}, from=1-1, to=2-3]
	\arrow["\alpha", from=2-3, to=2-5]
	\arrow["{\scriptscriptstyle T_tC}"{right}, from=2-1, to=3-2]
	\arrow["s", from=3-2, to=4-1]
	\arrow[""{name=4, anchor=center, inner sep=0}, "tC"', from=2-3, to=4-3]
	\arrow["s"', from=4-3, to=6-3]
	\arrow[from=3-2, to=4-3]
	\arrow["{\small \cong}"{description, yshift=1mm, xshift=-2mm}, draw=none, from=3-1, to=3-2]
	\arrow["\oncell{ \b }"{description}, draw=none, from=4-3, to=4-5]
	\arrow["\oncell{{\w'}}"{description, yshift=-2mm, xshift=1mm}, draw=none, from=0, to=4-3]
	\arrow["\oncell{{\w C}}"{description, yshift=2mm}, draw=none, from=2, to=4]
	\arrow["\cong"{description}, draw=none, from=3, to=1]
	\end{tikzcd}
      \]
    \end{minipage}
    \]
    \caption{Coherence axioms for a bistrong pseudomonad.}
    \label{fig:bistrongaxioms}
\end{figure*}

\Cref{def:bistrong-pseudomonad} is sufficient to recover the categorical situation: 
if $(\B, \otimes, I)$ is symmetric monoidal and $(T, t)$ is a left-strong
pseudomonad, then the composite pseudonatural transformation
with components as in~\eqref{eq:left-strength-to-right-strength}
can always be given the structure of a right strength for $T$. 

\begin{proposition}
	\label{res:left-strong-on-symmetric-is-bistrong}
  Every left-strong pseudomonad on a symmetric monoidal bicategory
  is bistrong in a canonical way. 
\end{proposition}

\begin{corollary}[{\rm Extending~\Cref{ex:strong-pseudomonad-on-para}}]
	\label{ex:strong-monad-to-bistrong-on-Para}
	If $(T, t)$ is a strong monad on a symmetric monoidal category $(\Ccat, \tens, \tensu)$, the induced pseudomonad on $\Para(\Ccat)$ is canonically bistrong.
\end{corollary}

\subsection{Commutative and monoidal pseudomonads}
\label{sec:towards-commutative}

We now define commutativity for bistrong pseudomonads. 
Following the usual pattern for bicategorification, the
definition is in
terms of an invertible 2-cell between the morphisms defined in
\eqref{eq:1} and \eqref{eq:2}.
Our definition is a straightforward adaptation of Hyland \& Power's \cite[Definition~5]{Hyland2002} to the weaker setting of bistrong pseudomonads on a monoidal bicategory.

\begin{definition}
\label{def:commutative-pseudomonad}
  A commutative pseudomonad on a monoidal bicategory $(\B, \otimes,
  I)$ is a bistrong pseudomonad $(T, \mu, \eta, t, s)$
  equipped with an invertible modification
\[\begin{tikzcd}[row sep = 1.5em, 
		execute at end picture={
			 			\foreach \nom in  {A,B,C, D, X, Y, Z, P, Q}
			   				{\coordinate (\nom) at (\nom.center);}
			 			\fill[\commutecolour,opacity=\opacity] 
			 	  				(X) -- (Y) -- (Z) -- (P);			
			      	   }
	]
	\alias{X} {T_A T_B} & {T_{A T_B}} & \alias{Y}  {T^2_{AB}} \\
	\alias{P} {T_{T_A B}} & {T^2_{AB}} & \alias{Z} {T_{AB}}
	\arrow["t"', from=1-1, to=2-1]
	\arrow["s", from=1-1, to=1-2]
	\arrow[""{name=0, anchor=center, inner sep=0}, "Tt", from=1-2, to=1-3]
	\arrow[""{name=1, anchor=center, inner sep=0}, "Ts"', from=2-1, to=2-2]
	\arrow["\mu"', from=2-2, to=2-3]
	\arrow["\mu", from=1-3, to=2-3]
	\arrow["\commute", shorten <=22pt, shorten >=22pt, Rightarrow, from=0, to=1]
      \end{tikzcd}\]
    subject to coherence axioms as in~\cite[Definition~5]{Hyland2002}. \end{definition}

\begin{example}[{Extending \Cref{ex:braided-pseudomonoid-pseudomonad}}]
	If $(\B, \tens, \tensu)$ is a symmetric monoidal bicategory and $M \in \B$ has the structure of a \emph{symmetric pseudomonoid} (see~\cite{Day1997}), the pseudomonad
		$(-) \tens M$ 
	is canonically commutative,
	with $\commute$ defined using the braiding on $M$ and symmetric structure on $\B$;
	the axioms follow by coherence~\cite{Gurski2013symmetricbicats,Verdon2017}.
\end{example}

\begin{example}[{Extending~\Cref{ex:strong-monad-to-bistrong-on-Para}}]
	If $(T, t)$ is a commutative monad on a symmetric monoidal category $(\Ccat, \tens, \tensu)$, the induced pseudomonad on $\Para(\Ccat)$ is canonically commutative.
\end{example}

With the axioms of \Cref{def:commutative-pseudomonad} we can verify those of a monoidal pseudomonad, and conversely, so Kock's correspondence result (\cite[Theorem~2.3]{Kock1972})
holds at this level.
We begin by defining monoidal pseudomonads. For the definition of monoidal pseudofunctors, transformations, modifications, see~\cite{SchommerPries2009, Cheng2011monbicat}. 

\begin{definition}
  \label{def:monoidalpseudomonad}
A monoidal pseudomonad on a monoidal bicategory $(\B, \otimes, I)$ is
a pseudomonad $(T, \mu, \eta)$ with additional structure:
\begin{itemize}
\item A 1-cell $\iota : I \to T I$, 
	and pseudonatural transformation 
		$\chi : T A \otimes TB \to T(A \otimes B)$
	with three (omitted) invertible modifications making 
$T$ a monoidal pseudofunctor;
\item invertible 2-cells making $\eta$ a monoidal pseudonatural transformation:
  \[
  	\hspace{2mm}
      \begin{tikzcd}[column sep=3em]
	I & {T_I}
	\arrow[""{name=0, anchor=center, inner sep=0}, "\eta", curve={height=-12pt}, from=1-1, to=1-2]
	\arrow[""{name=1, anchor=center, inner sep=0}, "\iota"', curve={height=12pt}, from=1-1, to=1-2]
	\arrow["\:\etaunit", shorten <=5pt, shorten >=5pt, Rightarrow, from=0, to=1]
	\end{tikzcd}
	 \hspace{5mm}
	\begin{tikzcd}[column sep = 2em]
		AB && {T_A T_B} \\
		&& {T_{AB}}
		\arrow["{\eta \eta}", from=1-1, to=1-3]
		\arrow["\chi", from=1-3, to=2-3]
		\arrow[""{name=0, anchor=center, inner sep=0}, curve={height=6pt}, "\eta"', from=1-1, to=2-3]
		\arrow["\etaprod"', shorten <=7pt, shorten >=7pt, Rightarrow, from=1-3, to=0]
	      \end{tikzcd}
    \]
  \item invertible 2-cells making $\mu$ a monoidal pseudonatural transformation:
    \[
    	\hspace{6mm}
      \begin{tikzcd}
    	I & {T_I} & {T^2_I} \\
    	&& {T_I}
    	\arrow["\iota", from=1-1, to=1-2]
    	\arrow["{T_\iota}", from=1-2, to=1-3]
    	\arrow["\mu", from=1-3, to=2-3]
    	\arrow[""{name=0, anchor=center, inner sep=0}, "\iota"', curve={height=6pt}, from=1-1, to=2-3]
    	\arrow["\muunit", shorten <=7pt, shorten >=7pt, Rightarrow, from=1-3, to=0]
    \end{tikzcd}
    \hspace{4mm}
    \begin{tikzcd}
	{T^2_AT^2_B} && {T_AT_B} \\
	{T_{T_AT_B}} & {T^2_{AB}} & {T_{AB}}
	\arrow[""{name=0, anchor=center, inner sep=0}, "\mu\mu", from=1-1, to=1-3]
	\arrow["\chi", from=1-3, to=2-3]
	\arrow["\chi"', from=1-1, to=2-1]
	\arrow[""{name=1, anchor=center, inner sep=0}, "{T_\chi}"', from=2-1, to=2-2]
	\arrow["\mu"', from=2-2, to=2-3]
	\arrow["\muprod", xshift=3mm, yshift = 1mm, shorten <=10pt, shorten >=10pt, Rightarrow, from=1, to=0]
    \end{tikzcd}
      \]
\end{itemize}
The pseudomonad modifications $(\m, \n, \p)$ must then satisfy the
axioms of monoidal modifications and the two pseudomonad laws.
\end{definition}

As is the case for monads, commutative and monoidal structures for pseudomonads are equivalent. We shall make this precise in \Cref{sec:justify}, where we discuss the mathematical context for our definitions; for now we observe a simpler corollary.

\begin{proposition}
\label{res:monoidal-iff-commutative}
For any pseudomonad $T$:
every monoidal structure on $T$ canonically induces a commutative structure on $T$, and every commutative structure on $T$ canonically induces a monoidal structure on $T$.
\end{proposition}

Note that when constructing monoidal structure there is a choice between two structures, since our commutativity is only
pseudo, but these are isomorphic via the modification $\commute$.
  
\subsection{Premonoidal Kleisli bicategories}
\label{sec:premonoidal-Kleisli-bicats}

A generalisation of Moggi's framework, which does not require a monad explicitly in the syntax, is given by
	\emph{Freyd categories}~\cite{Power1999kappa,Levy2003}.
This includes Moggi's approach: the functor
	$\eta \circ (-) : \catC \to \catC_T$,
which describes the interaction between pure programs (interpreted in $\catC$) and effectful ones (interpreted in $\catC_T$), forms a Freyd category.
In this section we study the Kleisli bicategories associated to the structures discussed above, and show they form 
	\emph{Freyd bicategories}~\cite{ACT2023}. 
Thus, the categorical interpretation of call-by-value programs lifts to the bicategorical setting as expected.

\paragraph{Kleisli bicategories.}
If $T$ is a pseudomonad on a bicategory $\baseCat$, the 
	\emph{Kleisli bicategory} $\baseCat_T$ (\eg~\cite{Cheng2003})
has the same objects as $\baseCat$ and hom-categories
$\baseCat_T(A, B) := \baseCat(A, TB)$. 
The identity on $A$ is the
1-cell $\eta_A \in \baseCat(A, TA)$ and the composition of 
	$f \in \baseCat(A, TB)$ 
and 
	$g \in \baseCat(B, TC)$ 
is given by
\[
	A \xra{f} TB \xra{Tg} T^2 C \xra{\mu} TC.
\]
The structural 2-cells $\a, \l, \r$ in $\B_T$ are constructed using
the 2-dimensional structure of the pseudomonad $T$.

\paragraph{Premonoidal structure.} If $\B$ is equipped with a monoidal structure
$(\otimes, I)$, then some of this structure is inherited by $\B_T$
when $T$ is strong. More precisely, if $T$ has a left strength $t$,
then for any object $A \in \B$ the mapping
\begin{equation} \label{eq:left-bowties-defined}
B \overset{f}{\longrightarrow} TB' \quad \longmapsto \quad 
	A \otimes B
		\xra{A \tens f}
		A \tens TB'
		\overset{t}{\longrightarrow}
		T(A \tens B')
\end{equation}
can be extended to a pseudofunctor $\B_T \to \B_T$ denoted 
	$A \ltie -$. 
Similarly, if $T$ is right-strong, then for every object
$A$ we have a pseudofunctor $- \rtie A : \B_T \to \B_T$.  

\begin{restatable}{proposition}{freydstructurefromstrong}
\label{res:freyd-structure}
For a bistrong pseudomonad $(T, s, t)$ on a monoidal bicategory $(\B,
\otimes, I)$ the families of pseudofunctors
$(-\rtie A)$ and $(A \ltie -)$ assemble into a \emph{premonoidal
  structure} on $\B_T$.  Together with the canonical pseudofunctor $\B \to \B_T$, which regards pure morphisms
as effectful ones, they determine a \emph{Freyd bicategory}. 
\end{restatable}

\paragraph{Monoidal Kleisli bicategories.}
When the pseudomonad $T$ is commutative, the premonoidal
structure on $\B_T$ canonically extends to a (pseudo) monoidal structure. The
only missing ingredient is the isomorphism $\phi$ making $\tens$ a pseudofunctor of two arguments. One constructs this using $\commute$, yielding the interchange law below:
\begin{equation*}
	\label{eq:monoidal-interchange-law}
\begin{tikzcd}[column sep = 5em, row sep = 3em]
  	{A \tens B} & {A' \tens B} \\
  	{A \tens B'} & {A' \tens B'}
  	\arrow["{f \ltie B}", from=1-1, to=1-2]
  	\arrow["{A' \rtie g}", from=1-2, to=2-2]
  	\arrow["{A \rtie g}"', from=1-1, to=2-1]
  	\arrow["{f \ltie B'}"', from=2-1, to=2-2]
  	\arrow[""{name=0, anchor=center, inner sep=0}, "{f \tens g}"{description}, from=1-1, to=2-2]
  	\arrow["\phi"', yshift = 0mm, xshift = 1mm, shorten >=10pt, Rightarrow, from=1-2, to=0]
  	\arrow["{\phi^{-1}}", xshift = -1mm, yshift = 0mm, shorten <=10pt, Rightarrow, from=0, to=2-1]
  \end{tikzcd}
\end{equation*}
This gives an isomorphism in~(\ref{eq:weak-interchange}). 
Next we will consider a generalised setting in which $\phi$ is not invertible.

\section{Concurrent pseudomonads}
\label{sec:concurrency}

Concurrent pseudomonads illustrate the expressive power of
2-dimensional category theory. Their definition is unequivocally
2-categorical because, for the first time in this paper, we make use of
non-invertible 2-cells (and so it would not be sufficient to
work with a category `up to isomorphism', as is commonly done).

\subsection{Definition and strength}
\label{sec:definition-of-concurrent-pseudomonads}

\begin{definition}
\label{def:concurrent-pseudomonad}
A \emph{concurrent pseudomonad} on a monoidal bicategory $(\B,
\otimes, I)$ consists of the same data as a monoidal pseudomonad (\Cref{def:monoidalpseudomonad}), with axioms modified as follows:
\begin{itemize}
\item The modification $\muprod$
is no longer required to be invertible;
  \item The composite 2-cells
\[
    \begin{tikzcd}
	{AT^2_B} & {T_AT^2_B} & {T^2_AT^2_B} && {T_AT_B} \\
	&& {T_{T_AT_B}} & {T^2_{AB}} & {T_{AB}}
	\arrow[""{name=0, anchor=center, inner sep=0}, "\mu\mu", from=1-3, to=1-5]
	\arrow["\chi", from=1-5, to=2-5]
	\arrow["\chi"', from=1-3, to=2-3]
	\arrow[""{name=1, anchor=center, inner sep=0}, "{{T_\chi}}"', from=2-3, to=2-4]
	\arrow["\mu"', from=2-4, to=2-5]
	\arrow["{\eta T^2_B}", from=1-1, to=1-2]
	\arrow["{\eta T^2_B}", from=1-2, to=1-3]
	\arrow["\muprod", xshift=3mm,  yshift = 1mm, shorten <=10pt, shorten >=10pt, Rightarrow, from=1, to=0]
      \end{tikzcd} \]
\[ \begin{tikzcd}
	{T^2_AB} & {T^2_AT_B} & {T^2_AT^2_B} && {T_AT_B} \\
	&& {T_{T_AT_B}} & {T^2_{AB}} & {T_{AB}}
	\arrow[""{name=0, anchor=center, inner sep=0}, "\mu\mu", from=1-3, to=1-5]
	\arrow["\chi", from=1-5, to=2-5]
	\arrow["\chi"', from=1-3, to=2-3]
	\arrow[""{name=1, anchor=center, inner sep=0}, "{{T_\chi}}"', from=2-3, to=2-4]
	\arrow["\mu"', from=2-4, to=2-5]
	\arrow["{T^2_A \eta}", from=1-1, to=1-2]
	\arrow["{ T^2_A\eta}", from=1-2, to=1-3]
	\arrow["\muprod", xshift=3mm, yshift = 1mm, shorten <=10pt, shorten >=10pt, Rightarrow, from=1, to=0]
      \end{tikzcd}
\]
    are now required to be invertible.
\end{itemize}
The coherence axioms are the same as for a monoidal pseudomonad.
\end{definition}

There are likely many examples of this structure. We give two simple
examples now, and a more involved example based on game semantics in
the second part of this section. 

\newcommand{\Poset}{\mathbf{Poset}}
\newcommand{\N}{\mathbf{N}}
\begin{example}
  Let $\Poset$ be the 2-category of posets and monotone functions,
  with 2-cells given by the pointwise order on functions. 
  The decreasing natural numbers $\N = (\Nat, \geq)$ form a monoid in
  $\Poset$ under addition, which induces a (strict) 2-monad \[ \N \times - : \Poset \to \Poset.\]
  This 2-monad has a concurrent structure: the monotone function $\max : \N \times \N \to \N$,
   induces a natural transformation $(\N \times A) \times (\N
  \times B) \to \N \times (A \times B)$, with the 2-cell $\muprod$
  representing the fact that
  \[
  \max(n + m, k + l) \leq \max(n, k) + \max(m, l).
\]
Note that we recover an equality if either $n = m = 0$ or $k = l =
0$, giving invertible composite 2-cells as required by
\Cref{def:concurrent-pseudomonad}. 
 \end{example}

\begin{example}
\label{ex:seans-example}
For any non-empty set $\Sigma$ the set of finite strings $\Sigma^*$ is
a monoid in $(\Set, \times, 1)$ and so also in the monoidal category
of sets and relations $(\Rel, \times, 1)$. Now, $\Rel$ is a
(degenerate) bicategory with the 2-cells given by the inclusion of
relations, and the induced writer pseudomonad $(-) \times \Sigma^*$ is
strong but not commutative. It has a concurrent structure with $\chi$
defined by 
\begin{align*}	
	(A \times \Sigma^*) \times (B &\times \Sigma^*) 
		\to \mathcal{P}((A \times B) \times \Sigma^*) \\
(a, u, b, v) &\mapsto \big\{ (a, b, w) \st w \text{ is an interleaving of } u \text{ and } v \big\}
\end{align*}
and $\muprod$ given by the inclusion, which is in general strict.
\end{example}

We verify the following essential result directly.
\begin{proposition}
\label{res:concurrent-implies-bistrong}
	Every concurrent pseudomonad has a canonical bistrong structure.
      \end{proposition}
      
     Note that the invertibility conditions of
      \Cref{def:concurrent-pseudomonad} are unavoidable, to get
      \emph{pseudo} left and right strengths, rather than lax ones. 

The next result shows that \Cref{def:concurrent-pseudomonad} does indeed capture the weak interchange law~(\ref{eq:weak-interchange}).
Lax normal functors are defined like pseudofunctors, without the constraint
 that the compositor 2-cell $\phi$ is invertible (\eg~\cite{Gambino2022}).

\begin{proposition}
	\label{res:concurrent-to-lax-monoidal}
	For any concurrent pseudomonad $T$ on a monoidal bicategory 
		$(\B, \tens, \tensu)$, 
	the families of pseudofunctors $(-\ltie A)$ and $(A \rtie -)$
        in the premonoidal structure of $\B_T$ assemble into 
	a \emph{lax normal functor} $\otimes$ of two arguments.
      \end{proposition}

\subsection{Illustration in concurrent game semantics}
\label{sec:games}

In this section we illustrate concurrent pseudomonads with the
continuation pseudomonad from concurrent games~\cite{cg1}. (Game
semantics plays no role in this paper outside this section.)
Our model is ``truly concurrent'', in the sense that programs are
represented as partially ordered sets of computational events, rather
than as sets of possible traces. This makes the concurrent structure
of our pseudomonad clear. The model is a simplified version of~\cite{cg1}.

\subsubsection{Event structures.}

\newcommand{\Donf}{\mathcal{D}}
\newcommand{\Games}{\mathscr{G}}
A \emph{(deterministic) event structure} is a partially ordered set
of events related by a partial order modelling causal
dependency. Formally it is a partial order $(E, \leq_E)$ 
such that every $e \in E$ depends on finitely many events, \ie~the set 
	$\{e' \mid e' \leq_E e \}$
is finite. Thus a finite, down-closed subset of $E$ represents a possible (partial) execution
of the concurrent process modelled by $E$.

A \emph{map of
  event structures} $(E, \leq_E) \to (D, \leq_D)$ is an injective
function $f : E \to D$ such that if $x \subseteq E$ is down-closed,
then the image $f \, x$ is also down-closed.
The map $f$ can be understood as a simulation of $E$ in $D$, or
in terms of possible execution traces. For example the map in
\Cref{fig:introweak} (in which the arrows are a Hasse diagram for $\leq$)
is valid because every possible execution of the domain is also
an execution in the codomain.

Event structures support a parallel composition operator  $E \otimes
D$ (sometimes $E \parallel D$), defined as the disjoint union of partial
orders. 

\subsubsection{Games and strategies}
In what follows we use somewhat informal language to focus on
illustrating the concepts. A \emph{game} is an event structure $A$ equipped with a polarity function 
	$A \to \{ {\color{blue}{+}}, {\color{red}-} \}$
assigning ``moves'' to the program ($\color{blue}{+}$) and the environment ($\color{red}{-}$). 
The game $A^\perp$ swaps the moves of the program and the environment in $A$: it has the same events, with polarity reversed. 
A strategy over the game $A$ is an event structure $S$ with a projection map 
	$p : S \to A$ 
satisfying a lifting condition which plays no role in this section \cite{cg1}.

There is a bicategory of concurrent games $\Games$ as follows:
\begin{itemize}
\item objects are \emph{negative games}: games whose minimal events are
  all negative (``the environment always acts first'').
\item 1-cells from $A$ to $B$ are negative strategies over the game 
	$A^\perp \otimes B$. Intuitively, these encode a program's moves as a function of the environment's behaviour.
\item 2-cells from a strategy 
	$p : S \to A^\perp \otimes B$ to a strategy 
	$p' : S' \to A^\perp \otimes B$ are maps of event
	structures $f : S \to S'$ which commute with the projections.
\end{itemize}

Strategies are composed using a pullback construction in the category of
event structures and maps. (This is only determined up to isomorphism,
and therefore is only weakly associative.) The identity on a game $A$ is the
\emph{copycat} strategy on $A^\perp \otimes A$, in which every 
environment move is copied by the program.

\subsubsection{A double-negation concurrent pseudomonad.} 
 We can turn a negative game $A$ into a
positive game $\lnot A$ by appending a single minimal positive
move. Similarly we can append a negative move at the beginning of a
positive game $A'$ to get a negative game $\lnot A'$. The induced
operation $\lnot \lnot$ is a pseudomonad on $\Games$, as shown
below. (In each diagram, moves of the strategy are positioned underneath the game to which they project. 
	For each strategy we only display the initial portion of appended
	$\lnot$-moves; the rest follows a copycat strategy.)
\[  
\begin{minipage}{3cm}
\scalebox{.95}{
	\begin{tikzpicture}
    \node (a) at (0.4, 3) {$A$};
    \node (d) at (1.1, 3.1) {$\overset{\eta}{\longrightarrow}$};
    \node (g) at (2, 3) {$\lnot \lnot A$};    
    \node[negnode] (1) at (2, 2.5) {};
    \node[posnode] (2) at (2, 1.6) {};
\draw [bend right=0, strat-causality] (1) to (2);
\end{tikzpicture}
  }
\end{minipage}  
  \hspace{1cm}
\begin{minipage}{3cm}
\scalebox{.95}{
  \begin{tikzpicture}
\node (t) at (4.4, 3) {$\lnot \lnot \lnot \lnot A$};
    \node (s) at (5.4, 3.1) {$\overset{\mu}{\longrightarrow}$};
    \node (sf) at (6.3, 3) {$\lnot \lnot A$};    
    \node[negnode] (3) at (6.3, 2.5) {};
    \node[posnode] (4) at (4.4, 2.2) {};
    \node[negnode] (5) at (4.4, 1.4) {};
    \node[posnode] (6) at (4.4, 0.6) {};
    \node[negnode] (7) at (4.4, -0.2) {};
    \node[posnode] (8) at (6.3, -0.5) {};
\draw [bend right=5, strat-causality] (3) to (4);
    \draw [bend right=0, strat-causality] (4) to (5);
    \draw [bend right=0, strat-causality] (5) to (6);
    \draw [bend right=0, strat-causality] (6) to (7);
    \draw [bend left=5, strat-causality] (7) to (8);
\end{tikzpicture}
 }
\end{minipage}    
\]
The effect of the pseudomonad
$\lnot\lnot$ is to track and make explicit the sequential order of function calls or
argument calls, like in continuation-passing style. 
This pseudomonad has a strength $t$. It also has a transformation
	$\chi$ 
showing that we can represent calls being made in parallel, using the
true concurrency of event structures: 
\[
\hspace{-8mm}
\begin{minipage}{3cm}
	\centering
\scalebox{.95}{
  \begin{tikzpicture}
    \node (t) at (4.4, 3) {$A \otimes \lnot \lnot B$};
    \node (s) at (5.4, 3.1) {$\overset{t}{\longrightarrow}$};
    \node (sf) at (6.5, 3) {$\lnot \lnot (A \otimes B)$};    
    \node[negnode] (3) at (6.1, 2.5) {};
    \node[posnode] (4) at (4.6, 2.2) {};
    \node[negnode] (5) at (4.6, 1.4) {};
    \node[posnode] (6) at (6.1, 1.1) {};

    \draw [bend right=5, strat-causality] (3) to (4);
    \draw [bend right=0, strat-causality] (4) to (5);
    \draw [bend left=5, strat-causality] (5) to (6);
  \end{tikzpicture}
 }
  \vspace{.2mm}
\end{minipage}
\hspace{12mm}
\begin{minipage}{3cm}
\scalebox{.95}{
  \begin{tikzpicture}
    \node (t) at (4.4, 3) {$\lnot \lnot A \otimes \lnot \lnot B$};
    \node (s) at (5.7, 3.1) {$\overset{\chi}{\longrightarrow}$};
    \node (sf) at (6.8, 3) {$\lnot \lnot (A \otimes B)$};    
    \node[negnode] (3) at (6.4, 2.5) {};
    \node[posnode] (4) at (4.8, 2.2) {};
    \node[negnode] (5) at (4.8, 1.4) {};
    \node[posnode] (6) at (3.8, 2.2) {};
    \node[negnode] (7) at (3.8, 1.4) {};
    \node[posnode] (8) at (6.4, 1.1) {};

    \draw [bend right=5, strat-causality] (3) to (4);
    \draw [bend right=0, strat-causality] (4) to (5);
    \draw [bend right=13, strat-causality] (3) to (6);
    \draw [bend left=0, strat-causality] (6) to (7);
    \draw [bend left=10, strat-causality] (5) to (8);
    \draw [out=-35, in=-150, strat-causality] (7) to (8);
  \end{tikzpicture}
  }
  \end{minipage}
  \]

This structure does not make $\lnot\lnot$ commutative, only concurrent. Indeed, one can calculate that the 2-cell 
$
	\muprod \ : \ \mu \circ \lnot\lnot\chi \circ \chi \ \To \ \chi \circ (\mu \tens \mu)
$
is the following \emph{non-invertible} map of strategies:
\[
	\begin{minipage}{4cm}
		\centering
			$\lnot^{(4)} A \tens \lnot^{(4)} B 
				\to \lnot\lnot(A \tens B)$
	\end{minipage}
	\hspace{.8cm}
	\begin{minipage}{4cm}
		\centering
			$\lnot^{(4)} A \tens \lnot^{(4)} B 
				\to \lnot\lnot(A \tens B)$
                              \end{minipage}
\]
\vspace{-2mm}
\[
\hspace{.4cm}
\begin{minipage}{2.5cm}
\centering
\scalebox{.85}{
\begin{tikzpicture}
\node[negnode] (3) at (5.4, 2.5) {};
	    \node[posnode] (4) at (4.4, 2.2) {};
	    \node[negnode] (5) at (4.4, 1.4) {};
	    \node[posnode] (6) at (4.4, 0.6) {};
	    \node[posnode] (4') at (3.4, 2.2) {};
	    \node[negnode] (5') at (3.4, 1.4) {};
	    \node[posnode] (6') at (3.4, 0.6) {};
	    \node[negnode] (7') at (3.4, -0.2) {};	    
	    \node[negnode] (7) at (4.4, -0.2) {};
	    \node[posnode] (8) at (5.4, -0.5) {};
\draw [bend right=5, strat-causality] (3) to (4);
	    \draw [bend right=0, strat-causality] (4) to (5);
	    \draw [bend right=0, strat-causality] (5) to (6);
	    \draw [bend right=0, strat-causality] (6) to (7);
	    \draw [bend left=5, strat-causality] (7) to (8);
\draw [bend right=9, strat-causality] (3) to (4');
	    \draw [bend right=0, strat-causality] (4') to (5');
	    \draw [bend right=0, strat-causality] (5') to (6');
	    \draw [bend right=0, strat-causality] (6') to (7');
	    \draw [bend right=20, strat-causality] (7') to (8);	    
\draw [bend right=0, strat-causality] (5) to (6');
	    \draw [bend right=0, strat-causality] (5') to (6);	    
\end{tikzpicture}
}
\end{minipage}
  \hspace{.5cm}
 \begin{minipage}{1cm}
 \centering
 	 \vspace{1.4cm}
         \begin{tikzcd}[column sep=3em]
           {} \arrow[Rightarrow, thick]{r}{} & {}
           \end{tikzcd}
  \vspace{9mm}
  \end{minipage}
  \hspace{1cm}
\begin{minipage}{4cm}  
\scalebox{.85}{
\begin{tikzpicture}
\node[negnode] (3) at (5.4, 2.5) {};
	    \node[posnode] (4) at (4.4, 2.2) {};
	    \node[negnode] (5) at (4.4, 1.4) {};
	    \node[posnode] (6) at (4.4, 0.6) {};
	    \node[posnode] (4') at (3.4, 2.2) {};
	    \node[negnode] (5') at (3.4, 1.4) {};
	    \node[posnode] (6') at (3.4, 0.6) {};
	    \node[negnode] (7') at (3.4, -0.2) {};	    
	    \node[negnode] (7) at (4.4, -0.2) {};
	    \node[posnode] (8) at (5.4, -0.5) {};
\draw [bend right=5, strat-causality] (3) to (4);
	    \draw [bend right=0, strat-causality] (4) to (5);
	    \draw [bend right=0, strat-causality] (5) to (6);
	    \draw [bend right=0, strat-causality] (6) to (7);
	    \draw [bend left=5, strat-causality] (7) to (8);
\draw [bend right=9, strat-causality] (3) to (4');
	    \draw [bend right=0, strat-causality] (4') to (5');
	    \draw [bend right=0, strat-causality] (5') to (6');
	    \draw [bend right=0, strat-causality] (6') to (7');
	    \draw [bend right=20, strat-causality] (7') to (8);	    
\end{tikzpicture}
}
\end{minipage}
\vspace{2mm}
\]
This makes plain the constraints of a midway synchronization point as
in the left-hand side of \eqref{eq:weak-interchange}, and generalizes
the basic example of \Cref{fig:introweak} to a polarized setting.

In summary, game semantics gives a very concrete illustration of a
concurrent pseudomonad, in which concurrency is modelled by the true
concurrency of event structures.

\section{Formal aspects of strong and monoidal pseudomonads}
\label{sec:justify}
A central challenge in developing higher-categorical definitions is
to identify suitable axioms on 2-cells to ensure coherence. 

In this technical section we justify the definitions in this paper in two ways. First, we lift a correspondence between strengths and certain \emph{actions} from the categorical setting (see~\eg~\cite{McDermott2022}) to the bicategorical one. 
This is important from a semantic perspective, but also yields a form of coherence result.
Second, we show our definitions arise naturally from higher-categorical considerations. 
This is a standard approach to verifying the correctness of a definition: \cf~\eg~\cite{Marmolejo2004,Gambino2021formaltheory}.

\subsection{ Strengths as actions }
\label{sec:actions}
Moggi's \emph{monadic metalanguage}~\cite{Moggi1991} extends the simply-typed $\uplambda$-calculus with explicit monadic types. 
It is modelled by a strong monad on a cartesian (more generally, monoidal) category.
His \emph{computational $\uplambda$-calculus}, on the other hand, has the same types as the simply-typed $\uplambda$-calculus. 
It is modelled by a Freyd category, which can equivalently be defined as an action extending the monoidal structure (see~\cite[B.3]{LevyBook}).
We can see these capture the same notion of program, because giving a left strength for a monad $T$ on 
	$(\Ccat, \otimes, I)$  
is equivalent to giving a left action of 
	$(\Ccat, \otimes, I)$ 
on the Kleisli category $\Ccat_T$ which extends the monoidal structure
(\eg~\cite[Proposition~4.3]{McDermott2022}).

This correspondence also holds bicategorically. For the definition of bicategorical actions, we use~\cite[Definition~19]{ACT2023}.
We first observe that every strong pseudomonad induces an action. 

\begin{restatable}{proposition}{strengthtoleftaction}
	\label{res:action-from-strength}
  Every strong pseudomonad $(T, t)$ on $(\B, \otimes, I)$ induces an action of
  $\B$ on the Kleisli bicategory $\B_T$, where the pseudofunctor 
  	$\act : \B \times \B_T \to \B_T$ 
  is given on objects by 
  	$A \act B = A \otimes B$, 
  and on morphisms as 
  \[
     f \act g \ \ := \ \  
     	\big(
   			A \otimes B 
\xra{f \tens g} A' \otimes TB' 
     		\overset{t}{\longrightarrow} 
     		T(A' \otimes B') 
     	\big)
   \]
   for $f : A \to A'$ and $g : B \to TB'$, with the same  action on 2-cells.
 \end{restatable}

 The action $\act : \B \times \B_T \to \B_T$ of \Cref{res:action-from-strength} extends the canonical action 
 	$\otimes : \B \times \B \to \B$ 
 given by the monoidal structure. Indeed, we have a pseudonatural transformation 
\begin{equation}
   \label{eq:extendingaction}
\begin{tikzcd}
  	{\B \times \B_T} & \B_T \\
  	{\B \times \B} & \B
  	\arrow["\act", from=1-1, to=1-2]
  	\arrow["K"', from=2-2, to=1-2]
  	\arrow["{\B \times K}", from=2-1, to=1-1]
  	\arrow["\tens"', from=2-1, to=2-2]
  	\arrow["\theta", shift right, shorten <=5pt, shorten >=5pt, Rightarrow, from=2-1, to=1-2]
  \end{tikzcd}
\end{equation}
where $K : \B \to \B_T$ is the identity-on-objects pseudofunctor sending 
	$f : A \to A'$ 
to 
	$\eta_{A'}\circ f : A \to TA'$.
Moreover, the two actions $\act$ and $\otimes$ agree on objects, and
the  1-cell
components $\theta_{A, B}$ of the transformation are all the
identity. Such a transformation is known as an \emph{icon}
\cite{Lack2008Icons}. The 2-cell components of $\theta$ are
nontrivial:
for each $f : A \to A'$ and $g : B \to B'$ we have an isomorphism
 \[
\theta_{f, g} : f \act K(g) \overset{\cong}\Longrightarrow K(f
    \otimes g)
  \]
derived from
  the modification $\z$, satisfying the
  coherence laws.

We now prove an equivalence between left
strengths and left actions.
Our correspondence theorem uses the following two categories for a
pseudomonad $T$ on $(\B, \otimes, I)$:
\begin{itemize}
  \item $\mathbf{LeftStr}(T)$, the category whose objects are left
    strengths for $T$, and whose morphisms from $t$ to $t'$ are
    modifications which commute with all the strength data;
  \item $\LeftAct(T)$, the category whose objects are extensions of the canonical action of $\B$ on itself, in the sense they are a \emph{0-strict morphism of actions} as defined in~\cite{ACT2023}, and whose morphisms from
    $(\act, \theta)$ to $(\act', \theta')$ are icons $\act \To \act'$
    which commute with $\theta$ and $\theta'$.     
\end{itemize}

\begin{restatable}{theorem}{leftactionsandleftstrengthsequivalent}
	\label{res:left-actions-and-left-strengths-equivalent}
  For any pseudomonad $T$ on a monoidal bicategory 
  	$(\B, \otimes, I)$, the categories $\mathbf{LeftStr}(T)$ and $\LeftAct(T)$ are equivalent. 
\end{restatable}

This theorem gives a slick way to prove \Cref{res:freyd-structure}, because constructing an action is easier than constructing the strength.
Moreover, \Cref{sec:premonoidal-Kleisli-bicats} suggests the following extension. 

\begin{theorem}
	\label{res:commutative-monads-and-monoidal-structures-equivalent}
	For any pseudomonad $T$ on a monoidal bicategory $(\B, \tens, \tensu)$,
	there is an equivalence of categories between: 
	\begin{enumerate}
	\item 
		Monoidal structures on $\B_T$ extending that on $\B$ analogously to the extension of actions in \Cref{res:left-actions-and-left-strengths-equivalent}; and
	\item 
		Commutative structures on $T$.
	\end{enumerate}
	In each case morphisms are defined analogously to \Cref{res:left-actions-and-left-strengths-equivalent}.
\end{theorem}

On the other hand, every monoidal structure on $T$ determines a monoidal structure on $B_T$ with the same action on objects as in $\B$ and the action on 1-cells given by 
	\[ 
	f \ast g :=
		A \tens B 
		\xra{f \tens g}
		T(A) \tens T(B)
		\xra{\chi_{A, B}}
		T(A \tens B)
	\]
Comparing with the situation for actions, the next result is then as expected. For closely-related results proven using sophisticated strictification techniques, see~\cite{Miranda2024eilenbergmoore,Miranda2023kleisliobjects,Miranda2024kleislistructures}.

\begin{theorem}
	\label{res:monoidal-monads-and-monoidal-structures-equivalent}
	For any pseudomonad $T$ on a monoidal bicategory $(\B, \tens, \tensu)$,
	there is an equivalence of categories between: 
	\begin{enumerate}
	\item 
		Monoidal structures on $\B_T$ extending that on $\B$ analogously to the extension of actions in \Cref{res:left-actions-and-left-strengths-equivalent}; and
	\item 
		Monoidal structures on $T$.
	\end{enumerate}
	In each case morphisms are defined analogously to \Cref{res:left-actions-and-left-strengths-equivalent}.
\end{theorem}

Putting together the preceding two theorems, we obtain the promised equivalence between monoidal and commutative structures.

\begin{theorem}
	\label{res:monoidal-monads-and-commutative-monads-equivalent}
	For any pseudomonad $T$ on a monoidal bicategory $(\B, \tens, \tensu)$,
	there is an equivalence of categories between monoidal structures on $T$ and commutative structures on $T$, where in each case morphisms are defined analogously to \Cref{res:left-actions-and-left-strengths-equivalent}.
\end{theorem}

\subsubsection{Coherence.}
Because they are degenerate tricategories, coherence for monoidal bicategories is a subtle matter. While it is true that in certain freely-generated monoidal bicategories all diagrams of structural 2-cells commute~\cite[Corollary 10.6]{Gurski2013}, one must take care about the basic data one is using. Indeed, in any monoidal bicategory $(\B, \tens, I)$ with non-equal endo-1-cells $a, b : I \to I$ there exist diagrams of structural 2-cells involving $a$ and $b$ which the coherence theorem does not require to commute (see \cite[\S 10.3]{Gurski2013}).

Accordingly, because the identity pseudomonad has canonical strong, monoidal, commutative, and concurrent structures, one cannot hope for every diagram involving such structures and the underlying monoidal bicategory $(\B, \tens, I)$ to commute for any choice of $(\B, \tens, I)$.

Nonetheless, \Cref{res:left-actions-and-left-strengths-equivalent,res:commutative-monads-and-monoidal-structures-equivalent,res:monoidal-monads-and-monoidal-structures-equivalent} may be seen as showing strong, commutative, and monoidal pseudomonads are as coherent as one would expect. Roughly, the argument for strong pseudomonads is as follows.
By 
	\Cref{res:left-actions-and-left-strengths-equivalent} 
every strong pseudomonad is isomorphic to one induced by an action.
But such actions are equivalently `trihomomorphisms' between degenerate tricategories (see~\cite[\S4]{ACT2023}); accordingly, Gurski's coherence theorem~\cite[Corollary 10.15]{Gurski2013} applies. It follows that this coherence applies likewise in the induced strong pseudomonad, and hence in the starting strong pseudomonad.  
Similar remarks hold for the monoidal and commutative cases.

\subsection{Strengths as internal pseudomonads}

We now place our definitions in a wider mathematical context.
We shall show the axioms for strong and monoidal pseudomonads (and hence also for concurrent pseudomonads) arise from standard higher-categorical definitions. 
It follows that our choice of coherence axioms is canonical.

We first recall the 1-dimensional situation. The axioms for strong monads and monoidal monads both arise from the definition of a
	\emph{monad internal to a 2-category $\C$}.
This is defined by taking the categorical definition and replacing the underlying functor $T$ by a 1-cell and the natural transformations $\mu$ and $\eta$ by 2-cells (see~\eg~\cite{Street1972formaltheory}). 
Taking $\C := \Cat$ recovers plain monads. 
Taking the 2-category $\MonCat$ of monoidal categories, lax monoidal functors, and monoidal natural transformations recovers monoidal monads.
For a monoidal category $(\Vcat, \tens, \tensu)$, taking the 2-category $\Vact{\Vcat}$ of $\Vcat$-actions, equivariant functors, and equivariant transformations (as defined in~\eg~\cite{McCrudden2000coalg}) recovers strong monads. 

Just as one can define monads in any 2-category, so one can define pseudomonads in any weak 3-category (known as a \emph{tricategory}~\cite{Gordon1995}): see~\eg~\cite{Lack2000}. 
Our definition of monoidal pseudomonads---and hence concurrent pseudomonads---was carefully chosen to guarantee the following.

\begin{theorem}
	A monoidal pseudomonad such that $\iota$ and $\chi$ are equipped with the structure of an adjoint equivalence is exactly a pseudomonad internal to the tricategory $\MonBicat$ of monoidal bicategories~\cite{Cheng2011monbicat}.
\end{theorem}

To justify strong pseudomonads we need to work a little harder, because we cannot rely on a pre-existing tricategory of actions. 
However, for any monoidal bicategory $(\V, \tens, \tensu)$ we can define a tricategory $\Vact{\V}$ by small adjustments to the definition of $\MonBicat$. 
We sketch the definitions.

The objects of $\Vact{\V}$ are left $\V$-actions. The 1-cells 
	$(\act, \alphatri, \lambdatri) \to (\actstar, \alphastar, \lambdastar)$ are 		
	\emph{equivariant morphisms}, 
which consist of a pseudofunctor $F : \B \to \C$ between the bicategories acted on, a pseudonatural transformation $\actmaptrans$ with components
		$\actmaptrans_{X, B} : X \actstar FB \to F(X \act B)$, 
and invertible modifications $\actmapassoccell$ and $\actmapunitcell$ with components as shown below, subject to an associativity law and two unit laws:
\[\begin{tikzcd}[scalenodes = .9, column sep = 1em]
		{(X \tens Y) \actstar FA} & {X \actstar (Y \actstar FA)} & {X \actstar F(Y \acttri A)} \\
		{F{\big( (X \tens Y) \acttri A \big)}} && {F{\big( X \acttri (Y \acttri A) \big)}}
		\arrow["\actmaptrans", from=1-3, to=2-3]
		\arrow["{X \actstar \actmaptrans}", from=1-2, to=1-3]
		\arrow["\alphastar", from=1-1, to=1-2]
		\arrow["\actmaptrans"', from=1-1, to=2-1]
		\arrow[""{name=0, anchor=center, inner sep=0}, "F\alphatri"', from=2-1, to=2-3]
		\arrow["\:\actmapassoccell", shorten >=3pt, Rightarrow, from=1-2, to=0]
	\end{tikzcd}
\hspace{1mm}
\begin{tikzcd}[scalenodes = .9, column sep = 1em]
		{\tensu \actstar FA} \\
		{F(\tensu \acttri A)} & FA
		\arrow["\actmaptrans"', from=1-1, to=2-1]
		\arrow["F\lambdatri"', from=2-1, to=2-2]
		\arrow[""{name=0, anchor=center, inner sep=0}, "\lambdastar", curve={height=-12pt}, from=1-1, to=2-2]
		\arrow["\actmapunitcell", shorten >=5pt, Rightarrow, from=2-1, to=0]
	\end{tikzcd}
\]

The 2-cells
		$(F, \actmaptrans, \actmapassoccell, \actmapunitcell)
				\to
				(F', \actmaptrans', \actmapassoccell', \actmapunitcell')$
are \emph{equivariant transformations}, 
consisting of a pseudonatural transformation 
		$\acttranstrans : F \To F'$
and an invertible modification  $\acttransmodif$ with components 
		$\acttransmodif : 
			\actmaptrans' \circ (X \actstar \acttranstrans) 
			\To 
			\acttranstrans \circ \actmaptrans$ 
subject to an associativity law and unit law.
The 3-cells 
	$(\acttranstrans, \acttransmodif) \to (\acttranstrans', \acttransmodif')$
are \emph{equivariant modifications}, which are modifications 
	$\actthreecell : \acttranstrans \to \acttranstrans'$
subject to a law relating $\acttransmodif$ and $\acttransmodif'$.

Because the data and axioms is similar to that for $\MonBicat$, it is relatively easy to show $\Vact{\V}$ forms a tricategory (\cf~\cite{Cheng2011monbicat}). 

In the 1-dimensional setting a $\Vcat$-action on a $\Ccat$ is equivalently a strong monoidal functor 
	$\Vcat \to [\Ccat, \Ccat]$ 
into the strict monoidal category of endofunctors on $\Ccat$. 
We verify our definition of $\Vact{\V}$ with a bicategorical version of this result.
To state the proposition, we restrict to equivariant data that strictly preserves the base bicategories: write $\mathbf{LAct}(\B)$ for the bicategory with objects $\V$-actions on $\B$, 1-cells equivariant morphisms with underlying pseudofunctor $\id_{\B}$, and 2-cells equivariant transformations of the form $(\id, \Gamma)$.

\begin{proposition}
	For any monoidal bicategory $(\V, \tens, \tensu)$ and bicategory $\B$,
	the currying biequivalence of \cite[\S1.34]{Street1980}
	lifts to a biequivalence between 
		$\mathbf{LAct}(\B)
			\simeq \MonBicat{\big(\V, \homBicat{\B}{\B}\big)}$.
\end{proposition}

We can now see that strong pseudomonads have a canonical status:
\begin{theorem}
A strong pseudomonad on $\V$ is equivalently a pseudomonad on the canonical action of $\V$ on itself in $\Vact{\V}$.
\end{theorem}

\section{Conclusion}

In this paper we have laid a method for modelling effectful programs
in 2-dimensional categories,
using bicategorical versions of strong and commutative monads 
	(\Cref{sec:strong-pseudomonads,sec:symmetry}). The extra
        structure available in this setting can be used to capture
phenomena that are otherwise invisible (\Cref{sec:concurrency}). 
In doing so, we have brought together observations in concurrency theory (\cf~\cite{Hoare2011,mellies2020concurrent})
with new kinds of models motivated by entirely different concerns (\eg~\cite{cg1,FioreSpecies,Cruttwell2022}).
Our definitions arise as expected from purely category-theoretic
concerns (\Cref{sec:justify}).

Moggi's framework paved the way for understanding effectful
programming from various new perspectives (\eg~\cite{Plotkin2001,Hyland2007,Mellies2012,Katsumata2014}).
We see this paper as the beginning of a fruitful line of
future work, mirroring these developments.
Syntactically, it would be natural to develop the internal languages of the various pseudomonad structures presented here (\cf~\cite{LICS2019}).
Semantically, the development of \Cref{sec:concurrency} suggests
making explicit the 2-dimensional structure implicit in long-standing
models such as those detailed in \Cref{sec:semant-2-dimens}.

The precise structure of the Kleisli bicategory
of a concurrent pseudomonad remains to be understood
(\cf~\cite{Gambino2022}), as are the connections between ``graded
pseudomonads'', strong pseudomonads, and pseudo-distributive laws
(\eg~\cite{Walker2021}).
In another direction, strong monads are
induced by strong adjunctions
\cite{levy2005adjunction,Mellies2012,fiore-munch-curien}, which should
also be generalized to a 2-dimensional setting, to provide a finer
setting for studying concurrency.  An important (and simpler) special
case is the class of dialogue categories
\cite{DBLP:journals/corr/abs-0705-0462,mellies2012game}, which
includes examples based on games.

\paragraph{Acknowledgements.} We have benefited from
        presenting this work at various workshops, and discussing it
        with a number of people. In particular, Flavien Breuvart and
        Dylan McDermott pointed us to the idea of 2-dimensional monads
        for concurrency, and we are grateful for discussions with
        Nathanael Arkor, Cristina Matache, Guy {McCusker}, Tarmo Uustalu, Sean Moss
        (who suggested \Cref{ex:seans-example}), and Sam Staton.
We are grateful to Adrian Miranda for discussions about coherence and strictification for the proof of \Cref{res:monoidal-monads-and-monoidal-structures-equivalent}.
        
        We acknoweldge support from a Royal Society
        University Research Fellowship, ERC grand BLaSt, a Paris Region
        Fellowship co-funded by the European Union (Marie
        Sklodowska-Curie grant agreement 945298), and the Air Force Office of Scientific Research (award number
        FA9550-21-1-0038).

\bibliographystyle{ACM-Reference-Format}
\bibliography{expanded}

\appendix
 
\clearpage

\section{Two redundant axioms for left-strong pseudomonads}
\label{sec:missing-axioms}

The two equations below are derivable from the 8 presented in the main body.
\vspace{0mm}
\[
\begin{tikzcd}[
    	scale cd= 1, 
    	column sep=1em,
    	execute at end picture={
    		  						\foreach \nom in  {A,B,C, X, Y, Z}
    		  			  				{\coordinate (\nom) at (\nom.center);}
    		  						\fill[\strongxcolour,opacity=\opacity] 
    		  			  				(X) to[bend left = 12] (Y) -- (Z);
    		  			  			\fill[\monoidallcolour,opacity=\opacity] 
    		  			  			  	(A) -- (B) -- (C);
    	}
    ]
	\alias{A} {I(AT_B)} && \alias{B} {(IA)T_B} \\
	\alias{X} {IT_{AB}} & \alias{C} {AT_B} \\
	\alias{Z} {T_{I(AB)}} && {T_{(IA)B}} \\
	& \alias{Y} {T_{AB}}
	\arrow["\alpha"', from=1-3, to=1-1]
	\arrow["It"', from=1-1, to=2-1]
	\arrow["t"', from=2-1, to=3-1]
	\arrow[""{name=0, anchor=center, inner sep=0}, "t", from=1-3, to=3-3]
	\arrow["{T_{\lambda B}}", from=3-3, to=4-2]
	\arrow["{T_{\lambda}}"', from=3-1, to=4-2]
	\arrow[""{name=1, anchor=center, inner sep=0}, "\lambda"{below}', from=1-1, to=2-2]
	\arrow[""{name=2, anchor=center, inner sep=0}, "t"'{swap}, from=2-2, to=4-2]
	\arrow["{\lambda T_B}"{below, xshift=2mm}, from=1-3, to=2-2]
	\arrow[""{name=3, anchor=center, inner sep=0}, 
						"\lambda"{near start, above}, 					
						curve={height=-6pt}, from=2-1, to=4-2]
	\arrow["{ \oncell{\x} }"{description, pos=-0.2}, draw=none, from=3-1, to=3]
	\arrow["\cong"{description}, draw=none, from=2, to=0]
	\arrow["\cong"{description, yshift =1mm}, draw=none, from=3, to=1]
	\arrow["{ \oncell{\montrianglel}  }"{description}, draw=none, from=1, to=1-3]
      \end{tikzcd}
      =
      \begin{tikzcd}[
      		scale cd = 1, 
      		column sep=1em,
	      	execute at end picture={
	 			\foreach \nom in  {A,B,C, D}
	   				{\coordinate (\nom) at (\nom.center);}
	 			\fill[\strongycolour,opacity=\opacity] 
	 	  				(A) -- (B) -- (C) -- (D);			
 	   }
		]
	\alias{A} {I(AT_B)} && \alias{B} {(IA)T_B} \\
	{IT_{AB}} \\
	\alias{D} {T_{I(AB)}} && \alias{C} {T_{(IA)B}} \\
	& \alias{E} {T_{AB}}
	\arrow["\alpha"', from=1-3, to=1-1]
	\arrow["It"', from=1-1, to=2-1]
	\arrow["t"', from=2-1, to=3-1]
	\arrow[""{name=0, anchor=center, inner sep=0}, "t", from=1-3, to=3-3]
	\arrow[""{name=1, anchor=center, inner sep=0}, "{T_{\lambda B}}", from=3-3, to=4-2]
	\arrow[""{name=2, anchor=center, inner sep=0}, "{T_{\lambda}}"', from=3-1, to=4-2]
	\arrow["{T_\alpha}"{description}, from=3-3, to=3-1]
	\arrow["{ {T_\montrianglel}  }"{description}, draw=none, from=2, to=1]
	\arrow["{ \oncell{\y} }"{description, pos=0.4}, draw=none, from=2-1, to=0]
      \end{tikzcd}
\]
\vspace{-2mm}
\[
      \begin{tikzcd}[
    	column sep = 2.5em, 
execute at end picture={
 			\foreach \nom in  {A,B,C, D, X, Y, Z, P, Q}
   				{\coordinate (\nom) at (\nom.center);}
 			\fill[\strongwcolour,opacity=\opacity] 
 	  				(X) -- (Z) -- (P) -- (Q);
 			\fill[\monadpcolour,opacity=\opacity] 
 	  				(X) -- (Y) -- (Z);	  				
      	   }
      ]
	\alias{X} {AT^2_B} &  \alias{Y} {AT_B} \\
	{T_{AT_B}} &  \alias{Z} {AT_B} \\
	\alias{Q} {T^2_{AB}} & \alias{P} {T_{AB}}
	\arrow["t"', from=1-1, to=2-1]
	\arrow["{T_t}"', from=2-1, to=3-1]
	\arrow["\mu"', from=3-1, to=3-2]
	\arrow[Rightarrow, no head, from=1-2, to=2-2]
	\arrow["{A T_\eta}"', from=1-2, to=1-1]
	\arrow[""{name=0, anchor=center, inner sep=0}, "t", from=2-2, to=3-2]
	\arrow[""{name=1, anchor=center, inner sep=0}, 
						"{\scriptstyle A\mu}"'{below}, from=1-1, to=2-2]
	\arrow["\oncell{ {A \p}  }"'{yshift=-2mm, xshift=1mm}, shift left=5, draw=none, from=1, to=1-2]
	\arrow["\oncell{ \w }"{description, pos=0.3, yshift=-1mm}, draw=none, from=2-1, to=0]
      \end{tikzcd}
\hspace{1mm}
      =
      \hspace{1mm}
\begin{tikzcd}[
    	column sep = 2.5em, 
execute at end picture={
 			\foreach \nom in  {A,B,C, D, X, Y, Z, P, Q}
   				{\coordinate (\nom) at (\nom.center);}
 			\fill[\strongzcolour,opacity=\opacity] 
 	  				(A) -- (B) -- (C);
 			\fill[\monadpcolour,opacity=\opacity] 
 	  				(B) -- (C) -- (Z);	  				
      	   }      
      ]
	{AT^2_B} & {AT_B} \\
	\alias{A} {T_{AT_B}} &  \alias{B} {T_{AB}} \\
	\alias{C} {T^2_{AB}} &  \alias{Z} {T_{AB}}
	\arrow["{AT_\eta}"', from=1-2, to=1-1]
	\arrow["T_t"', from=2-1, to=3-1]
	\arrow["\mu"', from=3-1, to=3-2]
	\arrow[""{name=0, anchor=center, inner sep=0}, "t"', from=1-1, to=2-1]
	\arrow[""{name=1, anchor=center, inner sep=0}, "t", from=1-2, to=2-2]
	\arrow["{T_{A \eta}}"'{yshift=-.3mm}, from=2-2, to=2-1]
	\arrow[""{name=2, anchor=center, inner sep=0}, 
						"{\scriptscriptstyle T_\eta}"{description}, from=2-2, to=3-1]
	\arrow[""{name=3, anchor=center, inner sep=0}, Rightarrow, no head, from=2-2, to=3-2]
	\arrow["{\iso}"{description}, draw=none, from=0, to=1, yshift=1mm]
	\arrow["\oncell{ {T_{\z}} }"{yshift=0mm, xshift=-3mm, description}, 
					draw=none, from=2-1, to=2]
	\arrow["\oncell{ \p }"{description, yshift=-2mm}, draw=none, from=2, to=3]
      \end{tikzcd}
\]

\section{The data for right-strong pseudomonads}
\label{sec:right-strong-data}

The data of a right-strong pseudomonad is shown below. The 8 axioms are essentially those of a left-strong pseudomonad, with the action of parameters on the left replaced by the the corresponding action on the right so that, for example, $I T_A$ is replaced by $T_A I$ and $\lambda$ is replaced by $\rho$.
\newcommand{\tweakyspace}{  \hspace{6mm}  }
\[
\begin{tikzcd}[
		execute at end picture={
		 			\foreach \nom in  {A,B,C, D, X, Y, Z, P, Q}
		   				{\coordinate (\nom) at (\nom.center);}
		 			\fill[\strongxcolour,opacity=\opacity] 
		 	  				(X) to[curve={height=6pt}] (Y) -- (Z);			
		      	   }
	]
		\alias{X} {T_A I } & \alias{Z} {T_{AI}} \\
		& \alias{Y} {T_A}
		\arrow["s", from=1-1, to=1-2]
		\arrow["{T_\rho}", from=1-2, to=2-2]
		\arrow[""{name=0, anchor=center, inner sep=0}, "\rho"', curve={height=6pt}, from=1-1, to=2-2]
		\arrow["{\x'}"{description}, shorten >=3pt, Rightarrow, no body, from=1-2, to=0]
	\end{tikzcd}
\tweakyspace
\begin{tikzcd}[		
		execute at end picture={
			 			\foreach \nom in  {A,B,C, D, X, Y, Z, P, Q}
			   				{\coordinate (\nom) at (\nom.center);}
			 			\fill[\strongycolour,opacity=\opacity] 
			 	  				(X) -- (Y) -- (Z) -- (P);			
			      	   }
		]
		\alias{X} {(T_A B) C} & {T_A(BC)} & \alias{Y} {T_{A(BC)}} \\
		\alias{P} {T_{AB} C} && \alias{Z} {T_{(AB)C}}
		\arrow["\alpha", from=1-1, to=1-2]
		\arrow["s", from=1-2, to=1-3]
		\arrow["{s C}"', from=1-1, to=2-1]
		\arrow[""{name=0, anchor=center, inner sep=0}, "s"', from=2-1, to=2-3]
		\arrow["{T_\alpha}"', from=2-3, to=1-3]
		\arrow["{\y'}"', shorten >=3pt, Rightarrow, from=1-2, to=0]
	\end{tikzcd}
\]
\vspace{-4mm}
\[
\begin{tikzcd}[
			execute at end picture={
			 			\foreach \nom in  {A,B,C, D, X, Y, Z, P, Q}
			   				{\coordinate (\nom) at (\nom.center);}
			 			\fill[\strongzcolour,opacity=\opacity] 
			 	  				(X) to[curve={height=6pt}] (Y) -- (Z);			
			      	   }
		]
		\alias{X} AB & \alias{Z} {T_A B} \\
		& \alias{Y} {T_{AB}}
		\arrow["{\eta B}", from=1-1, to=1-2]
		\arrow["s", from=1-2, to=2-2]
		\arrow[""{name=0, anchor=center, inner sep=0}, "\eta"', curve={height=6pt}, from=1-1, to=2-2]
		\arrow["{\z'}"', shorten >=3pt, Rightarrow, from=1-2, to=0]
	\end{tikzcd}
\tweakyspace
\begin{tikzcd}[		
			execute at end picture={
				 			\foreach \nom in  {A,B,C, D, X, Y, Z, P, Q}
				   				{\coordinate (\nom) at (\nom.center);}
				 			\fill[\strongwcolour,opacity=\opacity] 
				 	  				(X) -- (Y) -- (Z) -- (P);			
				      	   }
		]
		\alias{X} {T^2_A B} & {T_{T_A B}} & \alias{Y} {T^2_{AB}} \\
		\alias{P} {T_{A} B} && \alias{Z} {T_{AB}}
		\arrow["s", from=1-1, to=1-2]
		\arrow["{T_s}", from=1-2, to=1-3]
		\arrow["\mu", from=1-3, to=2-3]
		\arrow["{\mu B}"', from=1-1, to=2-1]
		\arrow[""{name=0, anchor=center, inner sep=0}, "s"', from=2-1, to=2-3]
		\arrow["{\w'}"', shorten >=3pt, Rightarrow, from=1-2, to=0]
	\end{tikzcd}
\]

 \section{The tricategory of \texorpdfstring{$\V$-actions}{actions}}
 \label{sec:tricategory-of-V-actions}
 We give the coherence axioms for the tricategory $\Vact{\V}$. Throughout this section we fix a monoidal bicategory $(\V, \tens, \tensu)$ with left $\V$-actions
		$(\acttri, \alphatri, \lambdatri) : \V \times \B \to \B$
and
		$(\actstar, \alphastar, \lambdastar) : \V \times \C \to \C$ 
defined as in~\cite{ACT2023}.
Composition and identities are defined using the definition of composition of pseudonatural transformations and modifications.  
Similarly,  the structural transformations and structural modifications are all defined by endowing the corresponding structure in $\Bicat$ with equivariant structure in the obvious way.
In each case, most of the work lies in showing the coherence conditions still hold; the various axioms hold because they hold in $\Bicat$.

\begin{definition}
	An 
		\emph{equivariant morphism}
		$(\acttri, \alphatri, \lambdatri) \to (\actstar, \alphastar, \lambdastar)$
	consists of:
		\begin{enumerate}[wide]
		\item 
			A pseudofunctor $F : \B \to \C$;
		\item 	
			A pseudonatural transformation $\actmaptrans_{X, B} : X \actstar FB \to F(X \act B)$;
		\item 
			Invertible modifications $\actmapassoccell$ and $\actmapunitcell$ as shown below, subject to the associativity law and two unit laws  in~\Cref{fig:V-act-1-cells}:
\[\begin{tikzcd}[column sep = 1em, scalenodes = .9]
				{(X \tens Y) \actstar FA} & {X \actstar (Y \actstar FA)} & {X \actstar F(Y \acttri A)} \\
				{F{\big( (X \tens Y) \acttri A \big)}} && {F{\big( X \acttri (Y \acttri A) \big)}}
				\arrow["\actmaptrans", from=1-3, to=2-3]
				\arrow["{X \actstar \actmaptrans}", from=1-2, to=1-3]
				\arrow["\alphastar", from=1-1, to=1-2]
				\arrow["\actmaptrans"', from=1-1, to=2-1]
				\arrow[""{name=0, anchor=center, inner sep=0}, "F\alphatri"', from=2-1, to=2-3]
				\arrow["\:\actmapassoccell", shorten >=3pt, Rightarrow, from=1-2, to=0]
			\end{tikzcd}
\hspace{1mm}
\begin{tikzcd}[column sep = 1em, scalenodes = .9]
				{\tensu \actstar FA} \\
				{F(\tensu \acttri A)} & FA
				\arrow["\actmaptrans"', from=1-1, to=2-1]
				\arrow["F\lambdatri"', from=2-1, to=2-2]
				\arrow[""{name=0, anchor=center, inner sep=0}, "\lambdastar", curve={height=-12pt}, from=1-1, to=2-2]
				\arrow["\actmapunitcell", shorten >=5pt, Rightarrow, from=2-1, to=0]
			\end{tikzcd}\]
		\end{enumerate}
\end{definition}

\begin{figure*}
	\centering
\begin{minipage}{\textwidth}
\[
		(1) \hspace{3mm}
\begin{tikzcd}[scalenodes = .9, column sep = 2.5em, row sep = 1.2em]
			& {\big( (XY)Z \big) F_A} \\
			{(XY)(Z F_A)} & {\big( X(YZ) \big) F_A} & {F_{((XY)Z)A}} \\
			{X \big( Y (Z F_A) \big)} & {X \big( (YZ) F_A \big)} & {F_{(X(YZ))A}} \\
			{X (Y F_{ZA})} & {X F_{(YZ)A}} \\
			{X F_{Y(ZA)}} && {F_{X( (YZ) A) }} \\
			& {F_{X(Y(ZA))}}
			\arrow[""{name=0, anchor=center, inner sep=0}, "\alphastar"', from=1-2, to=2-1]
			\arrow[""{name=1, anchor=center, inner sep=0}, "{\alpha \actstar F_A}"', from=1-2, to=2-2]
			\arrow["\alphastar"', from=2-2, to=3-2]
			\arrow[""{name=2, anchor=center, inner sep=0}, "{X \actstar \alphastar}"', from=3-2, to=3-1]
			\arrow["\alphastar"', from=2-1, to=3-1]
			\arrow["{X \actstar \actmaptrans}"', from=3-2, to=4-2]
			\arrow[""{name=3, anchor=center, inner sep=0}, "{X \actstar F_{\alphatri}}", from=4-2, to=5-1]
			\arrow["\actmaptrans"', from=4-2, to=5-3]
			\arrow["\actmaptrans"', from=5-1, to=6-2]
			\arrow["{F_{X \acttri \alphatri}}", from=5-3, to=6-2]
			\arrow["\actmaptrans", from=1-2, to=2-3]
			\arrow[""{name=4, anchor=center, inner sep=0}, "{F_{\alpha \acttri A}}", from=2-3, to=3-3]
			\arrow[""{name=5, anchor=center, inner sep=0}, "{F_{\alphatri}}"{description}, from=3-3, to=5-3]
			\arrow["\actmaptrans"{description}, from=2-2, to=3-3]
			\arrow["{X \actstar(Y \actstar \actmaptrans)}"', from=3-1, to=4-1]
			\arrow["{X \actstar \actmaptrans}"', from=4-1, to=5-1]
			\arrow["\cong"{description}, draw=none, from=4-2, to=6-2]
			\arrow["\pentstar"{description}, draw=none, from=0, to=2]
			\arrow["\iso"{description}, draw=none, from=1, to=4]
			\arrow["\actmapassoccell"{description}, draw=none, from=3-2, to=5]
			\arrow["{X \actstar \actmapassoccell}"{description}, draw=none, from=2, to=3]
		\end{tikzcd}
=
\begin{tikzcd}[scalenodes = .9, column sep = 1em, row sep = 1.2em]
			& {(XY)(Z F_A)} & {\big( (XY)Z \big) F_A} \\
			&&& {F_{((XY)Z)A}} \\
			{X \big( Y (Z F_A) \big)} & {(XY) F_{ZA}} & {F_{(XY)(ZA)}} & {F_{(X(YZ))A}} \\
			\\
			{X (Y F_{ZA})} &&& {F_{X( (YZ) A) }} \\
			& {X F_{Y(ZA)}} & {F_{X(Y(ZA))}}
			\arrow["\alphastar"', from=1-3, to=1-2]
			\arrow[""{name=0, anchor=center, inner sep=0}, "\alphastar"', from=1-2, to=3-1]
			\arrow["\actmaptrans"', from=6-2, to=6-3]
			\arrow[""{name=1, anchor=center, inner sep=0}, "{F_{X \acttri \alphatri}}", from=5-4, to=6-3]
			\arrow["\actmaptrans", from=1-3, to=2-4]
			\arrow["{F_{\alpha \acttri A}}", from=2-4, to=3-4]
			\arrow["{F_{\alphatri}}", from=3-4, to=5-4]
			\arrow["{X \actstar(Y \actstar \actmaptrans)}"', from=3-1, to=5-1]
			\arrow["{X \actstar \actmaptrans}"', from=5-1, to=6-2]
			\arrow["{(XY) \actstar \actmaptrans}"{description}, from=1-2, to=3-2]
			\arrow[""{name=2, anchor=center, inner sep=0}, "\alphastar"{description}, from=3-2, to=5-1]
			\arrow[""{name=3, anchor=center, inner sep=0}, "\actmaptrans"{description}, from=3-2, to=3-3]
			\arrow[""{name=4, anchor=center, inner sep=0}, "{F_{\alphatri}}"{description}, from=2-4, to=3-3]
			\arrow[from=3-3, to=6-3]
			\arrow["\actmapassoccell"{description}, draw=none, from=1-3, to=3-3]
			\arrow["{F_{\penttri}}"{description}, draw=none, from=4, to=1]
			\arrow["\iso"{description}, draw=none, from=0, to=2]
			\arrow["\actmapassoccell"{description}, draw=none, from=3, to=6-2]
		\end{tikzcd}
		\]
		\vspace{2mm}		
		\end{minipage}
\begin{minipage}{\textwidth}
		The following two diagrams are equal to the canonical structural isomorphisms:
		\vspace{2mm}
		\end{minipage}
\begin{minipage}{\textwidth}
		\[(2) \hspace{3mm}
\begin{tikzcd}[scalenodes = .9, row sep = 1.5em]
			& {(X I) F_A} & {X(IF_A)} \\
			{X F_A} && {X F_{IA}} & {X F_A} \\
			& {F_{(XI)A}} & {F_{X(IA)}} \\
			{F_{XA}} &&& {F_{XA}}
			\arrow[""{name=0, anchor=center, inner sep=0}, "\alphastar", from=1-2, to=1-3]
			\arrow["{X \actstar \actmaptrans}"', from=1-3, to=2-3]
			\arrow["\actmaptrans"'{}, from=2-3, to=3-3]
			\arrow[""{name=1, anchor=center, inner sep=0}, "\actmaptrans"{description}, from=1-2, to=3-2]
			\arrow[""{name=2, anchor=center, inner sep=0}, "{F_{\alphatri}}"', from=3-2, to=3-3]
			\arrow[""{name=3, anchor=center, inner sep=0}, "{X \acttri F_{\lambdatri}}"', from=2-3, to=2-4]
			\arrow[""{name=4, anchor=center, inner sep=0}, "{X \actstar \lambdastar}", curve={height=-6pt}, from=1-3, to=2-4]
			\arrow["\actmaptrans", from=2-4, to=4-4]
			\arrow[""{name=5, anchor=center, inner sep=0}, "{F_{X \acttri \lambdatri}}", from=3-3, to=4-4]
			\arrow["{\lambda \actstar F_A}"', from=1-2, to=2-1]
			\arrow[""{name=6, anchor=center, inner sep=0}, "\actmaptrans"', from=2-1, to=4-1]
			\arrow["{F_{\lambda \acttri A}}"', from=3-2, to=4-1]
			\arrow[""{name=7, anchor=center, inner sep=0}, "\Id"', from=4-1, to=4-4]
			\arrow["\actmapassoccell"{description}, draw=none, from=0, to=2]
			\arrow["\iso"{description}, draw=none, from=3, to=5]
			\arrow["{X \actstar \actmapunitcell}"{description}, draw=none, from=2-3, to=4]
			\arrow["\iso"{description}, draw=none, from=1, to=6]
			\arrow["{F_{\mtri}}"{description, pos=0.6}, draw=none, from=2, to=7]
\arrow[
                "{\Id}",
                rounded corners,
                to path=
                {
                -- ([yshift=2cm]\tikztostart.center)
                -- ([yshift=2cm]\tikztotarget.center)
                \tikztonodes
                -- (\tikztotarget.north)
                },
                from=2-1, to=2-4
            ]			
			\arrow[""{name=0, anchor=center, inner sep=0}, draw=none, yshift=4mm, "\mstar", from=1-2, to=1-3]                        
		\end{tikzcd}
\hspace{15mm}
(3) \hspace{3mm}
\begin{tikzcd}[scalenodes = .9, row sep = 1.5em]
			{X F_A} & {I(X F_A)} & {(IX)F_A} \\
			& {I F_{XA}} \\
			& {F_{I(XA)}} \\
			{F_{XA}} && {F_{(IX)A}}
			\arrow[""{name=0, anchor=center, inner sep=0}, "\alphastar"', from=1-3, to=1-2]
			\arrow["{I \actstar \actmaptrans}", from=1-2, to=2-2]
			\arrow[""{name=1, anchor=center, inner sep=0}, "\actmaptrans", from=2-2, to=3-2]
			\arrow["\actmaptrans", from=1-3, to=4-3]
			\arrow[""{name=2, anchor=center, inner sep=0}, "{F_{\alphatri}}"{description}, from=4-3, to=3-2]
			\arrow["{F_{\lambdatri}}"{description}, from=3-2, to=4-1]
			\arrow[""{name=3, anchor=center, inner sep=0}, "\lambdastar"', from=1-2, to=1-1]
			\arrow["\actmaptrans"', from=1-1, to=4-1]
			\arrow[""{name=4, anchor=center, inner sep=0}, "\lambdastar"', curve={height=6pt}, from=2-2, to=4-1]
			\arrow[""{name=5, anchor=center, inner sep=0}, "{F_{\lambda \acttri A}}", from=4-3, to=4-1]
			\arrow["{F_{\ltri}}"{description}, draw=none, from=3-2, to=5]
			\arrow["\iso"{description}, draw=none, from=3, to=4]
			\arrow["\actmapassoccell"{description}, draw=none, from=0, to=2]
			\arrow["\actmapunitcell"{description}, shift right, curve={height=-6pt}, draw=none, from=4-1, to=1]
\arrow[
                "{\Id}"',
                rounded corners,
                to path=
                {
                -- ([yshift=.8cm]\tikztostart.center)
                -- ([yshift=.8cm]\tikztotarget.center)
                \tikztonodes
                -- (\tikztotarget.north)
                },
                from=1-3, to=1-1
            ]			
			\arrow[""{name=0, anchor=center, inner sep=0}, draw=none, yshift=.6cm, "\lstar", from=1-3, to=1-1]     			
		\end{tikzcd}
		\]
		\end{minipage}
		\caption{Axioms for equivariant morphisms.}
		\label{fig:V-act-1-cells}
		\noindent\makebox[\linewidth]{\rule{\textwidth}{0.4pt}}
\end{figure*}

\begin{definition}
	An 
		\emph{equivariant 2-cell}
		$(F, \actmaptrans, \actmapassoccell, \actmapunitcell)
				\to
				(F', \actmaptrans', \actmapassoccell', \actmapunitcell')$
	between action morphisms of type 
		$(\act, \alphatri, \lambdatri) \to (\actstar, \alphastar, \lambdastar)$
	consists of:
\begin{enumerate}[wide]
	\item 
		A pseudonatural transformation 
			$\acttranstrans : F \To F'$;
	\item 
		An invertible modification $\acttransmodif$ with components as shown:
\[
		\begin{tikzcd}
			{X F_A} & {X F'_A} \\
			{F_{XA}} & {F'_{XA}}
			\arrow["\actmaptrans"', from=1-1, to=2-1]
			\arrow["{X \actstar \acttranstrans}", from=1-1, to=1-2]
			\arrow["{\actmaptrans'}", from=1-2, to=2-2]
			\arrow["\acttranstrans"', from=2-1, to=2-2]
			\arrow["\acttransmodif"{description}, Rightarrow, from=1-2, to=2-1]
		\end{tikzcd}	
		\]
		subject to the following unit and associativity laws:
	\end{enumerate}
\[
	(1) \hspace{2mm}
\begin{tikzcd}
		{I F_A} & {I F'_A} \\
		{F_{IA}} & {F'_{IA}} \\
		{F_A} & {F'_A}
		\arrow[""{name=0, anchor=center, inner sep=0}, "{I \actstar \acttranstrans}", from=1-1, to=1-2]
		\arrow["\actmaptrans"', from=1-1, to=2-1]
		\arrow[""{name=1, anchor=center, inner sep=0}, "{F_{\lambdatri}}"', from=2-1, to=3-1]
		\arrow["\acttranstrans"', from=3-1, to=3-2]
		\arrow[""{name=2, anchor=center, inner sep=0}, "{F'_{\lambdatri}}"{description}, from=2-2, to=3-2]
		\arrow["{\actmaptrans'}"{description}, from=1-2, to=2-2]
		\arrow[""{name=3, anchor=center, inner sep=0}, "\acttranstrans"{description}, from=2-1, to=2-2]
		\arrow[""{name=4, anchor=center, inner sep=0}, "\lambdastar", shift left, curve={height=-24pt}, from=1-2, to=3-2]
		\arrow["\acttransmodif"{description}, draw=none, from=0, to=3]
		\arrow["\iso"{description}, draw=none, from=1, to=2]
		\arrow["{\actmapunitcell'}"{description}, draw=none, from=2-2, to=4]
	\end{tikzcd}
=
\begin{tikzcd}
		{I F_A} & {I F'_A} \\
		{F_{IA}} \\
		{F_A} & {F'_A}
		\arrow["{I \actstar \acttranstrans}", from=1-1, to=1-2]
		\arrow["\actmaptrans"', from=1-1, to=2-1]
		\arrow["{F_{\lambdatri}}"', from=2-1, to=3-1]
		\arrow["\acttranstrans"', from=3-1, to=3-2]
		\arrow[""{name=0, anchor=center, inner sep=0}, "\lambdastar", shift left, curve={height=-12pt}, from=1-2, to=3-2]
		\arrow[""{name=1, anchor=center, inner sep=0}, "\lambdastar", curve={height=-24pt}, from=1-1, to=3-1]
		\arrow["\actmapunitcell"{description}, draw=none, from=2-1, to=1]
		\arrow["\iso"{description}, draw=none, from=1, to=0]
	\end{tikzcd}
	\]	
\[
	(2)
\begin{tikzcd}
		{(XY)F_A} & {X(YF_A)} & {XF_{YA}} \\
		{(XY)F'_A} & {X(YF'_A)} & {XF'_{YA}} & {F_{X(YA)}} \\
		{F'_{(XY)A}} && {F'_{X(YA)}}
		\arrow[""{name=0, anchor=center, inner sep=0}, "{F'_{\alphatri}}"', from=3-1, to=3-3]
		\arrow["\actmaptrans", from=1-3, to=2-4]
		\arrow["\acttranstrans", from=2-4, to=3-3]
		\arrow["\alphastar", from=1-1, to=1-2]
		\arrow[""{name=1, anchor=center, inner sep=0}, "{X \actstar \actmaptrans}", from=1-2, to=1-3]
		\arrow[""{name=2, anchor=center, inner sep=0}, "{(X \tens Y) \actstar \acttranstrans}"', from=1-1, to=2-1]
		\arrow[""{name=3, anchor=center, inner sep=0}, "{X \actstar (Y \actstar \acttranstrans)}"{description}, from=1-2, to=2-2]
		\arrow["\alphastar"', from=2-1, to=2-2]
		\arrow["{\actmaptrans'}"', from=2-1, to=3-1]
		\arrow[""{name=4, anchor=center, inner sep=0}, "{X \actstar \actmaptrans'}"', from=2-2, to=2-3]
		\arrow["{X \actstar \acttranstrans}"{description}, from=1-3, to=2-3]
		\arrow["{\actmaptrans'}"{description}, from=2-3, to=3-3]
		\arrow["\acttransmodif"{description}, draw=none, from=2-4, to=2-3]
		\arrow["{X \actstar \acttransmodif}"{description}, draw=none, from=1, to=4]
		\arrow["{\actmapassoccell'}"{description}, draw=none, from=2-2, to=0]
		\arrow["\iso"{description}, draw=none, from=2, to=3]
	\end{tikzcd}
	\]
\vspace{-2mm}
\[
		\vertequals
	\]
\vspace{-6mm}
\[
	\begin{tikzcd}
		{(XY)F_A} & {X(YF_A)} & {XF_{YA}} \\
		{(XY)F'_A} & {F_{(XY)A}} && {F_{X(YA)}} \\
		{F'_{(XY)A}} && {F'_{X(YA)}}
		\arrow["{F'_{\alphatri}}"', from=3-1, to=3-3]
		\arrow[""{name=0, anchor=center, inner sep=0}, "\actmaptrans", from=1-3, to=2-4]
		\arrow[""{name=1, anchor=center, inner sep=0}, "\acttranstrans", from=2-4, to=3-3]
		\arrow["\alphastar", from=1-1, to=1-2]
		\arrow["{X \actstar \actmaptrans}", from=1-2, to=1-3]
		\arrow["{(X \tens Y) \actstar \acttranstrans}"', from=1-1, to=2-1]
		\arrow["{\actmaptrans'}"', from=2-1, to=3-1]
		\arrow[""{name=2, anchor=center, inner sep=0}, "\actmaptrans", from=1-1, to=2-2]
		\arrow[""{name=3, anchor=center, inner sep=0}, "\acttranstrans", from=2-2, to=3-1]
		\arrow["{F_{\alphatri}}"{description}, from=2-2, to=2-4]
		\arrow["\iso"{description}, draw=none, from=3, to=1]
		\arrow["\actmapassoccell"{description}, draw=none, from=2, to=0]
		\arrow["\acttransmodif"{description}, draw=none, from=2, to=3]
	\end{tikzcd}
	\]
\end{definition}

\begin{definition}
	Let 
		$ (F, \actmaptrans, \actmapassoccell, \actmapunitcell) 
					\to (F', \actmaptrans', \actmapassoccell', \actmapunitcell')
					: (\acttri, \alphatri, \lambdatri) 
					\to (\actstar, \alphastar, \lambdastar)$
	be equivariant 1-cells related by equivariant 2-cells
		$(\acttranstrans, \acttransmodif),
				(\acttranstrans', \acttransmodif') 
			: (F, \actmaptrans, \actmapassoccell, \actmapunitcell) 
						\to (F', \actmaptrans', \actmapassoccell', \actmapunitcell')$.
	An 
		\emph{equivariant 3-cell}
	$(\acttranstrans, \acttransmodif)
		\to (\acttranstrans', \acttransmodif')$
	consists of a modification 
		$\actthreecell : \acttranstrans \to \acttranstrans'$
	such that
\[ 
\begin{tikzcd}[column sep=5em]
		{X F_A} & {X F'_A} \\
		{F_{XA}} & {F'_{XA}}
		\arrow[""{name=0, anchor=center, inner sep=0}, "{X \actstar \acttranstrans}", curve={height=-12pt}, from=1-1, to=1-2]
		\arrow["\actmaptrans"', from=1-1, to=2-1]
		\arrow["{\actmaptrans'}", from=1-2, to=2-2]
		\arrow[""{name=1, anchor=center, inner sep=0}, "{\acttranstrans'}"', curve={height=18pt}, from=2-1, to=2-2]
		\arrow[""{name=2, anchor=center, inner sep=0}, "{X \actstar \acttranstrans'}"', curve={height=12pt}, from=1-1, to=1-2]
		\arrow["{X \actstar \actthreecell}", shorten <=3pt, shorten >=3pt, Rightarrow, from=0, to=2]
		\arrow["{\acttransmodif'}"{description}, draw=none, from=1, to=2]
	\end{tikzcd}
=
\begin{tikzcd}[column sep = 5em]
		{X F_A} & {X F'_A} \\
		{F_{XA}} & {F'_{XA}}
		\arrow[""{name=0, anchor=center, inner sep=0}, "{X \actstar \acttranstrans}", curve={height=-12pt}, from=1-1, to=1-2]
		\arrow["\actmaptrans"', from=1-1, to=2-1]
		\arrow["{\actmaptrans'}", from=1-2, to=2-2]
		\arrow[""{name=1, anchor=center, inner sep=0}, "{\acttranstrans'}"', curve={height=18pt}, from=2-1, to=2-2]
		\arrow[""{name=2, anchor=center, inner sep=0}, curve={height=-18pt}, from=2-1, to=2-2]
		\arrow["\actthreecell", shorten <=5pt, shorten >=5pt, Rightarrow, from=2, to=1]
		\arrow["\acttransmodif"{description}, draw=none, from=0, to=2]
	\end{tikzcd}
	\]
\end{definition}

 \section{Coherence axioms for commutative pseudomonads}
 \label{sec:hyland-power-axioms}

 We collect the axioms for the commutativity modification $\commute$ of a commutative pseudomonad. These are obtained from Hyland \& Power's axioms~\cite{Hyland2002} by explicitly adding in the structural isomorphisms of the monoidal bicategory and the bistrong pseudomonad. 
The axioms are shown in \Cref{fig:hyland-power-axioms-one,fig:hyland-power-axioms-two}.
We have numbered them so they match \cite[Definition~5]{Hyland2002}. 

\newcommand{\axiomscalenodes}{.95}

\begin{figure*}

\[
	(1) \hspace{3mm}
\begin{tikzcd}[scalenodes = \axiomscalenodes, column sep = 1em, row sep = 2em,
		execute at end picture={
			 			\foreach \nom in  {A3, A4, B1, B2, B5, C2, C3, C4, D3, D5}
			   				{\coordinate (\nom) at (\nom.center);}  												 	  				
			 			\fill[\commutecolour,opacity=\opacity] 
			 	  				(B2) -- (A3) -- (A4) -- (B5) -- (C4) -- (C3);											 	  								 	  				
			 			\fill[\bistrongcolour,opacity=\opacity] 
			 	  				(C2) -- (C4) -- (D5) -- (D3);				 	  				
				}		
	]
		&& \alias{A3} {T_{(AB)T_C}} &  \alias{A4}{T^2_{(AB)C}} && {T_{(AB)C}} \\
		\alias{B1} {(AT_B)T_C} & \alias{B2} {T_{AB}T_C} &&& \alias{B5} {T_{(AB)C}} \\
		&\alias{C2}  {T_{(AT_B)C}} & \alias{C3} {T_{T_{AB}}C} &  \alias{C4} {T^2_{(AB)C}} && {T_{A(BC)}} \\
		&& \alias{D3} {T_{A(T_BC})} & {T_{AT_{BC}}} & \alias{D5} {T^2_{A(BC)}}
		\arrow[""{name=0, anchor=center, inner sep=0}, "{tT_C}", from=2-1, to=2-2]
		\arrow["s", from=2-2, to=1-3]
		\arrow[""{name=1, anchor=center, inner sep=0}, "{T_t}", from=1-3, to=1-4]
		\arrow["{T_\alpha}", from=2-5, to=3-6]
		\arrow["t"{description}, from=2-2, to=3-3]
		\arrow[""{name=2, anchor=center, inner sep=0}, "{T_s}"', from=3-3, to=3-4]
		\arrow[""{name=3, anchor=center, inner sep=0}, "\mu"{description}, from=3-4, to=2-5]
		\arrow["\mu", from=1-4, to=2-5]
		\arrow[""{name=4, anchor=center, inner sep=0}, "{T^2_\alpha}", from=3-4, to=4-5]
		\arrow[""{name=5, anchor=center, inner sep=0}, "\mu"{description}, from=4-5, to=3-6]
		\arrow["t"', from=2-1, to=3-2]
		\arrow[""{name=6, anchor=center, inner sep=0}, "{T_{tC}}"', from=3-2, to=3-3]
		\arrow[""{name=7, anchor=center, inner sep=0}, "{T_\alpha}"', from=3-2, to=4-3]
		\arrow["{T_{As}}"', from=4-3, to=4-4]
		\arrow["{T_t}"', from=4-4, to=4-5]
		\arrow["{T^2_{\alpha}}", from=1-4, to=1-6]
		\arrow["\mu", from=1-6, to=3-6]
		\arrow["\cong"{description}, draw=none, from=1-6, to=2-5]
		\arrow["\commute", draw=none, from=1, to=2]
		\arrow["\cong"{description}, draw=none, from=0, to=6]
		\arrow["{T_{\bistrong}}"{description}, draw=none, from=7, to=4]
		\arrow["\cong"{description}, draw=none, from=3, to=5]
	\end{tikzcd}
=
\begin{tikzcd}[scalenodes = .9, column sep = 1em, row sep = 1.5em,
		execute at end picture={
			 			\foreach \nom in  {A1, A2, A3, A4, B1, B2, B3, B4, B5, C2, C3, C4, D1, D2, D5, D6, E1, E2, E3, E4, E5, F3, F4}
			   				{\coordinate (\nom) at (\nom.center);}  												 	  				
			 			\fill[\strongycolour,opacity=\opacity] 
			 	  				(A3) -- (A4) -- (B5) -- (B3);											 	  								 	  				
			 			\fill[\strongwcolour,opacity=\opacity] 
			 	  				(C4) -- (B4) -- (B5) -- (D6) -- (D5);											 	  				 	  				
			 			\fill[\commutecolour,opacity=\opacity] 
			 	  				(D2) -- (C3) -- (C4) -- (D5) -- (E4) -- (E3);										 	  				
			 			\fill[\bistrongcolour,opacity=\opacity] 
			 	  				(D1) -- (B2) -- (A3) -- (C3) -- (D2);		
			 			\fill[\strongycolour,opacity=\opacity] 
			 	  				(D1) -- (D2) -- (E3) -- (F3);	
			 			\fill[\strongwcolour,opacity=\opacity] 
			 	  				(D5) -- (D6) -- (E5) -- (F4) -- (E4);	 	  							 	  				
				}			
	]
		\alias{A1} & \alias{A2} & \alias{A3} {T_{(AB)T_C}} & \alias{A4} {T^2_{(AB)C}} \\
		\alias{B1} & \alias{B2} {T_{AB}T_C} &  \alias{B3} {T_{A(BT_C)}} & \alias{B4} {T_{AT_{BC}}} & \alias{B5} {T_{(AB)C}} \\
		&& \alias{C3} {AT_{BT_C}} & \alias{C4} {AT^2_{BC}} \\
		\alias{D1} {(AT_B)T_C} & \alias{D2} {A(T_BT_C)} &&& \alias{D5} {AT_{BC}} & \alias{D6} {T_{A(BC)}} \\
		\alias{E1} & \alias{E2} {T_{(AT_B)C}} & \alias{E3} {AT_{T_B C}} & \alias{E4} {AT^2_{BC}} && \alias{E5} {T^2_{A(BC)}} \\
		&& \alias{F3} {T_{A(T_BC})} & \alias{F4} {T_{AT_{BC}}}
		\arrow["{tT_C}", from=4-1, to=2-2]
		\arrow["s", from=2-2, to=1-3]
		\arrow[""{name=0, anchor=center, inner sep=0}, "{T_t}", from=1-3, to=1-4]
		\arrow[""{name=1, anchor=center, inner sep=0}, "\mu"', from=5-6, to=4-6]
		\arrow[""{name=2, anchor=center, inner sep=0}, "t"', from=4-1, to=5-2]
		\arrow["{T_\alpha}"', from=5-2, to=6-3]
		\arrow["{T_{As}}"', from=6-3, to=6-4]
		\arrow["{T_t}"', from=6-4, to=5-6]
		\arrow["{T^2_{\alpha}}", from=1-4, to=2-5]
		\arrow[""{name=3, anchor=center, inner sep=0}, "\mu", from=2-5, to=4-6]
		\arrow[""{name=4, anchor=center, inner sep=0}, "t"', from=5-4, to=6-4]
		\arrow["A\mu"', from=5-4, to=4-5]
		\arrow["t", from=4-5, to=4-6]
		\arrow["\alpha", from=4-1, to=4-2]
		\arrow[""{name=5, anchor=center, inner sep=0}, "As", from=4-2, to=3-3]
		\arrow["{T_{\alpha}}"', from=1-3, to=2-3]
		\arrow[""{name=6, anchor=center, inner sep=0}, "t", from=3-3, to=2-3]
		\arrow["{T_{At}}"', from=2-3, to=2-4]
		\arrow[""{name=7, anchor=center, inner sep=0}, "{T_t}"', from=2-4, to=2-5]
		\arrow[""{name=8, anchor=center, inner sep=0}, "{AT_{t}}"', from=3-3, to=3-4]
		\arrow[""{name=9, anchor=center, inner sep=0}, "t"', from=3-4, to=2-4]
		\arrow["A\mu", from=3-4, to=4-5]
		\arrow["At"', from=4-2, to=5-3]
		\arrow[""{name=10, anchor=center, inner sep=0}, "{A{T_s}}"', from=5-3, to=5-4]
		\arrow[""{name=11, anchor=center, inner sep=0}, "t"', from=5-3, to=6-3]
		\arrow["{T_{\y}}"{description}, draw=none, from=0, to=7]
		\arrow["\iso"{marking, allow upside down}, draw=none, from=6, to=9]
		\arrow["\w"{description}, draw=none, from=3-4, to=3]
		\arrow["\w"{description}, draw=none, from=5-4, to=1]
		\arrow["A\commute"{description}, draw=none, from=8, to=10]
		\arrow["\iso"{description}, draw=none, from=11, to=4]
		\arrow["\y"{description}, draw=none, from=2, to=5-3]
		\arrow["\bistrong"{description}, draw=none, from=2-2, to=5]
	\end{tikzcd}
\]

\vspace{-1mm}

\[
	(2) \hspace{3mm}
\begin{tikzcd}[scalenodes = \axiomscalenodes, column sep = 1em, row sep = 2em,
		execute at end picture={
			 			\foreach \nom in  {A4, A5, B1, B3, B3, B4, B6, C2, C3, C4, C5, D1, D3, D5, E3, E4}
			   				{\coordinate (\nom) at (\nom.center);}  												 	  				
			 			\fill[\strongycolour,opacity=\opacity] 
			 	  				(C2) -- (C5) -- (D5) -- (D3);						   				
			 			\fill[\strongycolour,opacity=\opacity] 
			 	  				(B1) -- (B3) -- (C4) -- (C2);											 	  								 	  				
			 			\fill[\commutecolour,opacity=\opacity] 
			 	  				(B3) -- (A4) -- (A5) -- (B6) -- (C5) -- (C4);										 	  									 	  				
				}				
	]
		&& {T_{A(BT_C)}} & \alias{A4} {T_{AT_{BC}}} & \alias{A5} {T^2_{A(BC)}} \\
		\alias{B1} {(T_AB)T_C} & {T_A(BT_C)} & \alias{B3} {T_AT_{BC}} &&& \alias{B6} {T_{A(BC)}} \\
		& \alias{C2} {T_{(T_AB)C}} && \alias{C4} {T_{T_A (BC)}} & \alias{C5} {T^2_{A(BC)}} & {T_{(AB)C}} \\
		&& \alias{D3} {T_{T_{AB}C}} && \alias{D5} {T^2_{(AB)C}}
		\arrow["\alpha", from=2-1, to=2-2]
		\arrow["{T_A t}"', from=2-2, to=2-3]
		\arrow[""{name=0, anchor=center, inner sep=0}, "s"{description}, from=2-3, to=1-4]
		\arrow[""{name=1, anchor=center, inner sep=0}, "{T_t}", from=1-4, to=1-5]
		\arrow["\mu", from=1-5, to=2-6]
		\arrow[""{name=2, anchor=center, inner sep=0}, "t"{description}, from=2-3, to=3-4]
		\arrow[""{name=3, anchor=center, inner sep=0}, "{T_s}"', from=3-4, to=3-5]
		\arrow[""{name=4, anchor=center, inner sep=0}, "\mu"{description}, from=3-5, to=2-6]
		\arrow[""{name=5, anchor=center, inner sep=0}, "t"', from=2-1, to=3-2]
		\arrow["{T_{\alpha}}"', from=3-2, to=3-4]
		\arrow["{T_{sC}}"', from=3-2, to=4-3]
		\arrow[""{name=6, anchor=center, inner sep=0}, "{T_s}"', from=4-3, to=4-5]
		\arrow["{T^2_{\alpha}}", from=4-5, to=3-5]
		\arrow[""{name=7, anchor=center, inner sep=0}, "\mu"', from=4-5, to=3-6]
		\arrow["{T_\alpha}"', from=3-6, to=2-6]
		\arrow[""{name=8, anchor=center, inner sep=0}, "s", from=2-2, to=1-3]
		\arrow["{T_{At}}", from=1-3, to=1-4]
		\arrow["\iso"{description}, draw=none, from=8, to=0]
		\arrow["\iso"{description}, draw=none, from=4, to=7]
		\arrow["\y"{description}, draw=none, from=5, to=2]
		\arrow["{T_{\y'}}"{description}, draw=none, from=3-4, to=6]
		\arrow["\commute"{description}, draw=none, from=1, to=3]
	\end{tikzcd}
=
\begin{tikzcd}[scalenodes = .9, column sep = 1em, row sep = 2.5em,
		execute at end picture={
			 			\foreach \nom in  {A2, A5, B1, B2, B5, C3, C5}
			   				{\coordinate (\nom) at (\nom.center);}  												 	  				
			 			\fill[\strongycolour,opacity=\opacity] 
			 	  				(B1) -- (A2) -- (A5) -- (B5);		
			 			\fill[\commutecolour,opacity=\opacity] 
			 	  				(B2) -- (B5) -- (C5) -- (C3);					 	  								   												 	  				
				}					
	]
		& \alias{A2} {T_A(BT_C)} & {T_{A(BT_C)}} & {T_{AT_{BC}}} & \alias{A5} {T^2_{A(BC)}} \\
		\alias{B1} {(T_AB)T_C} & \alias{B2} {T_{AB}T_C} & {T_{(AB)T_C}} & {T^2_{(AB)C}} & \alias{B5} {T^2_{A(BC)}} & {T_{A(BC)}} \\
		& {T_{(T_AB)C}} & \alias{C3} {T_{T_{AB}C}} & {T^2_{(AB)C}} & \alias{C5} {T_{(AB)C}}
		\arrow["\alpha", from=2-1, to=1-2]
		\arrow["{T_t}", from=1-4, to=1-5]
		\arrow["\mu", from=1-5, to=2-6]
		\arrow[""{name=0, anchor=center, inner sep=0}, "t"', from=2-1, to=3-2]
		\arrow["{T_{sC}}"', from=3-2, to=3-3]
		\arrow["{T_s}"', from=3-3, to=3-4]
		\arrow["\mu"', from=3-4, to=3-5]
		\arrow["{T_\alpha}"', from=3-5, to=2-6]
		\arrow["s", from=1-2, to=1-3]
		\arrow["{T_{At}}", from=1-3, to=1-4]
		\arrow["{sT_{C}}"', from=2-1, to=2-2]
		\arrow[""{name=1, anchor=center, inner sep=0}, "s"', from=2-2, to=2-3]
		\arrow["{T_{\alpha}}"', from=2-3, to=1-3]
		\arrow["{T_t}"', from=2-3, to=2-4]
		\arrow["{T^2_\alpha}"', from=2-4, to=2-5]
		\arrow["{T^2_{\alpha}}"{description}, from=2-5, to=1-5]
		\arrow[""{name=2, anchor=center, inner sep=0}, "\mu"{description}, from=2-5, to=3-5]
		\arrow["\iso"{description}, draw=none, from=2-5, to=2-6]
		\arrow["{T_{\y}}"{description}, draw=none, from=1-4, to=2-4]
		\arrow[""{name=3, anchor=center, inner sep=0}, "t"{description}, from=2-2, to=3-3]
		\arrow["{\y'}"{description}, draw=none, from=1-2, to=1]
		\arrow["\iso"{description}, draw=none, from=0, to=3]
		\arrow["\commute"{description}, draw=none, from=3, to=2]
	\end{tikzcd}
\]

\vspace{-1mm}

\[
	(3)
\begin{tikzcd}[scalenodes = \axiomscalenodes, column sep = 1em, row sep = 2em,
		execute at end picture={
			 			\foreach \nom in  {A1, A3, A4, A5, B1, B2, B3, B6, C1, C2, C3, C4, C5, D4, D5}
			   				{\coordinate (\nom) at (\nom.center);}  												 	  				
			 			\fill[\strongycolour,opacity=\opacity] 
			 	  				(B1) -- (A1) -- (A3) -- (B2);						   				
			 			\fill[\commutecolour,opacity=\opacity] 
			 	  				(B3) -- (A4) -- (A5) -- (B6) -- (C5) -- (C4);											 	  								 	  				
			 			\fill[\bistrongcolour,opacity=\opacity] 
			 	  				(B1) -- (B3) -- (C4) -- (C2);										 	  									 	  				
			 			\fill[\strongycolour,opacity=\opacity] 
			 	  				(C3) -- (C5) -- (D5) -- (D4);										 	  									 	  							 	  				
				}					
	]
		\alias{A1} {T_{AT_B}C} & {T_{(AT_B)C}} & \alias{A3} {T_{A(T_BC)}} & \alias{A4} {T_{AT_{BC}}} & \alias{A5} {T^2_{A(BC)}} \\
		\alias{B1}	{(T_AT_B)C} & \alias{B2} {T_A(T_BC)} & \alias{B3} {T_AT_{BC}} &&& \alias{B6} {T_{A(BC)}} \\
		\alias{C1} & \alias{C2} {T_{T_A B} C} & \alias{C3} {T_{(T_A B)C}} & \alias{C4} {T_{T_A(BC)}} & \alias{C5} {T^2_{A(BC)}} & \alias{C6} {T_{(AB)C}} \\
		&&  {T^2_{AB}C} & \alias{D4} {T_{T_{AB}C}} & \alias{D5} {T^2_{(AB)C}}
		\arrow[""{name=0, anchor=center, inner sep=0}, "\alpha"', from=2-1, to=2-2]
		\arrow["{T_A s}"', from=2-2, to=2-3]
		\arrow[""{name=1, anchor=center, inner sep=0}, "s"{description}, from=2-3, to=1-4]
		\arrow[""{name=2, anchor=center, inner sep=0}, "{T_t}", from=1-4, to=1-5]
		\arrow["\mu", from=1-5, to=2-6]
		\arrow[""{name=3, anchor=center, inner sep=0}, "{T_s}"', from=3-4, to=3-5]
		\arrow[""{name=4, anchor=center, inner sep=0}, "\mu"{description}, from=3-5, to=2-6]
		\arrow["t"', from=2-3, to=3-4]
		\arrow["tC"', from=2-1, to=3-2]
		\arrow["{T_\alpha}"', from=3-3, to=3-4]
		\arrow[""{name=5, anchor=center, inner sep=0}, "s"', from=3-2, to=3-3]
		\arrow["{T_{sC}}"{description}, from=3-3, to=4-4]
		\arrow[""{name=6, anchor=center, inner sep=0}, "{T_s}"', from=4-4, to=4-5]
		\arrow["{T^2_{\alpha}}"', from=4-5, to=3-5]
		\arrow[""{name=7, anchor=center, inner sep=0}, "\mu"', from=4-5, to=3-6]
		\arrow["{T_{\alpha}}"', from=3-6, to=2-6]
		\arrow[""{name=8, anchor=center, inner sep=0}, "s", from=2-2, to=1-3]
		\arrow["{T_{As}}", from=1-3, to=1-4]
		\arrow["\b"{description}, draw=none, from=2-2, to=3-3]
		\arrow["sC", from=2-1, to=1-1]
		\arrow["s", from=1-1, to=1-2]
		\arrow["{T_{\alpha}}", from=1-2, to=1-3]
		\arrow["{T_s C}"', from=3-2, to=4-3]
		\arrow[""{name=9, anchor=center, inner sep=0}, "s"', from=4-3, to=4-4]
		\arrow["{T_{\y'}}"{description}, draw=none, from=3-4, to=6]
		\arrow["\iso"{description}, draw=none, from=4, to=7]
		\arrow["\commute"{description}, draw=none, from=2, to=3]
		\arrow["\iso"{description}, draw=none, from=8, to=1]
		\arrow["{\y'}"{description}, draw=none, from=1-2, to=0]
		\arrow["\iso"{description}, draw=none, from=5, to=9]
	\end{tikzcd}
=
\begin{tikzcd}[scalenodes = \axiomscalenodes, column sep = 1em, row sep = 1.8em,
		execute at end picture={
			 			\foreach \nom in  {A2, A4, B5, C4, C3, C2, D6, E5, E3, D3, D6}
			   				{\coordinate (\nom) at (\nom.center);}  												 	  			
						\fill[\bistrongcolour,opacity=\opacity] 
			 	  				(A2) -- (A4) -- (B5) -- (C4) -- (C3);										 	  									 	  				
			 			\fill[\strongwcolour,opacity=\opacity] 
								(C2) -- (C4) -- (D6) -- (E5) -- (E3) -- (D3);										 	  									 	  							 	  				
				}			
	]
		&  \alias{A2}{T_{(AT_B)C}} & {T_{A(T_B C)}} & \alias{A4} {T_{A T_{BC}}} \\
		{T_{AT_B}C} &&&& \alias{B5} {T^2_{A(BC)}} \\
		{(T_AT_B)C} & \alias{C2} {T^2_{AB}C} & \alias{C3} {T_{T_{AB}C}} & \alias{C4}  {T^2_{(AB)C}} && {T_{A(BC)}} \\
		& {T_{T_A B} C} & \alias{D3} {T_{AB}C} &&& \alias{D6} {T_{(AB)C}} \\
		&& \alias{E3} {T^2_{AB}C} & {T_{T_{AB}C}} &  \alias{E5} {T^2_{(AB)C}}
		\arrow[""{name=0, anchor=center, inner sep=0}, "\mu", from=2-5, to=3-6]
		\arrow["tC"', from=3-1, to=4-2]
		\arrow["{T_s}"', from=5-4, to=5-5]
		\arrow[""{name=1, anchor=center, inner sep=0}, "\mu"', from=5-5, to=4-6]
		\arrow["{T_{\alpha}}"', from=4-6, to=3-6]
		\arrow["sC", from=3-1, to=2-1]
		\arrow[""{name=2, anchor=center, inner sep=0}, "s", from=2-1, to=1-2]
		\arrow["{T_s C}"', from=4-2, to=5-3]
		\arrow["s"', from=5-3, to=5-4]
		\arrow["{T_t C}", from=2-1, to=3-2]
		\arrow[""{name=3, anchor=center, inner sep=0}, "{\mu C}", from=3-2, to=4-3]
		\arrow[""{name=4, anchor=center, inner sep=0}, "{\mu C}"', from=5-3, to=4-3]
		\arrow["{\commute C}"{description}, draw=none, from=3-2, to=4-2]
		\arrow["s", from=4-3, to=4-6]
		\arrow[""{name=5, anchor=center, inner sep=0}, "s", from=3-2, to=3-3]
		\arrow["{T_s}", from=3-3, to=3-4]
		\arrow[""{name=6, anchor=center, inner sep=0}, "\mu", from=3-4, to=4-6]
		\arrow["{T^2_\alpha}"{description}, from=3-4, to=2-5]
		\arrow["{T_{tC}}"{description}, from=1-2, to=3-3]
		\arrow["{T_\alpha}", from=1-2, to=1-3]
		\arrow["{T_{As}}", from=1-3, to=1-4]
		\arrow["{T_t}", from=1-4, to=2-5]
		\arrow["{T_{\b}}"{description}, draw=none, from=1-3, to=3-4]
		\arrow["{\w'}"{description}, draw=none, from=4, to=1]
		\arrow["{\w'}"{description}, draw=none, from=3, to=6]
		\arrow["\iso"{description}, draw=none, from=2, to=5]
		\arrow["\iso"{description}, draw=none, from=0, to=6]
	\end{tikzcd}
\]

\[
	\hspace{5mm}
	(4) \hspace{1mm}
\begin{tikzcd}[scalenodes = \axiomscalenodes, column sep = 1.5em, row sep = 2em,
				execute at end picture={
			 			\foreach \nom in  {A2, A3, B1, B2, C2, C3}
			   				{\coordinate (\nom) at (\nom.center);}  												 	  				
			 			\fill[\strongzcolour,opacity=\opacity] 
			 	  				(B1) to[curve={height=-6pt}] (A2) -- (B2);											 	  								 	  				
			 			\fill[\commutecolour,opacity=\opacity] 
			 	  				(A2) -- (C2) -- (C3) -- (A3);				 	  				
				}
	]
		& \alias{A2} {T_{AT_B}} & \alias{A3} {T^2_{AB}} \\
		\alias{B1} {AT_B} & \alias{B2} {T_A T_B} & {T_{AB}} \\
		{T_{AB}} & \alias{C2} {T_{T_A B}} & \alias{C3} {T^2_{AB}}
		\arrow["{\eta T_B}"', from=2-1, to=2-2]
		\arrow["s", from=2-2, to=1-2]
		\arrow[""{name=0, anchor=center, inner sep=0}, "{T_t}", from=1-2, to=1-3]
		\arrow["\mu", from=1-3, to=2-3]
		\arrow[""{name=1, anchor=center, inner sep=0}, "t", from=2-2, to=3-2]
		\arrow[""{name=2, anchor=center, inner sep=0}, "{T_s}"', from=3-2, to=3-3]
		\arrow["\mu"', from=3-3, to=2-3]
		\arrow[""{name=3, anchor=center, inner sep=0}, "t"', from=2-1, to=3-1]
		\arrow["{T_{\eta B}}"', from=3-1, to=3-2]
		\arrow[""{name=4, anchor=center, inner sep=0}, "\eta", curve={height=-6pt}, from=2-1, to=1-2]
		\arrow["\iso"{description}, draw=none, from=3, to=1]
		\arrow["{\z'}"{description}, draw=none, from=4, to=2-2]
		\arrow["\commute"{description}, draw=none, from=0, to=2]
	\end{tikzcd}
=
\begin{tikzcd}[scalenodes = 1, column sep = 1.5em, row sep = 2em,
			execute at end picture={
		 			\foreach \nom in  {A1, A2, A3, B1, B2, B3, C1, C2, C3, D1, D2}
		   				{\coordinate (\nom) at (\nom.center);}  												 	  				
		 			\fill[\monadmcolour,opacity=\opacity] 
		 	  				(C1) -- (A2) -- (B3) -- (C3);											 	  								 	  				
		 			\fill[\strongzcolour,opacity=\opacity] 
		 	  				(C1) -- (C3) -- (D2);				 	  				
			}
	]
		{T_{AT_B}} & \alias{A2} {T^2_{AB}} \\
		{AT_B} &&\alias{B3}  {T_{AB}} \\
		\alias{C1} {T_{AB}} && \alias{C3} {T^2_{AB}} \\
		& \alias{D2} {T_{T_A B}}
		\arrow["{T_t}", from=1-1, to=1-2]
		\arrow[""{name=0, anchor=center, inner sep=0}, "\mu", from=1-2, to=2-3]
		\arrow["{T_s}"', from=4-2, to=3-3]
		\arrow["\mu"', from=3-3, to=2-3]
		\arrow["t"', from=2-1, to=3-1]
		\arrow["{T_{\eta B}}"', from=3-1, to=4-2]
		\arrow["\eta", from=2-1, to=1-1]
		\arrow[""{name=1, anchor=center, inner sep=0}, "\eta"{description}, from=3-1, to=1-2]
		\arrow[""{name=2, anchor=center, inner sep=0}, "\Id"{description, pos=0.7}, from=3-1, to=2-3]
		\arrow[""{name=3, anchor=center, inner sep=0}, "{T_{\eta}}"{description}, from=3-1, to=3-3]
		\arrow["{T_{\z'}}"{description}, draw=none, from=3, to=4-2]
		\arrow["\p"{description}, draw=none, from=2, to=3-3]
		\arrow["\n"{description}, draw=none, from=0, to=3-1]
		\arrow["\iso"{description}, draw=none, from=1-1, to=1]
	\end{tikzcd}
\hspace{10mm}
(5) \hspace{1mm}
\begin{tikzcd}[scalenodes = 1, column sep = 1.5em, 
			execute at end picture={
		 			\foreach \nom in  {A1, A2, A3, B1, B2, B3, C1, C2, C3, D1, D2}
		   				{\coordinate (\nom) at (\nom.center);}  												 	  				
		 			\fill[\commutecolour, opacity=\opacity] 
		 	  				(A2) -- (A3) -- (C3) -- (C2);											 	  								 	  				
		 			\fill[\strongzcolour,opacity=\opacity] 
		 	  				(B1) to[curve={height=12pt}] (C2) -- (B2);				 	  				
			}	
	]
		{T_{AB}} & \alias{A2} {T_{AT_B}} & \alias{A3} {T^2_{AB}} \\
		\alias{B1} {T_A B} & \alias{B2} {T_A T_B} & {T_{AB}} \\
		& \alias{C2} {T_{T_A B}} & \alias{C3}  {T^2_{AB}}
		\arrow[""{name=0, anchor=center, inner sep=0}, "{T_A \eta}"', from=2-1, to=2-2]
		\arrow["s", from=2-2, to=1-2]
		\arrow[""{name=1, anchor=center, inner sep=0}, "{T_t}", from=1-2, to=1-3]
		\arrow["\mu", from=1-3, to=2-3]
		\arrow["t", from=2-2, to=3-2]
		\arrow[""{name=2, anchor=center, inner sep=0}, "{T_s}"', from=3-2, to=3-3]
		\arrow["\mu"', from=3-3, to=2-3]
		\arrow["s", from=2-1, to=1-1]
		\arrow[""{name=3, anchor=center, inner sep=0}, "{T_{A \eta}}", from=1-1, to=1-2]
		\arrow[""{name=4, anchor=center, inner sep=0}, "\eta"', curve={height=12pt}, from=2-1, to=3-2]
		\arrow["\commute"{description}, draw=none, from=1, to=2]
		\arrow["\iso"{description}, draw=none, from=3, to=0]
		\arrow["\z"{description}, yshift=-1mm, draw=none, from=2-2, to=4]
	\end{tikzcd}
=
\begin{tikzcd}[scalenodes = 1, column sep = 1.5em,
			execute at end picture={
		 			\foreach \nom in  {A1, A2, A3, B1, B2, B3, C1, C2, C3, C4, D1, D2, D3}
		   				{\coordinate (\nom) at (\nom.center);}  												 	  				
		 			\fill[\strongzcolour,opacity=\opacity] 
		 	  				(B1) -- (A2) -- (B3);											 	  								 	  				
		 			\fill[\monadmcolour,opacity=\opacity] 
		 	  				(B1) -- (B3) -- (C4) -- (D3);				 	  				
			}		
		]
		& \alias{A2} {T_{AT_B}} \\
		\alias{B1} {T_{AB}} && \alias{B3} {T^2_{AB}} \\
		{T_A B} &&& \alias{C4} {T_{AB}} \\
		& \alias{D2} {T_{T_A B}} & \alias{D3} {T^2_{AB}}
		\arrow["{T_t}", from=1-2, to=2-3]
		\arrow["\mu", from=2-3, to=3-4]
		\arrow["{T_s}"', from=4-2, to=4-3]
		\arrow[""{name=0, anchor=center, inner sep=0}, "\mu"', from=4-3, to=3-4]
		\arrow["s", from=3-1, to=2-1]
		\arrow["{T_{A \eta}}", from=2-1, to=1-2]
		\arrow[""{name=1, anchor=center, inner sep=0}, "\eta"', from=3-1, to=4-2]
		\arrow[""{name=2, anchor=center, inner sep=0}, "{T_{\eta}}"{description}, from=2-1, to=2-3]
		\arrow[""{name=3, anchor=center, inner sep=0}, "\Id"{description, pos=0.7}, from=2-1, to=3-4]
		\arrow[""{name=4, anchor=center, inner sep=0}, "\eta"{description}, from=2-1, to=4-3]
		\arrow["{T_{\z}}"{description}, draw=none, from=1-2, to=2]
		\arrow["\p"{description}, draw=none, from=2-3, to=3]
		\arrow["\iso"{description}, draw=none, from=4, to=1]
		\arrow["\n"{description}, draw=none, from=2-1, to=0]
	\end{tikzcd}
\]

\vspace{-2mm}
\caption{Axioms (1)--(5) for a commutative pseudomonad.}
\label{fig:hyland-power-axioms-one}
\end{figure*}

\begin{figure*}
\begin{minipage}{\textwidth}
\[
	(6)
\begin{tikzcd}[row sep = 2em,
			execute at end picture={
				 			\foreach \nom in  {A,B,C, D, X, Y, Z, P, Q, L, M, N, R, S, T, E}
				   				{\coordinate (\nom) at (\nom.center);}
				 			\fill[\strongwcolour,opacity=\opacity] 
				 	  				(X) -- (Y) -- (Z) -- (P);			
				 			\fill[\strongwcolour,opacity=\opacity] 
				 	  				(A) -- (B) -- (C) -- (D);					
				 			\fill[\commutecolour,opacity=\opacity] 
				 	  				(L) -- (M) -- (B) -- (N);										 	  				
				 			\fill[\monadmcolour,opacity=\opacity] 
				 	  				(R) -- (S) -- (M) -- (T);					
				 			\fill[\monadmcolour,opacity=\opacity] 
				 	  				(B) -- (M) -- (E) -- (C);									 	  								 	  				
				      	   }
		]
		& \alias{Y} {T_{T_A T_B}} & \alias{Z} {T^2_{A T_B}} & \alias{R} {T^3_{AB}} & \alias{S} {T^2_{AB}} \\
		\alias{X} {T^2_A T_B} & \alias{L} {T_AT_{B}} & \alias{P} {T_{AT_{B}}} & \alias{T} {T^2_{AB}} & \alias{M} {T_{AB}} \\
		& \alias{A} {T_{T^2_A B}} & \alias{N} {T_{T_A B}} & \alias{B} {T^2_{AB}} & \alias{E} {T^2_{AB}} \\
		&& \alias{D} {T^2_{T_A B}} & \alias{C} {T^3_{AB}}
		\arrow[""{name=0, anchor=center, inner sep=0}, "{\mu T_B}", from=2-1, to=2-2]
		\arrow["s", from=2-2, to=2-3]
		\arrow[""{name=1, anchor=center, inner sep=0}, "{T_t}", from=2-3, to=2-4]
		\arrow[""{name=2, anchor=center, inner sep=0}, "\mu", from=2-4, to=2-5]
		\arrow[""{name=3, anchor=center, inner sep=0}, "{T_s}"', from=3-3, to=3-4]
		\arrow[""{name=4, anchor=center, inner sep=0}, "\mu"{description}, from=3-4, to=2-5]
		\arrow["t"{description}, from=2-2, to=3-3]
		\arrow["t"', from=2-1, to=3-2]
		\arrow[""{name=5, anchor=center, inner sep=0}, "{T_{\mu B}}"', from=3-2, to=3-3]
		\arrow["{T_s}"', from=3-2, to=4-3]
		\arrow[""{name=6, anchor=center, inner sep=0}, "{T^2_s}"', from=4-3, to=4-4]
		\arrow["{T_\mu}"', from=4-4, to=3-4]
		\arrow[""{name=7, anchor=center, inner sep=0}, "\mu"', from=4-4, to=3-5]
		\arrow["\mu"', from=3-5, to=2-5]
		\arrow[""{name=8, anchor=center, inner sep=0}, "s", from=2-1, to=1-2]
		\arrow["{T_s}", from=1-2, to=1-3]
		\arrow[""{name=9, anchor=center, inner sep=0}, "\mu"{description}, from=1-3, to=2-3]
		\arrow["{T^2_t}", from=1-3, to=1-4]
		\arrow[""{name=10, anchor=center, inner sep=0}, "\mu", from=1-4, to=2-4]
		\arrow[""{name=11, anchor=center, inner sep=0}, "{T_\mu}", from=1-4, to=1-5]
		\arrow["\mu", from=1-5, to=2-5]
		\arrow["\commute"{description}, draw=none, from=1, to=3]
		\arrow["\iso"{description}, draw=none, from=0, to=5]
		\arrow["{T_{\w'}}"{description}, draw=none, from=3-3, to=6]
		\arrow["\m"{description}, draw=none, from=4, to=7]
		\arrow["{\w'}"{description}, draw=none, from=8, to=9]
		\arrow["\iso"{description}, draw=none, from=9, to=10]
		\arrow["\m"{description}, draw=none, from=11, to=2]
	\end{tikzcd}
\hspace{2mm}
	=
	\hspace{2mm}
\begin{tikzcd}[scalenodes = 1, column sep = 3em,	row sep = 2em,		
		execute at end picture={
	 			\foreach \nom in  {A, A', B, B', B'', C, C', C'', D, D'}
	   				{\coordinate (\nom) at (\nom.center);}
	 			\fill[\commutecolour,opacity=\opacity] 
	 	  				(A) -- (C) -- (D) -- (B);			
	 			\fill[\commutecolour,opacity=\opacity] 
	 	  				(A) -- (A') -- (B'') -- (B) -- (A);										 	  				
	 			\fill[\monadmcolour,opacity=\opacity] 
	 	  				(B'') -- (C'') -- (D') -- (B');											 	  								 	  				
		}
		]
		\alias{A} {T_{T_A T_B}} & {T^2_{A T_B}} && \alias{A'} {T^3_{AB}} \\
		{T^2_A T_B} & \alias{B} {T^2_{T_A B}} & \alias{B'} {T^3_{AB}} && \alias{B''} {T^2_{AB}} \\
		\alias{C} {T_{T^2_A B}} & \alias{C'} {T_{T_A B}} &&& \alias{C''} {T_{AB}} \\
		& \alias{D} {T^2_{T_A B}} & {T^3_{AB}} & \alias{D'} {T^2_{AB}}
		\arrow["t"', from=2-1, to=3-1]
		\arrow["{T_s}"', from=3-1, to=4-2]
		\arrow[""{name=0, anchor=center, inner sep=0}, "{T^2_s}"', from=4-2, to=4-3]
		\arrow["\mu"', from=4-3, to=4-4]
		\arrow[""{name=1, anchor=center, inner sep=0}, "\mu"', from=4-4, to=3-5]
		\arrow["s", from=2-1, to=1-1]
		\arrow["{T_s}", from=1-1, to=1-2]
		\arrow["{T^2_t}", from=1-2, to=1-4]
		\arrow[""{name=2, anchor=center, inner sep=0}, "{T_\mu}", from=1-4, to=2-5]
		\arrow["\mu", from=2-5, to=3-5]
		\arrow[""{name=3, anchor=center, inner sep=0}, "{T_t}"{description}, from=1-1, to=2-2]
		\arrow["\mu"', from=4-2, to=3-2]
		\arrow[""{name=4, anchor=center, inner sep=0}, "{T_s}", from=3-2, to=4-4]
		\arrow[""{name=5, anchor=center, inner sep=0}, "\mu"', from=2-2, to=3-2]
		\arrow["{T^2_s}"', from=2-2, to=2-3]
		\arrow[""{name=6, anchor=center, inner sep=0}, "{T_\mu}"', from=2-3, to=2-5]
		\arrow[""{name=7, anchor=center, inner sep=0}, "\mu"{description}, from=2-3, to=4-4]
		\arrow["\commute"{description}, draw=none, from=2-1, to=3-2]
		\arrow["\iso"{description}, draw=none, from=5, to=7]
		\arrow["\iso"{description}, draw=none, from=4, to=0]
		\arrow["{T_{\commute}}"{description}, draw=none, from=3, to=2]
		\arrow["\m"{description}, draw=none, from=6, to=1]
	\end{tikzcd}
\]

\[
	(7)
\begin{tikzcd}[scalenodes = 1, row sep = 2em,
		execute at end picture={
	 			\foreach \nom in  {A, A', B, B', B'', B''', C, C', C'', C''',  D, D', D'', D''', X}
	   				{\coordinate (\nom) at (\nom.center);}
	 			\fill[\strongwcolour,opacity=\opacity] 
	 	  				(A) -- (A') -- (B'') -- (B);	
	 			\fill[\strongwcolour,opacity=\opacity] 
	 	  				(C) -- (C') -- (D') -- (D);				 	  						
	 			\fill[\monadmcolour,opacity=\opacity] 
	 	  				(A') -- (B''') -- (C''') -- (B'');		
	 			\fill[\monadmcolour,opacity=\opacity] 
	 	  				(C'') -- (C''') -- (D''') -- (D'');			 	  												 	  				
	 			\fill[\commutecolour,opacity=\opacity] 
	 	  				(B') -- (B'') -- (C''') -- (X);											 	  								 	  				
		}
		]
		& \alias{A} {T^2_{A T_B}} && \alias{A'} {T^3_{AB}} \\
		\alias{B} {T_{A T^2_B}} & \alias{B'} {T_{AT_{B}}} && \alias{B''} {T^2_{AB}} & \alias{B'''} {T^2_{AB}} \\
		\alias{C} {T_A T^2_B} & \alias{X} {T_AT_{B}} & \alias{C'} {T_{T_AB}} & \alias{C''} {T^2_{AB}} & \alias{C'''} {T_{AB}} \\
		& \alias{D} {T_{T_A T_B}} & \alias{D'} {T^2_{T_A B}} & \alias{D''} {T^3_{AB}} & \alias{D'''} {T^2_{AB}}
		\arrow["{T_A \mu}"', from=3-1, to=3-2]
		\arrow[""{name=0, anchor=center, inner sep=0}, "s", from=3-2, to=2-2]
		\arrow[""{name=1, anchor=center, inner sep=0}, "{T_t}", from=2-2, to=2-4]
		\arrow[""{name=2, anchor=center, inner sep=0}, "\mu", from=2-4, to=3-5]
		\arrow[""{name=3, anchor=center, inner sep=0}, "{T_s}"', from=3-3, to=3-4]
		\arrow[""{name=4, anchor=center, inner sep=0}, "\mu"', from=3-4, to=3-5]
		\arrow["t"', from=3-2, to=3-3]
		\arrow[""{name=5, anchor=center, inner sep=0}, "s", from=3-1, to=2-1]
		\arrow["{T_{A \mu}}", from=2-1, to=2-2]
		\arrow["t"', from=3-1, to=4-2]
		\arrow[""{name=6, anchor=center, inner sep=0}, "{T_t}"', from=4-2, to=4-3]
		\arrow[""{name=7, anchor=center, inner sep=0}, "\mu"', from=4-3, to=3-3]
		\arrow["{T^2_s}"', from=4-3, to=4-4]
		\arrow[""{name=8, anchor=center, inner sep=0}, "\mu"', from=4-4, to=3-4]
		\arrow[""{name=9, anchor=center, inner sep=0}, "{T_\mu}"', from=4-4, to=4-5]
		\arrow["\mu"', from=4-5, to=3-5]
		\arrow["{T_t}", from=2-1, to=1-2]
		\arrow[""{name=10, anchor=center, inner sep=0}, "{T^2_t}", from=1-2, to=1-4]
		\arrow["{T_\mu}", from=1-4, to=2-4]
		\arrow[""{name=11, anchor=center, inner sep=0}, "\mu", from=1-4, to=2-5]
		\arrow["\mu", from=2-5, to=3-5]
		\arrow["\commute"{description}, draw=none, from=1, to=3]
		\arrow["\iso"{description}, draw=none, from=5, to=0]
		\arrow["\w"{description}, draw=none, from=3-2, to=6]
		\arrow["\iso"{description}, draw=none, from=7, to=8]
		\arrow["\m"{description}, draw=none, from=4, to=9]
		\arrow["{T_{\w}}"{description}, draw=none, from=10, to=2-2]
		\arrow["\m"{description}, draw=none, from=11, to=2]
	\end{tikzcd}
\hspace{2mm}
	=
	\hspace{2mm}
\begin{tikzcd}[scalenodes = .9, column sep = 3em, row sep = 1.5em,
			execute at end picture={
		 			\foreach \nom in  {A1, B1, C2, C4, D1, D4}
		   				{\coordinate (\nom) at (\nom.center);}  												 	  				
		 			\fill[\commutecolour,opacity=\opacity] 
		 	  				(B1) -- (A1) -- (C2) -- (D1);											 	  								 	  				
		 			\fill[\commutecolour,opacity=\opacity] 
		 	  				(D1) -- (C2) -- (C4) -- (D4);				 	  				
			}
			]
		& \alias{A1} {T^2_{A T_B}} & {T^3_{AB}} \\
		\alias{B1} {T_{A T^2_B}} & {T_{A T_B}} && {T^2_{AB}} \\
		{T_A T^2_B} & \alias{C2} {T^2_{A T_B}} && \alias{C4} {T_{AB}} \\
		\alias{D1} {T_{T_A T_B}} & {T^2_{T_A B}} & {T^3_{AB}} & \alias{D4} {T^2_{AB}}
		\arrow["s", from=3-1, to=2-1]
		\arrow["t"', from=3-1, to=4-1]
		\arrow["{T_t}"', from=4-1, to=4-2]
		\arrow["{T^2_s}"', from=4-2, to=4-3]
		\arrow["{T_\mu}"', from=4-3, to=4-4]
		\arrow[""{name=0, anchor=center, inner sep=0}, "\mu"', from=4-4, to=3-4]
		\arrow[""{name=1, anchor=center, inner sep=0}, "{T_t}", from=2-1, to=1-2]
		\arrow[""{name=2, anchor=center, inner sep=0}, "{T^2_t}", from=1-2, to=1-3]
		\arrow["\mu", from=1-3, to=2-4]
		\arrow["\mu", from=2-4, to=3-4]
		\arrow["\mu", from=1-2, to=2-2]
		\arrow[""{name=3, anchor=center, inner sep=0}, "{T_s}"{description}, from=4-1, to=3-2]
		\arrow["\mu"', from=3-2, to=2-2]
		\arrow[""{name=4, anchor=center, inner sep=0}, "{T^2_t}"{description}, from=3-2, to=3-4]
		\arrow[""{name=5, anchor=center, inner sep=0}, "{T_t}"{description}, from=2-2, to=2-4]
		\arrow["\iso"{description}, draw=none, from=2, to=5]
		\arrow["\iso"{description}, draw=none, from=5, to=4]
		\arrow["{T_{\commute}}"{description}, draw=none, from=3, to=0]
		\arrow["\commute"{description}, draw=none, from=1, to=3]
	\end{tikzcd}
\]
\end{minipage}

\caption{Axioms (6) \& (7) for a commutative pseudomonad.}
\label{fig:hyland-power-axioms-two}
\vspace{4mm}
\hrule
\end{figure*}

 \section{Relating commutative, concurrent and monoidal pseudomonads} 
\label{sec:monoidal-iff-commutative}

 We outline the construction of a commutative pseudomonad from a monoidal pseudomonad (\Cref{sec:monoidal-to-commutative}), and the construction of a monoidal pseudomonad from a commutative pseudomonad (\Cref{sec:commutative-to-monoidal}). 
In doing so, we see how to construct a bistrong pseudomonad from a concurrent pseudomonad (\Cref{sec:monoidal-to-commutative}).
Here we just show how to construct the data. The axioms are all checked directly: this is long-winded, but relatively straightforward.

 \subsection{From monoidal to commutative}
 \label{sec:monoidal-to-commutative}

 Fix a monoidal pseudomonad as in \Cref{def:monoidalpseudomonad}, with the three modifications of the underlying monoidal pseudofunctor denoted by
 	$\mongamma, \mondelta$ 
 and $\monomega$ as in~\cite{Cheng2011monbicat}.

 We give the data for a commutative pseudomonad. In doing so we need to construct bistrong structure; because we only use the invertibility of $\muprod$ in the definition of $\commute$, this also shows how to  construct a \emph{lax} bistrong pseudomonad from a \emph{lax} monoidal pseudomonad. 
A short check then shows that the invertibility conditions of a concurrent pseudomonad 	
 	(\Cref{def:concurrent-pseudomonad})
 suffice to make all the modifications for the induced bistrong pseudomonad invertible.

 For the two strengths, we take:
\begin{align*}
 	t_{A, B} := AT_B \xra{\eta T_B} T_AT_B \xra{\chi} T_{AB} \\
 	s_{A, B} := T_A B \xra{T_A \eta} T_AT_B \xra{\chi} T_{AB}
 \end{align*}
 
 We give the structural modifications making the pseudomonad $(T, \mu, \eta)$ left strong; the ones for the right strength are very similar.
First, the unit laws:
\vspace{-2mm}
 \[
 	\hspace{4mm}
 	\x :=
\begin{tikzcd}[column sep = 1.5em, ampersand replacement=\&]
 				{IT_A} \\
 				{T_I T_A} \& \: \& {T_A} \\
 				{T_{IA}}
 				\arrow[""{name=0, anchor=center, inner sep=0}, "{\eta T_A}"', curve={height=12pt}, from=1-1, to=2-1]
 				\arrow[""{name=1, anchor=center, inner sep=0}, "{\iota T_A}", curve={height=-12pt}, from=1-1, to=2-1]
 				\arrow["\chi"', from=2-1, to=3-1]
 				\arrow[""{name=2, anchor=center, inner sep=0}, "{T_{\lambda}}"', from=3-1, to=2-3]
 				\arrow[""{name=3, anchor=center, inner sep=0}, "\lambda", curve={height=-6pt}, from=1-1, to=2-3]
 				\arrow["{\etaunit T_A}", shift right, shorten <=7pt, shorten >=7pt, Rightarrow, from=0, to=1]
 				\arrow["\mongamma"{description}, draw=none, from=2, to=3]
 			\end{tikzcd} 
\hspace{5mm}
\z := 
 			\begin{tikzcd}[ampersand replacement=\&, column sep = 1.5em]
 				AB \\
 				{AT_B} \& {T_A T_B} \& {T_{AB}}
 				\arrow["A\eta"', from=1-1, to=2-1]
 				\arrow["{\eta T_B}"', from=2-1, to=2-2]
 				\arrow["\chi"', from=2-2, to=2-3]
 				\arrow[""{name=0, anchor=center, inner sep=0}, "\eta", curve={height=-12pt}, from=1-1, to=2-3]
 				\arrow[""{name=1, anchor=center, inner sep=0}, "\eta\eta"{yshift=-1mm, xshift=-1mm}, from=1-1, to=2-2]
 				\arrow["\iso"{description}, draw=none, from=1, to=2-1]
 				\arrow["\etaprod"', shorten >=3pt, draw=none, yshift=1mm, from=2-2, to=0]
 			\end{tikzcd}
 \]

Next, we give the two associativity laws. Note that if $T$ is concurrent, $\w$ is invertible even though $\muprod$ is not.
\begin{align*}
 	\y &:= \begin{tikzcd}[ampersand replacement=\&, column sep = 4em]
 			{(AB)T_C } \\
 			{A(BT_C)} \& {(T_AT_B)T_C} \& {T_{AB}T_C} \\
 			{A (T_B T_C)} \& {T_A(T_B T_C)} \& {T_{(AB)C}} \\
 			{A T_{BC}} \& {T_A T_{BC}} \& {T_{A(BC)}}
 			\arrow["\alpha"', from=1-1, to=2-1]
 			\arrow["{A(\eta T_C)}"', from=2-1, to=3-1]
 			\arrow["A\chi"', from=3-1, to=4-1]
 			\arrow[""{name=0, anchor=center, inner sep=0}, "{\eta T_{BC}}"', from=4-1, to=4-2]
 			\arrow[""{name=1, anchor=center, inner sep=0}, "\chi"', from=4-2, to=4-3]
 			\arrow[""{name=2, anchor=center, inner sep=0}, "{\eta(T_B T_C)}"', from=3-1, to=3-2]
 			\arrow["{T_A \chi}", from=3-2, to=4-2]
 			\arrow["{(\eta \eta)T_C}"{description}, from=1-1, to=2-2]
 			\arrow["\iso"{description}, draw=none, from=2-1, to=2-2]
 			\arrow["{\eta(\eta T_C)}"{description}, from=2-1, to=3-2]
 			\arrow["\alpha", from=2-2, to=3-2]
 			\arrow[""{name=3, anchor=center, inner sep=0}, "{\chi T_C}", from=2-2, to=2-3]
 			\arrow["\chi", from=2-3, to=3-3]
 			\arrow["{T_{\alpha}}", from=3-3, to=4-3]
 			\arrow[""{name=4, anchor=center, inner sep=0}, "{\eta T_C}", curve={height=-18pt}, from=1-1, to=2-3]
 			\arrow["\iso"{description}, draw=none, from=2, to=0]
 			\arrow["\iso"{description, pos=0.4}, shift left, draw=none, from=2, to=2-1]
 			\arrow["\monomega"{description, pos=0.3}, yshift = -3mm, shift left, draw=none, from=3, to=1]
 			\arrow["{\etaprod T_C}", shorten >=3pt, draw=none, yshift=-2mm, from=2-2, to=4]
 		\end{tikzcd}
\\[5mm]
\w &:=
 		\hspace{-1mm}
\begin{tikzcd}[ampersand replacement=\&, column sep = 3em]
 			{AT^2_B} \& \: \& \: \& {A T_B} \\
 			\& \: \& {T_A T^2_B} \\
 			{T_A T^2_B} \& {T^2_A T^2_B} \&\& {T_A T_B} \\
 			{T_{A T_B}} \& {T_{T_A T_B}} \& {T^2_{AB}} \& {T_{AB}}
 			\arrow["{\eta T^2_B}"', from=1-1, to=3-1]
 			\arrow["\chi"', from=3-1, to=4-1]
 			\arrow["{T_{\eta T_B}}"', from=4-1, to=4-2]
 			\arrow["{T_{\chi}}"', from=4-2, to=4-3]
 			\arrow["\mu"', from=4-3, to=4-4]
 			\arrow["{T_{\eta}T^2_B}"', from=3-1, to=3-2]
 			\arrow["\chi", from=3-2, to=4-2]
 			\arrow["\chi", from=3-4, to=4-4]
 			\arrow[""{name=0, anchor=center, inner sep=0}, "\mu\mu", from=3-2, to=3-4]
 			\arrow["{\mu T^2_B}"{description}, from=3-2, to=2-3]
 			\arrow["{T_A \mu}", from=2-3, to=3-4]
 			\arrow[""{name=1, anchor=center, inner sep=0}, "\Id", curve={height=-12pt}, from=3-1, to=2-3]
 			\arrow[""{name=2, anchor=center, inner sep=0}, "{A \mu}", from=1-1, to=1-4]
 			\arrow["{\eta T_B}", from=1-4, to=3-4]
 			\arrow["\muprod\:", shorten >=3pt, Rightarrow, from=4-3, to=0]
 			\arrow["\iso"{description}, yshift=1mm, draw=none, from=2-3, to=0]
 			\arrow["{\p T^2_B}"{description}, draw=none, from=1, to=3-2]
 			\arrow["\iso"{description}, draw=none, from=2, to=2-3]
 		\end{tikzcd}
\end{align*}
\vspace{2mm}

For the bistrong structure:
\[
 	\bistrong 
 	:=
\begin{tikzcd}[scalenodes = .95, row sep = 2.5em]
 		{(AT_B)C} & {A(T_B C)} & {A(T_B T_C)} & {A T_{BC}} \\
 		{(T_A T_B)C} & {(T_A T_B)T_C} & {T_A (T_B T_C)} & {T_{A} T_{BC}} \\
 		{T_{AB}C} & {T_{AB}T_C} & {T_{(AB)C}} & {T_{A(BC)}}
 		\arrow["{T_{\alpha}}"', from=3-3, to=3-4]
 		\arrow["\chi"', from=3-2, to=3-3]
 		\arrow[""{name=0, anchor=center, inner sep=0}, "{T_{AB}\eta}"', from=3-1, to=3-2]
 		\arrow["{\chi C}"', from=2-1, to=3-1]
 		\arrow[""{name=1, anchor=center, inner sep=0}, "{(\eta T_B) C}"', from=1-1, to=2-1]
 		\arrow["\alpha", from=1-1, to=1-2]
 		\arrow["{A(T_B \eta)}", from=1-2, to=1-3]
 		\arrow[""{name=2, anchor=center, inner sep=0}, "A\chi", from=1-3, to=1-4]
 		\arrow["{\eta T_{BC}}", from=1-4, to=2-4]
 		\arrow["\chi", from=2-4, to=3-4]
 		\arrow[""{name=3, anchor=center, inner sep=0}, "{(\eta T_B) \eta}"{description}, from=1-1, to=2-2]
 		\arrow[""{name=4, anchor=center, inner sep=0}, "{(T_A T_B)\eta}"', from=2-1, to=2-2]
 		\arrow["{\chi T_C}", from=2-2, to=3-2]
 		\arrow["{\eta  (T_B T_C)}", from=1-3, to=2-3]
 		\arrow[""{name=5, anchor=center, inner sep=0}, "{T_A \chi}"', from=2-3, to=2-4]
 		\arrow[""{name=6, anchor=center, inner sep=0}, "{\eta (T_B \eta)}"{description}, from=1-2, to=2-3]
 		\arrow["\alpha"', from=2-2, to=2-3]
 		\arrow["\monomega"{description}, draw=none, from=2-3, to=3-3]
 		\arrow["\iso"{description}, draw=none, from=4, to=0]
 		\arrow["\iso", shift left, draw=none, from=2, to=5]
 		\arrow["\iso"{description, pos=0.2}, draw=none, from=1-3, to=6]
 		\arrow["\iso"{description}, shift left, draw=none, from=1, to=4]
 		\arrow["\iso"{description}, draw=none, from=3, to=6]
 	\end{tikzcd}
 \]

 Finally, for the commutative structure we take the following. Note that this is the only pasting diagram that uses both $\muprod$ and its inverse.
\[
 		\commute 
 		:=
\begin{tikzcd}[row sep = 2.3em, column sep = 3em]
 			{T_A T_B} & {T_A T^2_B} & {T_{A T_B}} & {T_{T_A T_B}} \\
 			&& {T^2_A T^2_B} & {T^2_{AB}} \\
 			{T^2_A T_B} & {T^2_A T^2_B} & {T_A T_B} & {T_{AB}} \\
 			{T_{T_A B}} & {T_{T_A T_B}} && {T^2_{AB}}
 			\arrow["{\eta T_B}"', from=1-1, to=3-1]
 			\arrow["\chi"', from=3-1, to=4-1]
 			\arrow[""{name=0, anchor=center, inner sep=0}, "{T_{T_A \eta}}"', from=4-1, to=4-2]
 			\arrow["{T_\chi}"', from=4-2, to=4-4]
 			\arrow["\mu"', from=4-4, to=3-4]
 			\arrow[""{name=1, anchor=center, inner sep=0}, "{T^2_A T_\eta}", from=3-1, to=3-2]
 			\arrow[""{name=2, anchor=center, inner sep=0}, "\chi"', from=3-2, to=4-2]
 			\arrow[""{name=3, anchor=center, inner sep=0}, "\chi"{description}, from=3-3, to=3-4]
 			\arrow["\mu\mu", from=3-2, to=3-3]
 			\arrow["{T_A \eta}", from=1-1, to=1-2]
 			\arrow["\chi", from=1-2, to=1-3]
 			\arrow["{T_{\eta T_B}}", from=1-3, to=1-4]
 			\arrow["{T_\chi}", from=1-4, to=2-4]
 			\arrow[""{name=4, anchor=center, inner sep=0}, "\mu", from=2-4, to=3-4]
 			\arrow[""{name=5, anchor=center, inner sep=0}, "{T_\eta T^2_B}"{description}, from=1-2, to=2-3]
 			\arrow["\chi"{description}, from=2-3, to=1-4]
 			\arrow["\iso"{description}, draw=none, from=1-3, to=2-3]
 			\arrow[""{name=6, anchor=center, inner sep=0}, "\Id"{description, pos=0.3}, from=1-1, to=3-3]
 			\arrow[""{name=7, anchor=center, inner sep=0}, "\mu\mu", from=2-3, to=3-3]
 			\arrow["\muprod"{description}, draw=none, from=7, to=4]
 			\arrow["{\p \n}"{description}, draw=none, from=5, to=6]
 			\arrow["{\n \p}"{description}, draw=none, from=3-1, to=6]
 			\arrow["\iso"{description}, draw=none, from=1, to=0]
 			\arrow["\muprod"{description}, yshift=-3mm, draw=none, from=3, to=2]
 		\end{tikzcd}
 \]

\newpage
 \subsection{From commutative to monoidal}
 \label{sec:commutative-to-monoidal}

 Fix a commutative pseudomonad as in \Cref{def:commutative-pseudomonad}, with right-strong data denoted as in \Cref{sec:right-strong-data}. We define a monoidal pseudomonad structure by taking $\iota$ to be the identity and
\[
 	\chi :=
 		\big( 
 			T_A T_B \xra{s} T_{A T_B} \xra{T_t} T^2_{AB} \xra{\mu} T_{AB}
 		\big)
 \]
We may therefore take $\etaunit$ to be the identity. The modifications $\monomega, \mongamma$ and $\mondelta$ making $T$ a monoidal pseudofunctor as in~\cite[Definition~2.5]{SchommerPries2009} are then defined as follows.
 
\[
	\mongamma :=
		\hspace{-1mm}
\begin{tikzcd}[column sep = 1em]
			{IT_A} && {T_I T_A} \\
			& {T_{IA}} & {T_{I T_A}} \\
			&& {T^2_{IA}} \\
			{T_A} && {T_{IA}}
			\arrow["{\eta T_A}", from=1-1, to=1-3]
			\arrow["s", from=1-3, to=2-3]
			\arrow[""{name=0, anchor=center, inner sep=0}, "{T_t}", from=2-3, to=3-3]
			\arrow[""{name=1, anchor=center, inner sep=0}, "\eta"{description}, from=1-1, to=2-3]
			\arrow["t"', from=1-1, to=2-2]
			\arrow["\mu", from=3-3, to=4-3]
			\arrow["\eta"{description}, from=2-2, to=3-3]
			\arrow[""{name=2, anchor=center, inner sep=0}, "\Id"', curve={height=6pt}, from=2-2, to=4-3]
			\arrow["{T_\lambda}", from=4-3, to=4-1]
			\arrow["\lambda"{description}, from=1-1, to=4-1]
			\arrow["\x"{description}, draw=none, from=2-2, to=4-1]
			\arrow["{\z'}"{description}, draw=none, from=1-3, to=1]
			\arrow["\n"{description}, yshift = 1mm, draw=none, from=3-3, to=2]
			\arrow["\iso"{description}, shift right, draw=none, from=1, to=0]
		\end{tikzcd}
\hspace{5mm}
\mondelta :=
	\hspace{-1mm}
\begin{tikzcd}[column sep = 1em]
		{T_A I} && {T_A T_I} \\
		& {T_{AI}} & {T_{A T_I}} \\
		&& {T^2_{A I}} \\
		{T_A} && {T_{AI}}
		\arrow["{T_A \eta}", from=1-1, to=1-3]
		\arrow[""{name=0, anchor=center, inner sep=0}, "s", from=1-3, to=2-3]
		\arrow["{T_t}", from=2-3, to=3-3]
		\arrow["\mu", from=3-3, to=4-3]
		\arrow["{T_\rho}", from=4-3, to=4-1]
		\arrow["\rho"{description}, from=1-1, to=4-1]
		\arrow[""{name=1, anchor=center, inner sep=0}, "s"{description}, from=1-1, to=2-2]
		\arrow["{T_{A \eta}}", from=2-2, to=2-3]
		\arrow[""{name=2, anchor=center, inner sep=0}, "{T_\eta}"{description, pos=0.4}, curve={height=6pt}, from=2-2, to=3-3]
		\arrow[""{name=3, anchor=center, inner sep=0}, "\Id"{description}, curve={height=18pt}, from=2-2, to=4-3]
		\arrow["\iso"{description}, shift left=2, draw=none, from=1, to=0]
		\arrow["{T_{\z}}"{description, pos=0.4}, draw=none, from=2-3, to=2]
		\arrow["\p"{description, pos=0.4}, draw=none, from=3-3, to=3]
		\arrow["{\x'}"{description}, shift left=3, draw=none, from=1, to=4-1]
	\end{tikzcd}
\]
\vspace{-2mm}
\[
	\monomega :=
\begin{tikzcd}[scalenodes = .9, column sep = 1em]
		{(T_A T_B) T_C} & {T_{A T_B} T_C} && {T^2_{AB} T_C} && {T_{AB} T_C} \\
		{T_A (T_B T_C)} & {T_{(A T_B) T_C}} && {T_{T_{AB} T_C}} && {T_{(AB) T_C}} \\
		& {T_{A (T_B T_C)}} &&& {T^2_{(AB) T_C}} \\
		{T_A T_{B T_C}} & {T_{A T_{B T_C}}} && {T^2_{A (B T_C)}} && {T^2_{(AB)C}} \\
		{T_A T^2_{BC}} & {T_{A T^2_{BC}}} && {T^2_{A T_{BC}}} & {T^3_{(AB)C}} & {T_{(AB)C}} \\
		&&& {T^3_{A (BC)}} & {T^2_{(AB)C}} \\
		{T_A T_{BC}} & {T_{A T_{BC}}} && {T^2_{A(BC)} } && {T_{A(BC)}}
		\arrow[""{name=0, anchor=center, inner sep=0}, "{s T_C}", from=1-1, to=1-2]
		\arrow[""{name=1, anchor=center, inner sep=0}, "{T_t T_C}", from=1-2, to=1-4]
		\arrow["s", from=1-6, to=2-6]
		\arrow["{T_t}", from=2-6, to=4-6]
		\arrow["\mu", from=4-6, to=5-6]
		\arrow[""{name=2, anchor=center, inner sep=0}, "{T_\alpha}", from=5-6, to=7-6]
		\arrow["\alpha"', from=1-1, to=2-1]
		\arrow[""{name=3, anchor=center, inner sep=0}, "{T_A s}"', from=2-1, to=4-1]
		\arrow["{T_A t}"', from=4-1, to=5-1]
		\arrow["{T_A \mu}"', from=5-1, to=7-1]
		\arrow[""{name=4, anchor=center, inner sep=0}, "s"', from=7-1, to=7-2]
		\arrow[""{name=5, anchor=center, inner sep=0}, "{T_t}"', from=7-2, to=7-4]
		\arrow[""{name=6, anchor=center, inner sep=0}, "\mu"', from=7-4, to=7-6]
		\arrow[""{name=7, anchor=center, inner sep=0}, "{\mu T_C}", from=1-4, to=1-6]
		\arrow["s", from=1-2, to=2-2]
		\arrow["{T_{\alpha}}", from=2-2, to=3-2]
		\arrow[""{name=8, anchor=center, inner sep=0}, "s"{description}, from=2-1, to=3-2]
		\arrow[""{name=9, anchor=center, inner sep=0}, "{T_{t T_C}}", from=2-2, to=2-4]
		\arrow["s"', from=1-4, to=2-4]
		\arrow["{T_s}", from=2-4, to=3-5]
		\arrow[""{name=10, anchor=center, inner sep=0}, "\mu", from=3-5, to=2-6]
		\arrow["{T^2 t}"', from=3-5, to=5-5]
		\arrow[""{name=11, anchor=center, inner sep=0}, "\mu"{description}, from=5-5, to=4-6]
		\arrow["{T_{As}}", from=3-2, to=4-2]
		\arrow[""{name=12, anchor=center, inner sep=0}, "{T_t}", from=4-2, to=4-4]
		\arrow[""{name=13, anchor=center, inner sep=0}, "{T^2_{\alpha}}"', from=3-5, to=4-4]
		\arrow["{T_{\mu}}"', from=5-5, to=6-5]
		\arrow[""{name=14, anchor=center, inner sep=0}, "\mu"{description}, from=6-5, to=5-6]
		\arrow[""{name=15, anchor=center, inner sep=0}, "{T^2_\alpha}", from=6-5, to=7-4]
		\arrow[""{name=16, anchor=center, inner sep=0}, "{T^2_{A t}}"', from=4-4, to=5-4]
		\arrow[""{name=17, anchor=center, inner sep=0}, "{T^3_\alpha}"', from=5-5, to=6-4]
		\arrow["{T^2_{t}}"', from=5-4, to=6-4]
		\arrow[""{name=18, anchor=center, inner sep=0}, "{T_{A T_{t}}}"', from=4-2, to=5-2]
		\arrow[""{name=19, anchor=center, inner sep=0}, "{T_t}"', from=5-2, to=5-4]
		\arrow["{T_\mu}"', from=6-4, to=7-4]
		\arrow["{T_{A \mu}}", from=5-2, to=7-2]
		\arrow[""{name=20, anchor=center, inner sep=0}, "s"{description}, from=4-1, to=4-2]
		\arrow[""{name=21, anchor=center, inner sep=0}, "s"{description}, from=5-1, to=5-2]
		\arrow["{\y'}"{description}, draw=none, from=0, to=8]
		\arrow["{\w'}"{description}, draw=none, from=7, to=3-5]
		\arrow["\iso"{description}, draw=none, from=1, to=9]
		\arrow["{T_{\b}}"{description}, draw=none, from=9, to=12]
		\arrow["\iso"{description}, draw=none, from=10, to=11]
		\arrow["\m"{description}, draw=none, from=11, to=14]
		\arrow["\iso"{description}, draw=none, from=18, to=16]
		\arrow["\iso"{description}, draw=none, from=17, to=15]
		\arrow["\iso"{description}, draw=none, from=2, to=6]
		\arrow["\iso"{description}, draw=none, from=3-2, to=3]
		\arrow["\iso"{description}, draw=none, from=20, to=21]
		\arrow["\iso"{description}, draw=none, from=21, to=4]
		\arrow["{T_{\y}}"{description}, draw=none, from=13, to=17]
		\arrow["{T_{\w}}"{description}, shift right=2, draw=none, from=19, to=5]
	\end{tikzcd}
\]

Finally, for the modifications $\etaprod$, $\muunit$ and $\muprod$ we take:
\[
	\etaprod :=
		\hspace{-3mm}
\begin{tikzcd}[column sep = 1.2em]
			AB & {A T_B} & {T_A T_B} \\
			&& {T_{A T_B} } \\
			& {T_{AB}} & {T^2_{AB}} \\
			&& {T_{AB}}
			\arrow["{A \eta}", from=1-1, to=1-2]
			\arrow["{\eta T_B}", from=1-2, to=1-3]
			\arrow["s", from=1-3, to=2-3]
			\arrow["{T_t}", from=2-3, to=3-3]
			\arrow["\mu", from=3-3, to=4-3]
			\arrow[""{name=0, anchor=center, inner sep=0}, "\eta"', curve={height=6pt}, from=1-2, to=2-3]
			\arrow["t"', from=1-2, to=3-2]
			\arrow[""{name=1, anchor=center, inner sep=0}, "\eta", from=3-2, to=3-3]
			\arrow[""{name=2, anchor=center, inner sep=0}, "\Id"', curve={height=6pt}, from=3-2, to=4-3]
			\arrow[""{name=3, anchor=center, inner sep=0}, "\eta"', curve={height=6pt}, from=1-1, to=3-2]
			\arrow["{\z'}"{description}, draw=none, from=1-3, to=0]
			\arrow["\n"{description}, draw=none, from=3-3, to=2]
			\arrow["\iso"{description}, draw=none, from=0, to=1]
			\arrow["\z"{description}, draw=none, from=1-2, to=3]
		\end{tikzcd}
\hspace{7mm}
\muunit :=
	\hspace{-1mm}
\begin{tikzcd}[column sep = 1.2em]
		{I } & {T_I} & {T^2_I} \\
		&& {T_I}
		\arrow["{T_\eta}", from=1-2, to=1-3]
		\arrow["\mu", from=1-3, to=2-3]
		\arrow["\eta", from=1-1, to=1-2]
		\arrow[""{name=0, anchor=center, inner sep=0}, "\Id"', curve={height=6pt}, from=1-2, to=2-3]
		\arrow["\p"{description}, draw=none, from=1-3, to=0]
	\end{tikzcd}
\]

\[
	\muprod :=
\begin{tikzcd}[scalenodes = .9, column sep = 1em]
			& {T^2_A T^2_B} && {T_A T^2_B} && {T_A T_B} \\
			{T_{T_A T^2_B}} && {T^2_{A T^2_B}} & {T_{A T^2_B}} &&& {T_{A T_B}} \\
			&& {T^3_{A T_B}} & {T^2_{A T_B}} & {T^3_{AB}} && {T^2_{AB}} \\
			{T^2_{T_A T_B}} & {T^3_{A T_B}} & {T^2_{A T_B}} & {T_{A T_B}} & {T^2_{AB}} && {T_{AB}} \\
			& {T_{T_A T_B}} & {T^2_{A T_B}} && {T^3_{AB}} & {T^2_{AB}}
			\arrow[""{name=0, anchor=center, inner sep=0}, "{T_A\mu}", from=1-4, to=1-6]
			\arrow[""{name=1, anchor=center, inner sep=0}, "{\mu T^2_B}", from=1-2, to=1-4]
			\arrow["s", from=1-4, to=2-4]
			\arrow["s", from=1-6, to=2-7]
			\arrow["{T_t}", from=2-7, to=3-7]
			\arrow["\mu", from=3-7, to=4-7]
			\arrow[""{name=2, anchor=center, inner sep=0}, "{T_{A \mu}}", from=2-4, to=2-7]
			\arrow["{T_t}", from=2-4, to=3-4]
			\arrow[""{name=3, anchor=center, inner sep=0}, "\mu"', from=5-6, to=4-7]
			\arrow["s"', from=1-2, to=2-1]
			\arrow[""{name=4, anchor=center, inner sep=0}, "Ts", from=2-1, to=2-3]
			\arrow[""{name=5, anchor=center, inner sep=0}, "\mu", from=2-3, to=2-4]
			\arrow["{T^2_t}"', from=2-3, to=3-3]
			\arrow[""{name=6, anchor=center, inner sep=0}, "\mu", from=3-3, to=3-4]
			\arrow["{T_\mu}"', from=5-5, to=5-6]
			\arrow["{T^2_t}"', from=5-3, to=5-5]
			\arrow["{T_s}"', from=5-2, to=5-3]
			\arrow["{T_t}"', from=2-1, to=4-1]
			\arrow[""{name=7, anchor=center, inner sep=0}, "\mu"', from=4-1, to=5-2]
			\arrow["{T^2_s}", from=4-1, to=4-2]
			\arrow[""{name=8, anchor=center, inner sep=0}, "\mu"', from=4-2, to=5-3]
			\arrow[""{name=9, anchor=center, inner sep=0}, "{T^2_t}", from=3-4, to=3-5]
			\arrow[""{name=10, anchor=center, inner sep=0}, "{T_\mu}", from=3-5, to=3-7]
			\arrow["{T_\mu}"', from=3-3, to=4-3]
			\arrow["{T_\mu}", from=4-2, to=4-3]
			\arrow[""{name=11, anchor=center, inner sep=0}, "\mu", from=4-3, to=4-4]
			\arrow["\mu"', from=3-4, to=4-4]
			\arrow[""{name=12, anchor=center, inner sep=0}, "{T_t}", from=4-4, to=4-5]
			\arrow["\mu", from=3-5, to=4-5]
			\arrow[""{name=13, anchor=center, inner sep=0}, "\mu", from=4-5, to=4-7]
			\arrow[""{name=14, anchor=center, inner sep=0}, "\mu"', from=5-3, to=4-4]
			\arrow[""{name=15, anchor=center, inner sep=0}, "\mu", from=5-5, to=4-5]
			\arrow["\m"{description}, draw=none, from=4-3, to=5-3]
			\arrow["\iso"{description}, draw=none, from=0, to=2]
			\arrow["\iso"{description}, draw=none, from=5, to=6]
			\arrow["\iso"{description}, draw=none, from=7, to=8]
			\arrow["{T_{\commute}}"{description}, draw=none, from=4, to=4-2]
			\arrow["{\w'}"{description}, draw=none, from=1, to=2-3]
			\arrow["\m"{description}, draw=none, from=6, to=11]
			\arrow["{T_{\w}}"{description}, draw=none, from=2, to=3-5]
			\arrow["\m"{description}, draw=none, from=10, to=13]
			\arrow["\iso"{description}, draw=none, from=9, to=12]
			\arrow["\iso"{description}, draw=none, from=14, to=15]
			\arrow["\m"{description}, draw=none, from=15, to=3]
		\end{tikzcd}
\]

\section{Sketches of proofs omitted from the main body}

\subsection{Proofs for \Cref{sec:strong-pseudofunctors}}

\redundancy*
\begin{proof}
  The first of these axioms is proved using that $\mathfrak{l}$ is completely determined by the other structural data of a monoidal bicategory (see~\cite[p.~64]{Gurski2013}). For the second axiom, while $\p$ is not  determined by the rest of the data, the composite $\mu \circ \p_T$  is uniquely expressible using $\n$ and $\m$. This suffices for
  the proof.
\end{proof}

\subsection{Proofs for \Cref{sec:examples} }
\label{sec:proofs-for-basic-examples-section}

\easyexamplesofstrengths*
\begin{proof}
For both claims, one constructs the data by following the corresponding 1-categorical argument and filling the commuting diagrams with the appropriate 2-cells; the equations hold by coherence. 
For~(\ref{c:pseudomonad-from-pseudomonoid}), for instance, 
the structural modifications $\x$ and $\y$ are given using $\mathfrak{l}$ and 
$\mathfrak{p}$, respectively, while $\w$ and $\z$ are given using 
$\mathfrak{r}$ and $\mathfrak{p}$, respectively, together with the 
pseudonaturality of $\psinv\alpha$.
The axioms hold by the coherence of pseudomonoids~\cite{Lack2000}.
Similarly for (\ref{c:pseudomonads-strong-wrt-plus}): the strength has components
	$[T\inl \circ \eta_A, T\inr] : A + TB \to T(A + B)$ 
and the structural modifications are given by taking the categorical proof and filling in the commuting diagrams with the appropriate 2-cells. 
The equations hold by coherence for bicategories with finite products~\cite{Power1989bilimit} and the fact all the structural 2-cells are invertible.
\end{proof}

\everymonadstrong*
\begin{proof}
Similarly to the categorical proof, for any pseudofunctor $F : \Cat \to \Cat$, and $a \in \catA$ one has
	$F(\lambda b \bind \seq{a, b}) : F(\catB) \to F(\catA \times \catB)$.
However, since $F$ is now a pseudofunctor we also have a natural transformation 
	$F(\lambda b \bind \seq{f, b})$ 
for each $f : a \to a'$ in $\catA$, with components
	$F(\lambda b \bind \seq{f, b})_w : 
			F(\lambda b \bind \seq{a, b}) 
			\to 
			F(\lambda b \bind \seq{a', b})$
in $F(\catA \times \catB)$.
We may therefore define a functor 
	$t_{\catA, \catB} : \catA \times F\catB \to F(\catA\times\catB)$
sending a pair of objects
	$(a, w)$
to 
	$F(\lambda b \bind \seq{a, b})(w)$
and a pair of morphisms
	$(a \xra{f} a', w \xra{g} w')$
to the composite 
	$F(\lambda b \bind \seq{a', b})(g) \circ F(\lambda b \bind \seq{f, b})_w$.
This is functorial because $F$ is functorial on natural transformations and
	$F(\lambda b \bind \seq{a, b})$
is a functor, and pseudonatural via the compositor for $F$.

Then $\x$ and $\y$ are defined using the compositor and unitor for $F$, and the coherence of pseudomonads ensures the axioms hold. 
Finally, if $T$ is a pseudomonad then one defines $\w$ and $\z$ using the pseudonaturality of $\eta$ and $\mu$: this is similar to the proof in the categorical setting, where one uses the naturality of the unit and multiplication to show the two compatibility laws hold. Again, the axioms follow from coherence.
\end{proof}

\subsection{Proofs for \Cref{sec:premonoidal-Kleisli-bicats}}

\freydstructurefromstrong*
\begin{proof}
We use \Cref{res:left-actions-and-left-strengths-equivalent}.
The binoidal structure is as in~(\ref{eq:left-bowties-defined}). Then,
  the proof consists in constructing a compatible pair of a left
  action and a right action, where ``compatible'' means that the
  two associators coincide on 1-cells, and the structural 2-cells
  $\actpentagonator$ and $\acttrianglem$ coincide. All of this is
  verified directly, based on the construction of the actions in Theorem~\ref{res:left-actions-and-left-strengths-equivalent}.

  The only remaining difficulty is the pseudonaturality of $\alpha$ in
  its middle argument, since this is not required for either of the
  actions. This is where the 2-cell given by the bistrong structure is used.
\end{proof}

\subsection{Proofs for \Cref{sec:actions}} 

\strengthtoleftaction*
 \begin{proof}
   The action on 2-cells is the same as that on morphisms. The
   compositor and unitor for $\act$ are given by the modifications
   $\w$ and $\z$ that come with the strength. We use $\x$ and $\y$ to
   construct the strength data $\actlambda$ and $\actalpha$ from the
   monoidal data $\lambda$ and $\alpha$, and finally we use $\z$
   again to lift $\pentagonator, \montrianglem, \montrianglel$ to 
   $\actpentagonator, \acttrianglem, \acttrianglel$.
   The strength axioms ensure that this forms an action.
 \end{proof}

\leftactionsandleftstrengthsequivalent*
\begin{proof}
The proof follows the categorical construction (see~\cite[Prop. 4.3]{McDermott2022}).
For every left strength $t$ for $T$, the induced action $\act :
\B \times \B_T \to \B_T$ extends the canonical action of $\B$ on
itself, by construction.
Conversely, any extension $(\act, \theta)$ induces a strength 
  	$t_{A,B} = \id_A \act \id_{TB}$, 
where  $\id_{TB}$ is regarded as an element of $\B_T(TB, B)$. 
These constructions are inverses, up to isomorphism, as we verify
directly. 
We then verify the axioms. In each direction, there is a tight match-up between the equations of the given structure and the equations for the required structure; 
for an outline, see \Cref{table:action-strength-correspondence}.
\end{proof}

\begin{table}[t]
\centering
\renewcommand{\arraystretch}{1.5}
\begin{tabular}{@{}cc@{}}
\toprule
\textsc{Strength} & \textsc{Action}                                                       
 \\ \midrule
\makecell{Axioms for a strong \\ pseudofunctor 
  (Fig. \ref{fig:coherence-axioms-for-strong-pseudofunctor})} & \makecell{Modification axioms
                                                               \\for the 2-cells
                                                               $\acttrianglem,
                                                               \acttrianglel,
                                                               \actpentagonator$ \\
                                                               determined
                                                               by
                                                               the action extension}\\[3mm]
\makecell{Compatibility between \\[-.3mm] 
$\mathbf{m, n, p}$ and $\mathbf{z, w}$  \\[1mm] 
}      
  & 
\makecell{Pseudofunctor axioms \\[-.3mm] for $\act : \B \times \B_T \to \B_T$}     \\[1mm]
  \makecell{Compatibility between \\[-.3mm]  $\mathbf{x}$ and $\mathbf{z, w}$ \\
}
&
                                                                \makecell{Pseudonaturality 
                                                                of  \\[-.3mm]  the transformation 
  \\[-.3mm] $\actlambda$ determined by the action extension } \\[4mm]
\makecell{Compatibility between \\[-.3mm] $\mathbf{y}$ and $\mathbf{z, w}$ 
}     
& \makecell{Pseudonaturality                                                                of \\[-.3mm] 
                 the transformation                                             \\
  $\actalpha$ determined by the action extension} \\[2mm]
  \bottomrule
\end{tabular}
 \caption{Relating the data and equations on each side of the correspondence in \Cref{res:left-actions-and-left-strengths-equivalent}.}
\label{table:action-strength-correspondence}
\vspace{-.5cm}
\end{table}

\vfill
\end{document}